\documentclass[onecolumn,ldraftcls, 11pt]{IEEEtran}

\usepackage[includefoot,left=15mm,right=15mm,top = 13mm, bottom=10mm]{geometry}
\usepackage{xcolor}
\usepackage{cancel}
\definecolor{plotblue}{HTML}{1F77B4}  
\definecolor{plotorange}{HTML}{FF7F0E}
\usepackage{pgfplots}
\pgfplotsset{compat=newest}
\usepackage{float}
\usepackage[utf8]{inputenc} 
\usepackage[T1]{fontenc}
\usepackage{url}
\usepackage{ifthen}
\usepackage{cite}
\usepackage{algorithm}
\usepackage{algpseudocode}

\usepackage[cmex10]{amsmath}
\usepackage{wrapfig}
\usepackage{tikz}
\usepackage{circuitikz}
\usepackage{multirow}
\usepackage{makecell}
\usetikzlibrary{shapes, arrows, positioning, patterns, fit, backgrounds}
\usepackage{cite}
\usepackage{setspace}
\usepackage[colorlinks=false]{hyperref}

\usepackage{mleftright}       
\mleftright                   
\usepackage{subcaption}
\usepackage{graphicx}         
\usepackage{booktabs}         
\usepackage{amsthm}
 \usepackage{graphicx}
\usepackage{amsmath}
\usepackage{graphics,graphicx,psfrag,color,float}
\usepackage{epsfig,bbold,bbm}
\usepackage{bm}
\usepackage{mathtools}
\usepackage{amsmath}
\usepackage{amssymb}
\usepackage{amsfonts}
\usepackage{amsfonts}
\usepackage {times}
\usepackage{bbm}
\usepackage{soul}
\usepackage{cleveref}
\usepackage{enumitem}
\usepackage[normalem]{ulem}
\usepackage{xcolor}
\usepackage{txfonts}
\makeatletter
\newcommand{\biggg}{\bBigg@{3}}
\newcommand{\Biggg}{\bBigg@{3.5}}
\def\namedlabel#1#2{\begingroup
    #2
    \def\@currentlabel{#2}
    \phantomsection\label{#1}\endgroup
}
\makeatother
\DeclareUnicodeCharacter{0325}{}

\def\QED{\mbox{\rule[0pt]{1.5ex}{1.5ex}}}

\newtheorem{theorem}{\bf{Theorem}}

\newtheorem{assumption}{Assumption}

\newcommand{\RNum}[1]{\uppercase\expandafter{\romannumeral #1\relax}}

\newcommand{\beq}{\begin{equation*}}
\newcommand{\enq}{\end{equation*}}
\newcommand{\bel}{\begin{lemma}}
\newcommand{\enl}{\end{lemma}}

\newcommand{\bet}{\begin{theorem}}
\newcommand{\ent}{\end{theorem}}

\newcommand{\tr}{\mathrm{Tr}}

\newcommand{\customfootnotetext}[2]{{
  \renewcommand{\thefootnote}{#1}
  \footnotetext[0]{#2}}}

\newcommand*{\cP}{\mathcal{P}}

\newcommand*{\cH}{\mathcal{H}}

\newcommand*{\cB}{\mathcal{B}}
\newcommand*{\cD}{\mathcal{D}}

\newcommand*{\cT}{\mathcal{T}}

\mathchardef\mhyphen="2D

\makeatletter
\newcommand*{\rom}[1]{\expandafter\@slowromancap\romannumeral #1@}
\makeatother
\mathchardef\mhyphen="2D
\newlist{steps}{enumerate}{1}
\setlist[steps, 1]{leftmargin = 1.1cm, label = Step \arabic*.}
\newtheorem{remark}{Remark}
\newtheorem{definition}{Definition}

\usepackage{makecell}

\newtheorem{lemma}{Lemma}
\newtheorem{corollary}{Corollary}
\newtheorem{proposition}{Proposition}

\usepackage[utf8]{inputenc} 
\usepackage[T1]{fontenc}
\usepackage{url}
\usepackage{ifthen}
\usepackage[cmex10]{amsmath}
\usepackage{amsfonts} 
\usepackage{times}
\usepackage{pdflscape}

\allowdisplaybreaks

\date{}

\usepackage{textcomp}
\DeclareUnicodeCharacter{2212}{\textminus}
\begin{document}
\title{Cram\'er--Rao Efficiency and Concentration Bounds for Quantum U-Statistics}
\author{
Ayanava Dasgupta\textsuperscript{$*$},  Naqueeb Ahmad Warsi\textsuperscript{$*$} and  Premanshu Chatterjee\textsuperscript{$*$}
}

\customfootnotetext{$*$}{
Indian Statistical Institute,
Kolkata 700108, India.
Email: 
{\sf 
 [ayanavadasgupta\_r, naqueebwarsi, premanshuchatterjee\_r]@isical.ac.in
}
}
\maketitle
\begin{abstract}
   We study the unbiased estimation of polynomial functionals of unknown quantum states from independent copies. We establish an equivalence between the first-order marginal of a permutation-invariant finite-copy estimator and the functional gradient. We prove that among unbiased permutation-invariant estimators of any polynomial functional, a quantum U-statistic is the unique extension of an estimator to an arbitrary number of copies. We further derive a universal variance expansion in which the leading $1/n$ term is determined by the variance of the functional gradient, while higher-order contributions are of order $O(1/n^2)$. This leading variance coincides with the multiparameter quantum Cram\'er--Rao limit, establishing asymptotic efficiency of quantum U-statistics. We also characterize the higher-order scaling at points where the variance of the first-order functional gradient vanishes. As an application, we analyze the Bures $\chi^2$-divergence and show that a spectral lower bound on the reference state is sufficient but not necessary for bounded-variance estimation. Beyond asymptotic variance, we establish variance-sensitive exponential concentration bounds, thereby deriving a closed-form matrix Bernstein-type inequality that yields a finite-sample complexity lower bound on $n$, while also capturing the finite-sample tail behaviour. We further establish the Moderate Deviation Principle under the assumption of a non-vanishing variance of the (first-order) functional gradient, which exactly characterizes the intermediate asymptotic regime in terms of the quantum Fisher information. 
\end{abstract}

\section{Introduction}

A central problem in quantum statistical inference and quantum information theory is determining the properties of an unknown quantum state, represented by a density matrix $\rho$ acting on a finite-dimensional Hilbert space $\mathcal{H}$ \cite{helstrom1976quantum, Holevo1982, paris2009quantum, hayashi2006quantum}. While full quantum state tomography provides a complete description of the state, it requires experimental resources that scale exponentially with the system size \cite{OW16, Haah_2017, flammia2011direct, blume2010optimal, gross2010quantum, smolin2012efficient, cramer2010efficient, Guta_2018}. In many practical scenarios, the objective is to estimate a specific scalar-valued functional of an unknown quantum state, and therefore complete state reconstruction may be unnecessary \cite{Ekert_2002, huang2020predicting, 10.1145/3188745.3188802}. Such functionals naturally encompass a broad class of physically and information-theoretically significant quantities, including purities, generic polynomial evaluations \cite{Rath_2021, ZT25, Wang_2026, bovino2005direct}, and various quantum divergences \cite{Keyl_2001, BOW19}.

Mathematically, any scalar-valued \textit{polynomial trace functional}\footnote{Sticking to the standard literary convention as followed by \cite{BOW19}, we simply refer to such polynomial trace functionals as \textit{polynomial functionals} henceforth.} $f(\rho)$ can be expressed as a linear combination of trace expressions in which the unknown density matrix {$\rho$} is interleaved with known, fixed coefficient matrices $A^{(i)}_j$, i.e.,
\begin{equation*}
    f(\rho) = c_0 + \sum_{i=1}^{k} \mathrm{Tr}\left[ \prod_{j=1}^{i} \big(A^{(i)}_j \rho\big) \right]\,.
\end{equation*}

This representation provides a convenient operator form for studying polynomial functionals and their estimation from multiple copies of an unknown quantum state. In fact, several important quantum information-theoretic measures can be reduced to this polynomial structure for specific choices of the known coefficients; for example, purity, (squared) Hilbert-Schmidt distance, maximal/sandwiched/Bures $\chi^2$-divergence, the quantity $B_n^{\mathrm{(Kry)}}$ defined in Krylov Shadow Tomography \cite{ZT25} used to estimate QFI, (exponential) Rényi entropy and divergences of integer order, and many more. Our primary objective is to study the global estimation of such functionals given $n$ independent and identically distributed copies ($\rho^{\otimes n}$) of an unknown quantum state.

From a physical measurement perspective, estimating a degree-$k$ polynomial functional $f(\rho)$ naturally begins with identifying a self-adjoint observable acting on $k$ copies of the state. If the expectation value of this observable equals the target functional, we refer to it as a \emph{kernel} for the functional $f(\rho)$. When $n$ $(>k)$ copies are available, however, restricting the measurement to a single $k$-copy kernel does not fully exploit the available data. One can instead extend the kernel to a \emph{global observable} acting on all the $n$ copies, thereby enabling the collective use of the available samples \cite{massar1995optimal}. In this paper, we combine the framework of quantum U-statistics with multiparameter Cramér--Rao theory to study the global estimation of polynomial functionals. We establish the geometric structure of the resulting multi-copy observables and characterize their fundamental asymptotic variance. For non-asymptotic finitely many samples, we present an exponential probability concentration bound that is variance-sensitive.

Our results are organized around four main contributions. First, we identify a direct geometric link between the partial trace of these permutation-invariant, finite-copy kernels (called the marginal kernel) and the gradient of the target functional. Second, we prove that the quantum U-statistic \cite{GB2010} is the unique permutation-invariant unbiased extension of a finite-copy kernel to an arbitrary number of copies. Third, using the marginal-gradient correspondence, we derive a universal expression for the leading-order variance and show that it coincides with the relevant quantum Cramér--Rao limit \cite{braunstein1994statistical, fujiwara2006strong}. In particular, the resulting U-statistics are asymptotically Cramér--Rao efficient without requiring state-dependent measurements, preliminary tomography, or adaptive procedures. Fourth, we apply our framework to the estimation of the Bures $\chi^2$-divergence and show that the assumption ($\lambda_{\min}(\sigma)\geq\delta>0$) on the reference state ($\sigma$) imposed in \cite{OW16} can be replaced by the substantially weaker requirement of bounded variance of the functional gradient. Finally, we also present a variance-sensitive exponential probability concentration bound to characterize the non-asymptotic finite-sample tail behaviour. We simplify the bound further to get a closed-form Bernstein-type probability concentration bound, which readily gives a lower bound on the finite-sample complexity for $n$, and also study the tail behaviour in the moderate deviation regime. We further show that under a non-degenerate estimation geometry -- i.e., a non-vanishing quantum Fisher information (QFI), the probability concentration convergence rate is exactly characterized by the QFI in the moderate deviation regime in the large-$n$ limit, thereby imparting this fundamental quantum information-theoretic quantity a new operational meaning.

\subsection{Gradients--Marginal Equivalence for Polynomial Functionals}

Firstly, we observe that when extending a local kernel to a global observable, it is sufficient to restrict attention to permutation-invariant estimators. Indeed, the quantum extension of the Rao--Blackwell theorem \cite{Holevo1982} implies that averaging an estimator over permutation symmetries cannot increase its variance. Thus, restricting our analysis to globally permutation-invariant multi-copy estimators entails no loss of statistical efficiency.

To analyze the variance of these global estimators, we first examine their behaviour on fewer subsystems. In particular, we consider the \textit{marginalization} of the local \emph{kernel}, obtained by tracing out the remaining subsystems with respect to the state $\rho$. This marginalization captures the action of the multi-copy kernel on a specified subset of subsystems and, in particular, allows us to identify its first-order marginal. Understanding this first-order marginal is central to relating the multi-copy kernel to the local geometry of the target functional.

A main contribution of our work is proving a direct geometric link between the first-order marginal kernel of a physical multi-copy observable (kernel) and the mathematical gradient of the target functional. We show that when a mixed-degree polynomial functional is mapped to a permutation-invariant $k$-copy kernel $O^{\text{sym}}_k$ using an identity-padding process, its scaled first-order marginal kernel $kO^{\text{sym}}_{k,1}$ perfectly matches the gradient $\nabla f(\rho)$ up to a scalar multiple of the identity matrix. Specifically, we prove that $$k O^{\text{sym}}_{k,1} = \nabla f(\rho) + C(\rho)\mathbb{I}.$$ Here, $k$ is the number of copies of $\rho$ the kernel $O^{\text{sym}}_k$ acts on, and the first-order marginal kernel $O^{\text{sym}}_{k,1}$ is computed by tracing out $(k-1)$ subsystems of $O^{\text{sym}}_k$, and $C(\rho)\mathbb{I}$ is an extra scalar shift that appears since we embed lower-degree polynomials into a higher-dimensional space.

Moreover, in quantum mechanics, any valid change to a density matrix must conserve total probability, meaning the physical perturbation $X$ must be traceless ($\mathrm{Tr}[X] = 0$). Due to this, the extra identity shift disappears entirely when we calculate the inner product, resulting in 
$$\mathrm{Tr}[X (k O^{\text{sym}}_{k,1})] = \mathrm{Tr}[X \nabla f(\rho)].$$ 

This result directly connects the physical partial trace operation defining the marginal kernel to the mathematical functional derivative. It proves that building a permutation-invariant multi-copy observable naturally captures the local geometry of the target polynomial functional.

\subsection{Uniqueness of Quantum U-Statistics}

While a local $k$-copy kernel provides a natural estimation strategy, in practice we typically have access to $n>k$ independent copies of the unknown state. It is therefore natural to extend the $k$-copy kernel to an $n$-copy permutation-invariant observable, thereby making use of all available copies.

The standard construction for this extension is the quantum U-statistic, which averages the $k$-copy kernel uniformly over all $k$-element subsets of the $n$ copies \cite{GB2010}. This raises a fundamental question: is the quantum U-statistic merely one possible unbiased permutation-invariant extension of a local $k$-copy kernel, or is the extension uniquely determined?

We show that the latter is true. Specifically, we prove that the quantum U-statistic is the unique permutation-invariant unbiased $n$-copy extension of a given $k$-copy kernel. The key ingredient is a structural property of the permutation-invariant operator space: it is spanned by identical tensor powers of the form $A^{\otimes n}$. Consequently, a permutation-invariant Hermitian operator whose expectation vanishes for every density operator must itself be the zero operator.

To establish uniqueness, consider the difference between the quantum U-statistic and any other permutation-invariant unbiased $n$-copy estimator of the same functional. This difference is permutation-invariant and has zero expectation for every density operator. The preceding property therefore implies that the difference vanishes identically. Hence, once the local $k$-copy kernel and permutation invariance are fixed, the unbiased $n$-copy extension is uniquely determined and is precisely the quantum U-statistic.

\subsection{Universal Variance Limit and Asymptotic Cram\'er-Rao Efficiency} 
The uniqueness result established in this manuscript has an important implication for the statistical analysis of global estimators for scalar-valued polynomial functionals. Since any unbiased $n$-copy permutation-invariant observable must coincide with the corresponding quantum U-statistic, there is no freedom to choose a different unbiased permutation-invariant extension with a different variance. The statistical performance of all such global estimators is therefore completely determined by the variance structure of quantum U-statistics. Utilizing the variance analysis of quantum U-statistics via the Hoeffding decomposition studied in \cite{GB2010}, we establish that the ultimate asymptotic variance limit for estimating $f(\rho)$ with any unbiased permutation-invariant observable $O^{\text{sym}}_n$ is universally governed by the intrinsic variance of the gradient. Thus, 
$$\mathrm{Var}_{\rho^{\otimes n}}(O^{\text{sym}}_n) = \frac{1}{n} \mathrm{Var}_{\rho}\big(\nabla f(\rho)\big) + \mathcal{O}\left(\frac{1}{n^2}\right).$$

To evaluate how well our global permutation-invariant estimators perform, we connect this universal variance limit to the fundamental bounds of quantum statistical inference. In multiparameter quantum estimation theory, the ultimate limit on precision is given by the Quantum Cramér-Rao Bound (QCRB). This limit is determined by the Symmetric Logarithmic Derivative (SLD) operators and the quantum Fisher information (QFI) matrix \cite{Helstrom1967, Helstrom1968, Boixo_2007, Suzuki_2020, Hayashi_2025}.

A central result of this manuscript proves that any unbiased $n$-copy permutation-invariant observable built for a polynomial functional is asymptotically Cramér-Rao efficient. To achieve this, we adapt the local multi-parameter estimation framework with nuisance parameters developed by Suzuki et al. \cite{Suzuki_2020}. While their comprehensive framework accommodates various definitions of QFI based on different types of logarithmic derivatives, we specifically adapt the formulation that uses the Symmetric Logarithmic Derivative (SLD), also utilised in \cite{Hayashi_2025}. Under this framework, to bound the variance of a target functional tightly without penalties from other parameter fluctuations, the parameter space must meet some orthogonality conditions. Under this orthogonal setup, the variance bound for estimating the functional simplifies exactly to the inverse variance of the primary, idealized single-copy SLD operator.

We prove that for any scalar-valued polynomial functional, this primary SLD operator is directly proportional to the gradient of the functional. Since of this, the single-copy Cramér-Rao bound equals the state variance of the gradient operator, $V_{\mathrm{CR}} = \mathrm{Var}_{\rho}\big(\nabla f(\rho)\big)$. Since QFI is additive across independent subsystems \cite{hayashi2006quantum}, the asymptotic variance limit for any unbiased $n$-copy measurement strategy scales as $\frac{1}{n}\mathrm{Var}_{\rho}\big(\nabla f(\rho)\big)$.

This variance limit also highlights a practical limitation of standard local quantum estimation. Attaining the optimal local precision requires measuring in the eigenbasis of the primary SLD operator. However, this operator depends on the functional gradient $\nabla f(\rho)$ and therefore on the unknown state $\rho$ itself. The optimal measurement is thus state-dependent and cannot be specified without prior information about the state being estimated \cite{hayashi2006quantum}. In practice, overcoming this state dependence typically requires adaptive estimation protocols, in which preliminary measurements are used to obtain an estimate of $\rho$ and subsequently configure the measurement toward the corresponding optimal SLD basis.

Global $n$-copy permutation-invariant U-statistic observables bypass this difficulty. Since a U-statistic is constructed directly from the algebraic definition of the target functional, its construction is state-independent, and the resulting estimator is unbiased for every $\rho$. It therefore requires neither preliminary tomography nor adaptive calibration. The variance analysis of quantum U-statistics developed in \cite{GB2010} shows that the leading-order contribution to the variance is determined entirely by the first-order marginal of the underlying kernel, while contributions from higher-order marginals enter only at order $\mathcal{O}(1/n^2)$. The first-order marginal--gradient equivalence established above identifies this leading contribution with $\frac{1}{n}\mathrm{Var}_{\rho}\big(\nabla f(\rho)\big)$, which coincides with the asymptotic Cramér--Rao limit.

As $n$ tends to infinity, the higher-order contributions become negligible relative to the leading $1/n$ term. Consequently, the global permutation-invariant estimator attains the fundamental quantum variance limit at leading order. Together with the uniqueness result, this establishes that the quantum U-statistic is the unique unbiased permutation-invariant estimator and is asymptotically Cramér--Rao efficient for every polynomial functional.

\subsection{Relaxing Spectral Assumptions: The Bures $\chi^2$-Divergence}
We apply our results to the estimation of various fundamental quantum information-theoretic quantities. In particular, we analyze the statistical estimation of the Bures $\chi^2$-divergence \cite{Bures1969A,Uhlmann1976} between an unknown state $\rho$ and a fully known reference state $\sigma$ \cite{BOW19}. The Bures $\chi^2$-divergence is structurally defined as $$\chi^2_{\mathrm{B}}(\rho\|\sigma) = \mathrm{Tr}[\rho \Omega_\sigma(\rho)] - 1,$$ where $\Omega_\sigma$ represents an inverse symmetric logarithmic derivative super-operator.

Previous research in the literature, notably the seminal work by Bădescu, O'Donnell, and Wright \cite{BOW19}, proposed unbiased estimators for this specific divergence utilizing quantum U-statistics. However, to mathematically guarantee that the statistical variance of their proposed estimator remains bounded, the authors in \cite{BOW19} imposed a strict spectral restriction on the reference state. They required that the minimum eigenvalue of the reference state $\sigma$ be strictly bounded away from zero ($\lambda_{\min}(\sigma) \ge \delta > 0$).

In this paper, we show that while this strict lower-bound assumption on the spectrum of $\sigma$ is sufficient to guarantee bounded variance, it is not necessary. We argue that demanding boundedness of the intrinsic quantum variance of the functional gradient—expressed as $\mathrm{Var}_{\rho}\big(\nabla \chi^2_{\mathrm{B}}(\rho\|\sigma)\big)$—provides a much weaker, more natural, and physically meaningful condition for statistical analysis. 

To establish this claim, we derive a novel, exact continuous-time integral representation for the Bures $\chi^2$-divergence. By identifying the symmetric logarithmic derivative relation to the continuous-time Lyapunov equation and deploying its integral solution, we show that the divergence can be evaluated as $$\chi^2_{\mathrm{B}}(\rho \| \sigma) = 2 \int_0^\infty \mathrm{Tr}[(\rho e^{-\tau\sigma})^2] d\tau - 1.$$ 

Utilizing this continuous-time integral formulation, we explicitly construct parametric families of pairs of quantum states $\{(\rho_n,\sigma_n)\}_n$ where $\lambda_{\mathrm{min}}(\sigma_n)\to0$ asymptotically. For this pair of quantum states, we prove that both the divergence and its estimation variance remain uniformly bounded despite $\lambda_{\mathrm{min}}(\sigma_n)\to0$. 

Thus, we demonstrate that a spectral lower bound on the reference state can be replaced by a condition directly expressed in terms of the variance of the functional gradient, substantially enlarging the regime in which efficient estimation can be established.

\subsection{Variance-Sensitive, Finite-Sample Concentration and Convergence Rate}
The asymptotic Cram\'er--Rao efficiency shows the best possible performance of quantum U-statistics when the sample size is infinitely large. However, real quantum parameter estimation works with a limited number of samples. With a finite sample size, just knowing the expected variance is not enough. To guarantee reliable results, we also need high-probability tail bounds. Standard concentration inequalities usually rely only on the spectral norm of the observable, although recent work by De Palma and Pastorello \cite{De_Palma_2025} has explored quantum local norm metrics to improve this. Otherwise, these standard bounds are often loose and do not accurately capture the true statistical fluctuations of the quantum state. 

To fix this, we prove an exponential two-sided probability concentration bound. This new bound explicitly keeps the exact finite-sample variance of the quantum U-statistic in its leading quadratic term. We do this by using a moment-expansion approach designed for the specific combinatorial structure of permutation-invariant U-statistic kernels. We break the moment expansion across all possible overlapping subsets, and map these subsets onto intersection graphs followed by a spanning tree counting argument. This allows us to cleanly separate the exact second-order variance while conveniently bounding the complicated higher-order connected moments. 

Using this exact variance, we derive a closed-form, variance-sensitive Bernstein-type inequality. This clear formula separates the typical Gaussian fluctuations, which depend on the true variance, from the worst-case extreme errors driven by the spectral norm. It also avoids the need for any extra local norm metrics. Moreover, this Bernstein-type bound readily yields a lower bound on the finite sample complexity for $n$. Finally, we test these bounds along a dynamically scaling deviation sequence to establish the \textit{Moderate Deviation Principle (MDP)} for quantum U-statistics. We show that when errors shrink to zero slower than the standard $1/\sqrt{n}$ parametric rate, then the rate of convergence of the probability of deviation of the outcome of estimating a polynomial functional using a quantum U-statistic from the actual functional value is exactly governed by the QFI as $n$ tends to infinity. This result thereby bridges the theoretical gap between the small scale of the Central Limit Theorem and true Large Deviations, while also imparting the fundamental quantum information-theoretic quantity QFI a new operational meaning.

The rest of the paper is organized as follows. In Section \ref{sec:kernels}, we introduce finite-copy representations of monomial and generic polynomial functionals, formalize the corresponding permutation-invariant kernels, and establish the first-order marginal--gradient equivalence. Section \ref{sec:U_stat_unique} introduces quantum U-statistics and proves that, for a given polynomial functional, the quantum U-statistic is the unique unbiased permutation-invariant global extension of the corresponding local kernel. In Section \ref{sec:variance_optimality}, we combine the uniqueness of U-statistic observables with the local multiparameter estimation framework to characterize the quantum Cram\'er--Rao limit and prove asymptotic efficiency of the resulting estimators. We then apply the variance characterization to several quantum information-theoretic quantities, including the Bures $\chi^2$-divergence, state purity, and squared Hilbert--Schmidt distance. In Section \ref{sec:concentration}, we establish variance-sensitive exponential concentration bounds for quantum U-statistics, thereby also deriving a closed-form Bernstein-type inequality. We also establish a finite sample complexity lower bound on $n$, along with proving the MDP. Finally, the Appendices contain technical proofs and further discussion, including the connection between quantum marginal kernels and classical conditional expectations, and a detailed discussion on quantum cumulants and the normalized cumulant generating function (or normalized CGF, in short) for \textit{centred} quantum U-statistic observables, and how to bound them using a technique similar to the spanning-tree argument to bound connected absolute quantum moments also presented in this paper.

\section{Gradient--Marginal Equivalence for Polynomial Functionals}
\label{sec:kernels}

Consider the problem of estimating a scalar-valued polynomial functional $f(\rho)$ where the maximum polynomial degree to the input of the functional is $k$, and where $\rho\in\cD(\cH)$ is unknown. In quantum statistical estimation, these functionals can be written as a linear combination of interleaved monomials as
\begin{equation}
\label{eq:generic_poly}
    f(\rho) = c_0 + \sum_{i=1}^{k} \mathrm{Tr}\left[ \prod_{j=1}^{i} \big(A^{(i)}_j \rho\big) \right]\,,
\end{equation}
where $k$ is the maximum polynomial degree to the input of the functional, $c_0$ is a scalar constant, and each $A^{(i)}_j$ is a known coefficient matrix. In this structure, the state $\rho$ alternates with the known operators $A^{(i)}_j$ inside the trace. To motivate this  representation, we demonstrate how several fundamental quantum information-theoretic measures reduce to this exact algebraic structure for specific choices of the known coefficients:
\begin{itemize}
    \item \textbf{Purity:} The purity of any arbitrary quantum state $\rho\in\cD(\cH)$, defined as $\mathrm{Pur}(\rho) := \mathrm{Tr}[\rho^2]$, is a homogeneous functional where the polynomial degree to the input of the functional is $2$. It is recovered by setting the maximum polynomial degree to the input of the functional to $k=2$, the constant $c_0 = 0$, the first-order coefficient $A^{(1)}_1 = 0$, and the second-order coefficients to the identity matrix, $A^{(2)}_1 = A^{(2)}_2 = \mathbb{I}$.
    
    \item \textbf{$k^{\text{th}}$-Order Purity:} Generalizing the above, the $k^{\text{th}}$-order purity of any arbitrary quantum state $\rho\in\cD(\cH)$, given by $\mathrm{Tr}[\rho^k]$ (which also defines the (exponential) integer Rényi-$k$ entropy) is recovered by setting $c_0 = 0$, all coefficients corresponding to a lower polynomial degree to the input of the functional to zero (i.e., $A^{(i)}_j=0, \forall i\in\{1,2,\ldots,k-1\}, \forall j\in\{1,2,\ldots,i\}$), and $A^{(k)}_j = \mathbb{I}$ for all $j \in \{1, \dots, k\}$.
    
    \item \textbf{(Squared) Hilbert-Schmidt Distance:} Given a fully known reference state $\sigma$, the (squared) Hilbert-Schmidt distance between any arbitrary quantum state $\rho\in\cD(\cH)$ and $\sigma$ is defined as $$D_{\mathrm{HS}}^2(\rho, \sigma) := \mathrm{Tr}[(\rho-\sigma)^2] = \mathrm{Tr}[\sigma^2] - 2\mathrm{Tr}[\sigma\rho] + \mathrm{Tr}[\rho^2]\,.$$ This is a polynomial functional where the polynomial degree to the input of the functional is $2$, obtained by setting $c_0 = \mathrm{Tr}[\sigma^2]$, the first-order coefficient $A^{(1)}_1 = -2\sigma$, and the second-order coefficients $A^{(2)}_1 = A^{(2)}_2 = \mathbb{I}$.
    
    \item \textbf{Maximal $\chi^2$-Divergence:} For a full-rank reference state $\sigma$, the maximal $\chi^2$-divergence between any arbitrary quantum state $\rho\in\cD(\cH)$ and $\sigma$ is defined as $$\chi^2_{\mathrm{max}}(\rho\|\sigma) := \mathrm{Tr}[\rho^2 \sigma^{-1}] - 1\,.$$ This maps directly to Eq.~\eqref{eq:generic_poly} by setting the maximum polynomial degree to the input of the functional to $k=2$, the constant shift $c_0 = -1$, the first-order coefficient $A^{(1)}_1 = 0$, and the second-order coefficients as $A^{(2)}_1 = \sigma^{-1}$ and $A^{(2)}_2 = \mathbb{I}$ (since $\mathrm{Tr}[\sigma^{-1}\rho\mathbb{I}\rho] = \mathrm{Tr}[\rho^2\sigma^{-1}]$).

    \item \textbf{Sandwiched $\chi^2$-Divergence:} Also requiring a full-rank reference state $\sigma$, the sandwiched $\chi^2$-divergence between any arbitrary quantum state $\rho\in\cD(\cH)$ and $\sigma$ is defined as $$\chi^2_{\mathrm{sand}}(\rho\|\sigma) = \mathrm{Tr}\big[ \rho \sigma^{-1/2} \rho \sigma^{-1/2} \big] - 1\,.$$ This perfectly mirrors our polynomial representation by simply setting the maximum polynomial degree to the input of the functional to $k=2$, the constant $c_0 = -1$, the first-order coefficient $A^{(1)}_1 = 0$, and both second-order coefficients to $A^{(2)}_1 = A^{(2)}_2 = \sigma^{-1/2}$.

    \item \textbf{Bures $\chi^2$-Divergence:} We show in Lemma \ref{lem:integral_measured_chi2} in Subsection \ref{subsec:app_bures_chi} that the Bures $\chi^2$-divergence between any arbitrary quantum state $\rho\in\cD(\cH)$ and a full-rank reference state $\sigma$ has an integral representation given by
    \begin{equation*}
        \chi^2_{\mathrm{B}}(\rho \| \sigma) = 2 \int_0^\infty \mathrm{Tr}[(\rho e^{-\tau\sigma})^2] d\tau - 1.
    \end{equation*}
    
    When $\sigma$ is a fully known, full-rank reference state, it ceases to be a statistical variable. By writing $\sigma$ in its spectral decomposition $\sigma = \sum_x \lambda_x |x\rangle\langle x|$ and substituting it into the trace, we can exactly evaluate the integral over the known deterministic parameters:
    \begin{align*}
        \int_0^\infty \mathrm{Tr}[\rho e^{-\tau\sigma} \rho e^{-\tau\sigma}] d\tau &= \sum_{x,y} \left( \int_0^\infty e^{-\tau(\lambda_x + \lambda_y)} d\tau \right) |\langle x|\rho|y\rangle|^2 \\
        &= \sum_{x,y} \frac{1}{\lambda_x + \lambda_y} \mathrm{Tr}\big[\rho |y\rangle\langle y| \rho |x\rangle\langle x|\big]\,.
    \end{align*}
    
    Since the eigenvalues $\lambda_x$ and eigenvectors $|x\rangle$ are fixed constants, this integral simplifies to a finite linear combination of quadratic terms in the unknown state $\rho$. Consequently, under the assumption of a known reference state, the divergence rigorously reduces to a polynomial functional, where the polynomial degree to the input of the functional is $2$, perfectly satisfying the generic polynomial form defined in Eq.~\eqref{eq:generic_poly}.
    
    \item \textbf{Krylov Shadow Tomography:} In \cite{ZT25}, the authors study the problem of estimating QFI by defining a quantity $B_n^{\mathrm{(Kry)}}$ which converges to QFI. This quantity is constructed using a sequence of polynomials $\{T_k\}_{k=0}^{n-1}$,  where $T_k$ as shown in \cite[Eq. (S41)]{ZT25}, for a fully-known Hermitian operator $H$ has the following form,
        \begin{equation}
            T_k = \frac{1}{2^k} \sum_{l=0}^{k+2} \mu_l^{(k)} \mathrm{Tr} (H \rho^l H \rho^{k-l+2})\,,\label{eq:Tk_general}
        \end{equation}
    where $\mu_l^{(k)} := \binom{k}{l} - 2\binom{k}{l-1} + \binom{k}{l-2}\,$.
    It is easy to see that the form of Eq.~\eqref{eq:Tk_general} is the same as that of Eq.~\eqref{eq:generic_poly}.
\end{itemize}

More complex measures, such as (exponential) Rényi divergences of integer order, can similarly be expressed as linear combinations of these interleaved monomial trace structures by expanding them in the eigenbasis of the reference state. 

Suppose we have access to \(n\) independent and identically distributed copies of \(\rho\), with \(n>k\). We first discuss how a homogeneous monomial can be represented as the expectation value of an observable acting jointly on multiple copies of the state. This representation will subsequently provide the basic building block for constructing an estimator from all available copies.

\subsection{Permutation-Invariant, Finite-Copy Local Kernels for Monomial Functionals}\label{subsec:homo_monomial}

We first consider the simpler case of a single, homogeneous monomial functional of degree $m$, i.e., $f_{\cB}(\rho):= \tr[\cB \rho^m],$ where $\cB$ is a known observable, and $m < n$.

To represent this degree-$m$ polynomial functional in terms of an $m$-copy kernel, we use the cyclic permutation operator $P_{\gamma_m}$ acting on $\mathcal{H}^{\otimes m}$. Let $\gamma_m = (1\; 2\; \dots\; m) \in S_m$ be the backward cyclic permutation such that the operator $P_{\gamma_m}$ acts on the product basis as
\begin{equation*}
   P_{\gamma_m} |x_1, x_2, \dots, x_m\rangle = |x_2, x_3, \dots, x_m, x_1\rangle\,.
\end{equation*} 
 For any $m$ operators $A_1, A_2, \dots, A_m$, the above operator $P_{\gamma_m}$ satisfies  the following trace identity
\begin{equation*}
    \mathrm{Tr}\big[ P_{\gamma_m} (A_1 \otimes A_2 \otimes \dots \otimes A_m) \big] = \mathrm{Tr}[A_1 A_2 \dots A_m]\,.
\end{equation*}
This holds since evaluating the tensor trace cyclically connects the subsystem inner products (e.g., $\langle x_{i-1} | A_i | x_i \rangle$), combining the independent operators into a single matrix multiplication within a global trace. The above trace identity generalizes the familiar SWAP ($\mathbb{S}$) identity, given by $\tr[\mathbb{S}(X\otimes Y)]=\tr[XY].$ Applying this trace identity, with
$
A_1=\mathcal B,
$ and $
A_2=A_3=\cdots=A_m=I,
$ and noting that $
(\mathcal B\otimes I^{\otimes(m-1)})\rho^{\otimes m}
=
\mathcal B\rho\otimes\rho\otimes\cdots\otimes\rho\,,
$
we obtain  
\begin{equation}
    f_{\cB}(\rho) = \tr\left[P_{\gamma_m} (\cB\otimes I^{\otimes (m-1)}) (\rho^{\otimes m})\right]\,. 
\end{equation}
Note that $P_{\gamma_m} (\cB\otimes I^{\otimes (m-1)})$ need not be self-adjoint and therefore does not, in general, correspond directly to an observable. To address this issue, we instead consider the self-adjoint counterpart of the operator $O_{\cB}$,
\begin{equation}
    O_{\cB}  = \frac{P_{\gamma_m} (\cB\otimes I^{\otimes (m-1)}) + (\cB\otimes I^{\otimes (m-1)})P_{\gamma_m}}{2}.\label{eq:mono_kernel_cons}
\end{equation}

Since $f_{\cB}(\rho)$ is real valued, $O_{\cB}$ satisfies $\tr[O_{\cB}\rho^{\otimes m}] = f_{\cB}(\rho).$ Hence, $O_\cB$ serves as an observable realization of the monomial $f_{\cB}$ acting on $m$ copies. 

There is, however, a useful freedom in the choice of this observable. Since the state $\rho^{\otimes m}$ is permutation-invariant, we may average $O_{\cB}$ over the permutation group $S_m$. Operationally, this corresponds to applying a \emph{twirling channel} over $S_m$ to define a permutation-invariant observable $O^{\mathrm{sym}}_{\cB}$, i.e.,
\begin{equation}
    O^{\mathrm{sym}}_{\cB} := \frac{1}{m!}\sum\limits_{\pi\in S_m}P_{\pi^{-1}}O_{\cB}P_{\pi}\,.\label{eq:sym_kernel}
\end{equation}
Since this twirling channel leaves the permutation-invariant state $\rho^{\otimes m}$ unaffected, this symmetrization preserves the expectation value, i.e.,
\begin{equation}
     \mathrm{Tr}[O^{\mathrm{sym}}_{\cB} \rho^{\otimes m}] =  f_{\cB}(\rho)\,. 
\end{equation}

The symmetrization in \eqref{eq:sym_kernel} employs the same fundamental variance-reduction principle established by the classical Rao--Blackwell theorem \cite{Rao1945, Blackwell1947}, which was later extended to the quantum domain by Holevo \cite{Holevo1982}. Namely, averaging an estimator over a symmetry that leaves the underlying state invariant preserves its expectation while reducing its variance. In the present setting, the averaging is performed over the permutation group \(S_m\). Since \(\rho^{\otimes m}\) is invariant under every permutation of its subsystems, permutation averaging leaves the expectation value unchanged. Moreover, by the operator convexity of \(X\mapsto X^2\), this averaging cannot increase the second moment, and hence cannot increase the intrinsic quantum variance. We formalize this observation in the following proposition.

\begin{proposition}[Variance Reduction via Symmetrization]\label{prop:rao_blackwell}
    Let $O_{\mathcal{B}}$ be an asymmetric $m$-copy kernel for a monomial functional $f(\rho) = \mathrm{Tr}[\mathcal{B} \rho^m]$, and let ${O^{\mathrm{sym}}_{\mathcal{B}}}$ be its fully permutation-invariant equivalent as defined in Eq.~\eqref{eq:sym_kernel}. For any product state $\rho^{\otimes m}$, the intrinsic quantum variance satisfies
    \begin{equation*}
        \mathrm{Var}_{\rho^{\otimes m}}({O^{\mathrm{sym}}_{\mathcal{B}}}) \le \mathrm{Var}_{\rho^{\otimes m}}(O_{\mathcal{B}})\,.
    \end{equation*}
\end{proposition}

\begin{proof}
    Since the state $\rho^{\otimes m}$ is a tensor product of identical states, it commutes with any permutation operator, i.e., $P_\pi \rho^{\otimes m} P_\pi^\dagger = \rho^{\otimes m}, \forall \pi\in S_m$. Consequently, we have 
    \begin{align*}
        \mathrm{Tr}[{O^{\mathrm{sym}}_{\mathcal{B}}}\rho^{\otimes m}]&= \frac{1}{m!} \sum_{\pi\in S_m} \mathrm{Tr}[P_\pi^\dagger O_{\mathcal{B}} P_\pi\rho^{\otimes m}]\\ 
        &= \frac{1}{m!} \sum_{\pi\in S_m} \mathrm{Tr}[O_{\mathcal{B}} P_\pi\rho^{\otimes m}P_\pi^\dagger]\\ 
        &= \mathrm{Tr}[O_{\mathcal{B}}\rho^{\otimes m}]\,.
    \end{align*}
    
    To compare the variances, it suffices to compare the second moments. Since the map $g(X): X \mapsto X^2$ is operator-convex, we apply the discrete operator Jensen's inequality (Kadison's inequality \cite{Kadison1952}) to the permutation-invariant observable to get
    \begin{align*}
        {\left(O^{\mathrm{sym}}_{\mathcal{B}}\right)}^2 &= \left( \frac{1}{m!} \sum_{\pi \in S_m} P_\pi^\dagger O_{\mathcal{B}} P_\pi \right)^2\\
        &\le \frac{1}{m!} \sum_{\pi \in S_m} \big(P_\pi^\dagger O_{\mathcal{B}} P_\pi\big)^2\,.
    \end{align*}
    Since $P_\pi$ is unitary ($P_\pi P_\pi^\dagger = \mathbb{I}$), we have $\big(P_\pi^\dagger O_{\mathcal{B}} P_\pi\big)^2 = P_\pi^\dagger O_{\mathcal{B}}^2 P_\pi\,.$
    Therefore, taking the expectation value with respect to $\rho^{\otimes m}$ yields
    \begin{align*}
        \mathrm{Tr}\big[ {\left(O^{\mathrm{sym}}_{\mathcal{B}}\right)}^2 \rho^{\otimes m} \big] &\le \frac{1}{m!} \sum_{\pi \in S_m} \mathrm{Tr}\big[ P_\pi^\dagger O_{\mathcal{B}}^2 P_\pi \rho^{\otimes m} \big] \\
        &= \frac{1}{m!} \sum_{\pi \in S_m} \mathrm{Tr}\big[ O_{\mathcal{B}}^2 P_\pi \rho^{\otimes m} P_\pi^\dagger \big] \\
        &= \frac{1}{m!} \sum_{\pi \in S_m} \mathrm{Tr}\big[ O_{\mathcal{B}}^2 \,\rho^{\otimes m} \big] \\
        &= \mathrm{Tr}\big[ O_{\mathcal{B}}^2 \,\rho^{\otimes m} \big]\,.
    \end{align*}
    Subtracting the identical squared expectations from both sides yields $\mathrm{Var}_{\rho^{\otimes m}}({O^{\mathrm{sym}}_{\mathcal{B}}}) \le \mathrm{Var}_{\rho^{\otimes m}}(O_{\mathcal{B}})$.
\end{proof}

\begin{remark}[Quantum Rao--Blackwell Theorem]
    The variance reduction demonstrated in Proposition \ref{prop:rao_blackwell} acts as a specialized, finite-sample realization of the quantum Rao--Blackwell theorem. While the foundational concepts of quantum sufficient statistics and Rao-Blackwell bounds were established in the seminal work of Holevo \cite{Holevo1982}, the theorem has been recently expanded into abstract operator algebras by Sinha \cite{Sinha2022} and formalized through Generalized Conditional Expectations (GCEs) by Tsang \cite{Tsang2023}. Rather than relying on this heavy $C^*$-algebraic machinery or spectral factorization, we provide a self-contained derivation. By specializing the symmetrization operation to a finite-group twirling channel (see Eq.~\eqref{eq:sym_kernel}), we obtain a direct and constructive proof via Kadison's inequality \cite{Kadison1952} that is immediate and tailored to the polynomial estimation framework developed in this manuscript.
\end{remark}

The preceding construction and discussion motivates the following definition.

\begin{definition}
Given a density operator \(\rho\in\mathcal D(\mathcal H)\), an observable \(O_m\) acting on \(\mathcal H^{\otimes m}\) is called an \(m\)-copy kernel of a real-valued polynomial functional \(f(\rho)\) if
$$
\tr[O_m\rho^{\otimes m}]
=
f(\rho).
$$
\end{definition}

Thus, the permutation-invariant observable \(O^{\text{sym}}_m\), constructed by applying the twirling map to \(O_m\) as demonstrated in Eq.~\eqref{eq:sym_kernel}, remains a valid \(m\)-copy kernel of \(f(\rho)\). Since symmetrization preserves the expectation value and does not increase the intrinsic quantum variance, we may, without loss of statistical efficiency, restrict our subsequent analysis to permutation-invariant kernels.

The notion of an \(m\)-copy kernel provides the connection between the algebraic functional \(f(\rho)\) and an observable that can be measured on \(m\) copies of the state. For the monomial considered above, \(O^{\text{sym}}_{\mathcal B}\) is such a kernel. For the monomial \(f_{\mathcal B}(\rho)=\tr[\mathcal B\rho^m]\), we now consider the first-order marginal of its permutation-invariant kernel \(O^{\mathrm{sym}}_{\cB}\). Define
\begin{equation}
\label{eq:first_marg}
O^{\mathrm{sym}}_{\mathcal B,1}
:=
\tr_{2,\ldots,m}
\left[
O^{\mathrm{sym}}_{\cB}
\left(
I\otimes\rho^{\otimes(m-1)}
\right)
\right].
\end{equation}
Since \(O^{\mathrm{sym}}_{\cB}\) is invariant under permutations of its \(m\) subsystems, and the state on each subsystem is the same, \(\rho^{\otimes m}\), the above reduction is identical for every choice of the retained subsystem. 
Thus, the first-order marginal kernel is a well-defined single-copy operator, independent of which subsystem is retained. The above technique of marginalization of a permutation-invariant kernel can be extended to any order $k \leq m$ as mentioned in the definition below,

\begin{definition}\label{def:marginal_kernel}
    Given an unknown state $\rho \in \mathcal{D}(\mathcal{H})$, and an $m$-copy permutation-invariant kernel $O^{\text{sym}}_m$ of a functional $f(\rho)$ in $\rho$, the $k^{\text{th}}$-order marginal kernel $O^{\text{sym}}_{m,k}$ is defined by tracing out (any of) the $(m-k)$ subsystems of the kernel's action, i.e., 
    \begin{equation}
        O^{\text{sym}}_{m,k}:=\mathrm{Tr}_{(k+1)..m}\left[O^{\text{sym}}_m\big(\mathbb{I}^{\otimes k}\otimes\rho^{\otimes(m-k)}\big)\right]\,.
    \end{equation}
\end{definition}

 The above definition is a quantum counterpart of the conditional expectation in the classical case. For completeness, we discuss this analogy in Appendix \ref{app:A}.

\subsection{Gradient Characterization of Monomial Functionals}

The first-order marginal of a permutation-invariant observable realizing a monomial functional (as given by Definition~\eqref{eq:first_marg}) has a direct connection to the sensitivity of the functional. We formalize this in Lemma \ref{lem:grad_marg} below.

\begin{lemma}[Gradient--Marginal Equivalence for Monomial Functionals]\label{lem:grad_marg}
    Given a monomial functional $f_{\mathcal{B}}(\rho) = \mathrm{Tr}[\mathcal{B}\rho^m]$, and its corresponding permutation-invariant $m$-copy kernel $O^{\mathrm{sym}}_{\mathcal{B}}$ such that $f_{\mathcal{B}}(\rho) = \mathrm{Tr}[O^{\mathrm{sym}}_{\mathcal{B}} \rho^{\otimes m}]$, the matrix gradient $\nabla f_{\mathcal{B}}(\rho)$ is directly proportional to the first-order marginal kernel of $O^{\mathrm{sym}}_{\mathcal{B}}$. Specifically, we have
    \begin{equation*}
        \nabla f_{\mathcal{B}}(\rho) = m O^{\mathrm{sym}}_{\mathcal{B},1}\,,
    \end{equation*}
    where the first-order marginal kernel is defined as $O^{\mathrm{sym}}_{\mathcal{B},1} := \mathrm{Tr}_{2\dots m}\left[O^{\mathrm{sym}}_{\mathcal{B}}(\mathbb{I} \otimes \rho^{\otimes (m-1)})\right]$.
\end{lemma}

\begin{proof}
    We prove this lemma by evaluating the derivative of $f_{\mathcal{B}}(\rho+t X)$ at $t=0$ using two different methods. First, we calculate it from the polynomial trace representation as $ \frac{d \mathrm{Tr}[\mathcal{B}(\rho+tX)^m]}{dt}\Big|_{t =0}$. Second, we calculate it using the tensor-product expectation value as $ \frac{d \mathrm{Tr}[O^{\mathrm{sym}}_{\mathcal{B}}(\rho+tX)^{\otimes m}]}{dt}\Big|_{t =0}$. Equating these two results links the matrix gradient to the partial trace operation defining the first-order marginal kernel.

    Recall from standard matrix differential calculus (see, e.g., \cite{magnus2019matrix, helstrom1976quantum}) that for any smooth parameterized family of matrices $\rho_\theta$, the Hermitian matrix gradient $\nabla f_{\mathcal{B}}(\rho_\theta)$ is defined via the generalized chain rule as
    \begin{equation}
    \label{eq:operator_derivative_chain_rule}
        \frac{\partial f_{\mathcal{B}}(\rho_\theta)}{\partial \theta^j} = \mathrm{Tr}\left[ \frac{\partial \rho_\theta}{\partial \theta^j} \nabla f_{\mathcal{B}}(\rho_\theta) \right].
    \end{equation}
    We evaluate this for a single-parameter affine perturbation $\rho_t := \rho + t X$, where the parameter is $t$ and $X$ is an arbitrary Hermitian matrix. Since the derivative of the state with respect to the parameter is $\frac{d \rho_t}{dt} = X$, substituting this into the general relation at $t=0$ yields the directional derivative 
    \begin{equation}
        \frac{d f_{\mathcal{B}}(\rho+t X)}{dt}\bigg|_{t=0} = \mathrm{Tr}[X \nabla f_{\mathcal{B}}(\rho)]\,. \label{eq:dir_derv}
    \end{equation}
    
    In the binomial expansion of the matrix polynomial $(\rho + tX)^m$, the term linear in $t$ takes the form $\sum_{k=0}^{m-1} \rho^{m-k-1} X \rho^k$. By isolating $X$ using the cyclic property of the trace, the directional derivative evaluates to
    \begin{equation}\label{eq:grho}
        \frac{d f_{\mathcal{B}}(\rho+t X)}{dt}\bigg|_{t=0} = \mathrm{Tr}\left[ \left(\sum_{k=0}^{m-1} \rho^{m-k-1} \mathcal{B} \rho^k\right) X \right]\,,
    \end{equation}
    which, using Eq.~\eqref{eq:dir_derv}, identifies the analytical matrix gradient as
    $$\nabla f_{\mathcal{B}}(\rho) = \sum_{k=0}^{m-1} \rho^{m-k-1} \mathcal{B} \rho^k\,.$$

    Next, we evaluate the derivative using the permutation-invariant tensor product representation. In the multilinear expansion of $(\rho+tX)^{\otimes m}$, the term linear in $t$ is composed of the sum over all configurations where exactly one subsystem contains $X$ and the remaining $(m-1)$ subsystems contain $\rho$,
    \begin{equation}\label{eq:tp}
        \sum_{i=0}^{m-1} \rho^{\otimes i} \otimes X \otimes \rho^{\otimes (m-i-1)}\,.
    \end{equation}

    To evaluate the global trace inner product of this sum with the kernel $O^{\mathrm{sym}}_{\mathcal{B}}$, we utilize the permutation-invariance of the kernel, i.e., $O^{\mathrm{sym}}_{\mathcal{B}} = P_{\pi}^\dagger O^{\mathrm{sym}}_{\mathcal{B}} P_{\pi}, \forall \pi \in S_m$. For each $i^{\text{th}}$ term in the summation, let $P_{\pi_i}$ be the specific permutation operator that swaps the first subsystem with the $(i+1)^{\text{th}}$ subsystem. Applying this permutation transforms the tensor product by shifting $X$ to the first position, $P_{\pi_i} (\rho^{\otimes i} \otimes X \otimes \rho^{\otimes (m-i-1)}) P_{\pi_i}^\dagger = X \otimes \rho^{\otimes (m-1)}$.

    By invoking the cyclic property of the trace and the kernel's permutation-invariance, we establish that each term in the summation yields an identical expectation value of
    \begin{align*}
        \mathrm{Tr}\Big[O^{\mathrm{sym}}_{\mathcal{B}} \big(\rho^{\otimes i} \otimes X \otimes \rho^{\otimes (m-i-1)}\big)\Big] &= \mathrm{Tr}\Big[\big(P_{\pi_i}^\dagger O^{\mathrm{sym}}_{\mathcal{B}} P_{\pi_i}\big) \big(\rho^{\otimes i} \otimes X \otimes \rho^{\otimes (m-i-1)}\big)\Big] \\
        &= \mathrm{Tr}\Big[O^{\mathrm{sym}}_{\mathcal{B}} \big( P_{\pi_i} (\rho^{\otimes i} \otimes X \otimes \rho^{\otimes (m-i-1)}) P_{\pi_i}^\dagger \big)\Big] \\
        &= \mathrm{Tr}\Big[O^{\mathrm{sym}}_{\mathcal{B}} \big(X \otimes \rho^{\otimes (m-1)}\big)\Big].
    \end{align*}
    
    Since this global trace inner product is identical for each of the $m$ distinct terms in Eq.~\eqref{eq:tp}, we can consolidate the sum into $m$ times the evaluation of the first term, i.e.,
    \begin{equation}\label{eq:additivity}
        \frac{d \mathrm{Tr}[O^{\mathrm{sym}}_{\mathcal{B}}(\rho+tX)^{\otimes m}]}{dt}\bigg|_{t=0} = m \mathrm{Tr}\big[O^{\mathrm{sym}}_{\mathcal{B}}(X \otimes \rho^{\otimes (m-1)})\big].
    \end{equation}

    We now manipulate this tensor trace to isolate $X$. We can rewrite the argument by separating the local operator $X$ from the $(m-1)$ copies of $\rho$ as
    $$\mathrm{Tr}\big[O^{\mathrm{sym}}_{\mathcal{B}}(X \otimes \rho^{\otimes (m-1)})\big] = \mathrm{Tr}\big[O^{\mathrm{sym}}_{\mathcal{B}}(\mathbb{I} \otimes \rho^{\otimes (m-1)})(X \otimes \mathbb{I}^{\otimes (m-1)})\big].$$
    At this stage, we apply a standard property of the partial trace. For any global operator $Y \in \mathcal{B}(\mathcal{H}^{\otimes m})$, and a local operator $X \in \mathcal{B}(\mathcal{H})$ acting solely on the first subsystem, the global trace resolves into a local trace over the first subsystem after tracing out the remainder, i.e., $\mathrm{Tr}[Y(X \otimes \mathbb{I}^{\otimes (m-1)})] = \mathrm{Tr}[(\mathrm{Tr}_{2\dots m} Y) X]$. This holds universally since tracing over the identity on subsystems $2$ through $m$ effectively integrates out those spaces without acting on $X$.

    By substituting $Y = O^{\mathrm{sym}}_{\mathcal{B}}(\mathbb{I} \otimes \rho^{\otimes (m-1)})$ into this identity, the evaluation simplifies to
    \begin{align}
        \mathrm{Tr}\big[O^{\mathrm{sym}}_{\mathcal{B}}(X \otimes \rho^{\otimes (m-1)})\big] &= \mathrm{Tr}\Big[ \big( \mathrm{Tr}_{2\dots m}[O^{\mathrm{sym}}_{\mathcal{B}}(\mathbb{I} \otimes \rho^{\otimes (m-1)})] \big) X \Big] \nonumber \\
        \label{eq:krho}
        &= \mathrm{Tr}[O^{\mathrm{sym}}_{\mathcal{B},1} X]\,,
    \end{align}
    where we have substituted the definition of the first-order marginal kernel, $O^{\mathrm{sym}}_{\mathcal{B},1} := \mathrm{Tr}_{2\dots m}[O^{\mathrm{sym}}_{\mathcal{B}}(\mathbb{I} \otimes \rho^{\otimes (m-1)})]$.

    Substituting Eq.~\eqref{eq:krho} back into Eq.~\eqref{eq:additivity} yields the complete derivative as
    $$\frac{d \mathrm{Tr}[O^{\mathrm{sym}}_{\mathcal{B}}(\rho+tX)^{\otimes m}]}{dt}\bigg|_{t=0} = \mathrm{Tr}[m O^{\mathrm{sym}}_{\mathcal{B},1} X].$$

    Finally, equating the outcomes of the evaluation in Eq.~\eqref{eq:grho} with the above evaluation gives
    $$\mathrm{Tr}[\nabla f_{\mathcal{B}}(\rho) X] = \mathrm{Tr}[m O^{\mathrm{sym}}_{\mathcal{B},1} X].$$
    Since this trace equality stems from the definition of the matrix gradient and holds for any arbitrary Hermitian perturbation $X$, the operators themselves must be equal. Thus, it follows that
    $$\nabla f_{\mathcal{B}}(\rho) = m O^{\mathrm{sym}}_{\mathcal{B},1}\,.$$
\end{proof}

\subsection{Permutation-Invariant, Finite-Copy Local Kernels for Generic Polynomial Functionals}

In this subsection, we will generalise the techniques and results obtained in the preceding subsections to the case of a generic polynomial functional $f(\rho)$ mentioned in Eq.~\eqref{eq:generic_poly}.

Consider a degree-$i$ interleaved monomial 
\begin{equation*}
    P_i(\rho) = \mathrm{Tr}\left[ \prod_{j=1}^i (A^{(i)}_j \rho) \right]\,.
\end{equation*}
Let $\gamma_i = (1\; 2\; \dots\; i) \in S_i$ be the backward cyclic permutation such that the operator $P_{\gamma_i}$ acts on the product basis as 
\begin{equation*}
    P_{\gamma_i} |x_1, x_2, \dots, x_i\rangle = |x_2, x_3, \dots, x_i, x_1\rangle\,.
\end{equation*}

Then, the monomial trace can be expressed as a tensor expectation $P_i(\rho) = \mathbb{E}_{\rho^{\otimes n}}\left[\widetilde{K}_i\right]\equiv \mathrm{Tr}\left[ \widetilde{K}_i \rho^{\otimes i} \right]$,
where 
\begin{equation*}
    \widetilde{K}_i = \frac{1}{i!} \sum\limits_{\pi \in S_i} P_{\pi^{-1}} K_i P_\pi\,,\qquad \text{with}\qquad K_i = (A^{(i)}_1 \otimes A^{(i)}_2 \otimes \dots \otimes A^{(i)}_i) P_{\gamma_i}\,.
\end{equation*}

Since the maximum degree of $f(\rho)$ is $k$, we embed the lower-degree kernel $\widetilde{K}_i$ into the bigger joint Hilbert space $\mathcal{H}^{\otimes k}$ by padding the kernel with $(k-i)$ copies of the identity matrix. To maintain symmetry across all $k$ subsystems, we average over all $\binom{k}{i}$ possible subsystem placements, i.e.,
\begin{equation}\label{eq:O_i_cons}
    O_{(i)} = \frac{1}{\binom{k}{i}} \sum\limits_{\beta \in \binom{[k]}{i}} \widetilde{K}_i^{(\beta)} \otimes \mathbb{I}^{(k \setminus \beta)}\,.
\end{equation}

For a scalar constant term $c_0$ (a degree-0 monomial), the permutation-invariant kernel is padded as $O_{(0)} = c_0 \mathbb{I}^{\otimes k}$.

Thus, for a generic polynomial functional $f(\rho) = c_0 + \sum_{i=1}^k P_i(\rho)$, the complete permutation-invariant $k$-partite kernel is given by the linear combination
\begin{equation}\label{eq:sym_k_cons}
    O^{\mathrm{sym}}_{k} = c_0 \mathbb{I}^{\otimes k} + \sum_{i=1}^k O_{(i)}\,.
\end{equation}

\begin{remark}[Kernels of Product Polynomial Functionals]
    To illustrate the generality of the above construction, consider a polynomial functional which is a product of some trace functionals, each of the form specified by Eq.~\eqref{eq:generic_poly}, i.e., $g(\rho):=\prod_{j=1}^s f_j(\rho)$, where each $f_j$ is a polynomial functional of the form given by Eq.~\eqref{eq:generic_poly}, and let $\sum_{j=1}^s k_j=r$. Then, as per the above construction, there exists a permutation-invariant $k_j$-partite kernel $O^{\mathrm{sym}}_{k_j}$ for each of the functionals $f_j$, where $k_j$ is the maximum polynomial degree of the input to the functional $f_j$. Then, we can equivalently write 
    \begin{align*}
        g(\rho)&=\prod_{j=1}^s\mathrm{Tr}\left[O^{\mathrm{sym}}_{k_j}\rho^{\otimes k_i}\right]\\
        &=\mathrm{Tr}\left[\bigotimes_{j=1}^s \big(O^{\mathrm{sym}}_{k_j}\rho^{\otimes k_i}\big)\right]\\
        &=\mathrm{Tr}\left[\left(\bigotimes_{j=1}^s O^{\mathrm{sym}}_{k_j}\right) \rho^{\otimes r}\right]=\mathrm{Tr}\big[O^{\mathrm{sym}}_{r}\rho^{\otimes r}\big]\,,
    \end{align*}
    where $O^{\mathrm{sym}}_{r}:=\bigotimes_{j=1}^s O^{\mathrm{sym}}_{k_j}$ is permutation-invariant $r$-partite kernel for $g(\rho)$, as tensor-product preserves symmetricity.
\end{remark}

In the subsequent subsection, similar to the case for monomial functionals, we will show that the first-order marginal of the kernel $O^{\mathrm{sym}}_{k}$ constructed above is equivalent (up to some scalar multiple of the identity matrix) to the matrix gradient of the corresponding polynomial functional.

\subsection{Geometric Equivalence}

In this subsection, we generalize the equivalence between the first-order marginal kernel and the functional gradient by first deriving the matrix gradient of the interleaved representation of Eq.~\eqref{eq:generic_poly}. We state this construction in the lemma below.

\begin{lemma}[Analytical Gradient of Generic Operator Polynomials]\label{lemma:poly_gradient_construction}
    For any unknown quantum state $\rho\in\mathcal{D}(\mathcal{H})$, if a generic polynomial functional $f$ evaluated at $\rho$ is expressed as a linear combination of interleaved monomials as
        \begin{equation*}
            f(\rho) = c_0 + \sum_{i=1}^{k} \mathrm{Tr}\left[ \prod_{j=1}^{i} \big(A^{(i)}_j \rho\big) \right]\,,
        \end{equation*}
        where $c_0$ is a scalar constant, and $A^{(i)}_j$ are constant coefficient matrices, then its matrix gradient can be constructed via the cyclic permutation of its individual terms as
        \begin{equation*}
            \nabla f(\rho) = \sum_{i=1}^{k} \sum_{m=1}^{i} \left( \prod_{l=m+1}^{i} \big(A^{(i)}_l \rho\big) \right) \left( \prod_{j=1}^{m-1} \big(A^{(i)}_j \rho\big) \right) A^{(i)}_m\,.
        \end{equation*}
\end{lemma}

\begin{proof}
    To derive the functional form of $\nabla f(\rho)$, we evaluate the directional derivative along the affine path $\rho_t = \rho + t X$ using the generic polynomial functional defined above. The scalar constant $c_0$ vanishes under differentiation. Expanding the ordered matrix product $\prod_{j=1}^{i} B_j = B_1 B_2 \dots B_i$ to the first order in $t$ yields
    \begin{equation*}
        \prod_{j=1}^{i} \big(A^{(i)}_j (\rho + t X)\big) = \prod_{j=1}^{i} \big(A^{(i)}_j \rho\big) + t \sum_{m=1}^{i} \left( \prod_{j=1}^{m-1} \big(A^{(i)}_j \rho\big) \right) A^{(i)}_m X \left( \prod_{l=m+1}^{i} \big(A^{(i)}_l \rho\big) \right) + \mathcal{O}(t^2)\,.
    \end{equation*}
    
    Evaluating the derivative of the global trace at $t = 0$ isolates this linear $\mathcal{O}(t)$ sum as
    \begin{align*}
        \frac{d}{dt}f(\rho_t)\bigg|_{t=0} &= \sum_{i=1}^{k} \sum_{m=1}^{i} \mathrm{Tr}\left[ \left( \prod_{j=1}^{m-1} \big(A^{(i)}_j \rho\big) \right) A^{(i)}_m X \left( \prod_{l=m+1}^{i} \big(A^{(i)}_l \rho\big) \right) \right] \\
        &\overset{\tiny (a) }{=} \sum_{i=1}^{k} \sum_{m=1}^{i} \mathrm{Tr}\left[ \left( \prod_{l=m+1}^{i} \big(A^{(i)}_l \rho\big) \right) \left( \prod_{j=1}^{m-1} \big(A^{(i)}_j \rho\big) \right) A^{(i)}_m X \right] \\
        &= \mathrm{Tr}\left[ \left( \sum_{i=1}^{k} \sum_{m=1}^{i} \left( \prod_{l=m+1}^{i} \big(A^{(i)}_l \rho\big) \right) \left( \prod_{j=1}^{m-1} \big(A^{(i)}_j \rho\big) \right) A^{(i)}_m \right) X \right]\,,
    \end{align*}
    where equality ($a$) follows from the cyclic property of the trace, shifting the operators appearing to the right of the perturbation $X$ to the front of the trace. (By convention, any empty ordered product evaluates to the identity matrix $\mathbb{I}$.)
    
    The composite operator enclosed in the large outer parentheses matches the proposed algebraic expression for $\nabla f(\rho)$. Since the definition of the matrix gradient establishes that $\frac{d}{dt}f(\rho_t)\big|_{t=0} = \mathrm{Tr}[\nabla f(\rho) X]$, and this holds for any arbitrary Hermitian perturbation $X$, we obtain the gradient
    \begin{equation*}
        \nabla f(\rho) = \sum_{i=1}^{k} \sum_{m=1}^{i} \left( \prod_{l=m+1}^{i} \big(A^{(i)}_l \rho\big) \right) \left( \prod_{j=1}^{m-1} \big(A^{(i)}_j \rho\big) \right) A^{(i)}_m\,.
    \end{equation*}
\end{proof}

With the derivation of the matrix gradient of a generic polynomial functional, we can now formalise its algebraic relationship to the first-order marginal of the original permutation-invariant kernel.

\begin{lemma}[Gradient--Marginal Equivalence for Generic Polynomial Functionals]\label{lemma:generic_poly_algebraic}
    Let $f(\rho) = c_0 + \sum_{i=1}^k P_i(\rho)$ be a generic polynomial functional of maximum degree $k$, and $O^{\text{sym}}_k \in \mathcal{B}(\mathcal{H}^{\otimes k})$ be its corresponding permutation-invariant $k$-copy kernel as given by Eq.~\eqref{eq:O_i_cons} and \eqref{eq:sym_k_cons}. Then, the scaled first-order marginal kernel $k O^{\text{sym}}_{k,1}$ recovers the matrix gradient $\nabla f(\rho)$ up to an additive scalar multiple of the identity such that
    \begin{equation*}
        k O^{\text{sym}}_{k,1} = \nabla f(\rho) + C(\rho) \mathbb{I}\,,
    \end{equation*}
    where $C(\rho) = k c_0 + \sum_{i=1}^k (k-i) P_i(\rho)$ is a real-valued, data-dependent scalar.
\end{lemma}

\begin{proof}
    The proof proceeds by evaluating the first-order marginal of the permutation-invariant $k$-copy local kernel $O^{\text{sym}}_k$, and linking the result to the functional's gradient.
    
    By the linearity of the partial trace, we have $O^{\text{sym}}_{k,1} = c_0 \mathbb{I} + \sum_{i=1}^k O_{(i),1}$. To compute $O_{(i),1}$, observe that the first subsystem is included in the $i$-element subset $\beta$ with probability $\binom{k-1}{i-1} \Big/ \binom{k}{i} = \frac{i}{k}$. If the first subsystem is included, then tracing out the remaining $(k-1)$ subsystems leaves the first-order marginal of the unpadded $i$-copy kernel, denoted by $\widetilde{K}_{i,1} = \mathrm{Tr}_{2 \dots i}[\widetilde{K}_i (\mathbb{I} \otimes \rho^{\otimes i-1})]$. 

    Conversely, if the first subsystem is excluded from $\beta$ (with probability $1-i/k$), then it is occupied by the identity operator $\mathbb{I}$. Tracing out the remaining subsystems evaluates the full expectation of $\widetilde{K}_i$, which yields $P_i(\rho)\mathbb{I}\,$. 

    Therefore, the first-order marginal of the padded kernel splits into
    \begin{equation*}
        O_{(i),1} = \frac{i}{k} \widetilde{K}_{i,1} + \left(1-\frac{i}{k}\right) P_i(\rho) \mathbb{I}\,.
    \end{equation*}

    We now evaluate $\widetilde{K}_{i,1}$ to connect the marginal kernel with the matrix gradient. Expanding the partial trace over the group-averaged kernel yields
    \begin{equation*}
        \widetilde{K}_{i,1} = \frac{1}{i!} \sum_{\pi \in S_i} \mathrm{Tr}_{2 \dots i}\left[ \big( P_{\pi^{-1}} K_i P_\pi \big) (\mathbb{I}_1 \otimes \rho^{\otimes i-1}) \right]\,.
    \end{equation*}

    By utilizing the cyclic property of the trace and the unitarity of the permutation operators, we shift the conjugation onto the state to get
    \begin{equation*}
        \widetilde{K}_{i,1} = \frac{1}{i!} \sum_{\pi \in S_i} \mathrm{Tr}_{2 \dots i}\left[ K_i \Big( P_\pi (\mathbb{I}_1 \otimes \rho^{\otimes i-1}) P_{\pi^{-1}} \Big) \right]\,.
    \end{equation*}

    The conjugated state $P_\pi (\mathbb{I}_1 \otimes \rho^{\otimes i-1}) P_{\pi^{-1}}$ moves the single untraced identity operator $\mathbb{I}$ from the first tensor factor to the $m^{\text{th}}$ tensor factor, where $m = \pi(1)$. Therefore, the symmetric group $S_i$ partitions into $i$ equivalence classes based on this index $m \in \{1, \dots, i\}$. We split the sum over the permutations accordingly as
    \begin{equation*}
        \widetilde{K}_{i,1} = \frac{1}{i!}\sum_{m=1}^{i} \sum_{\substack{\pi \in S_i:\\\pi(1) =m}} \mathrm{Tr}_{2 \dots i}\left[ K_i \Big( P_\pi (\mathbb{I}_1 \otimes \rho^{\otimes i-1}) P_{\pi^{-1}} \Big) \right]\,.
    \end{equation*}

    Since the remaining $(i-1)$ tensor factors are populated by identical copies of $\rho$, any permutations that shuffle these identical copies leave the joint state invariant. Thus, for any permutation in the equivalence class defined by $\pi(1)=m$, the conjugated state evaluates to $\rho^{\otimes m-1} \otimes \mathbb{I}_m \otimes \rho^{\otimes i-m}$. Since each class contains $(i-1)!$ permutations, the inner sum evaluates to $(i-1)!$ identical terms. Furthermore, since the untraced identity slot has shifted to the $m^{\text{th}}$ position, the effective partial trace operates over the remaining subsystems, denoted as $[i] \setminus m$. The expression simplifies to a sum over the $i$ spatial configurations as
    \begin{equation*}
        \widetilde{K}_{i,1} = \frac{1}{i!} \sum_{m=1}^i (i-1)! \, \mathrm{Tr}_{[i]\setminus m}\left[ K_i \Big( \rho^{\otimes m-1} \otimes \mathbb{I}_m \otimes \rho^{\otimes i-m} \Big) \right]\,.
    \end{equation*}

    We now substitute $K_i = (A^{(i)}_1 \otimes \dots \otimes A^{(i)}_i)P_{\gamma_i}$ and evaluate the partial trace. The cyclic shift $P_{\gamma_i}$ connects subsystem $s$ to subsystem $(s-1)$, weaving the tensor factors into a contiguous operator loop. Since the identity matrix sits at the $m^{\text{th}}$ position (acting as the untraced open slot), the loop breaks exactly at the operator $A^{(i)}_m$. Tracing out the remaining $(i-1)$ subsystems against $\rho$ evaluates the rest of the chain, leaving the $m^{\text{th}}$ subsystem as the output, i.e.,
    \begin{equation*}
        \mathrm{Tr}_{[i]\setminus m}\left[ K_i \Big( \rho^{\otimes m-1} \otimes \mathbb{I}_m \otimes \rho^{\otimes i-m} \Big) \right] = \left( \prod_{l=m+1}^{i} (A^{(i)}_l \rho) \right) \left( \prod_{j=1}^{m-1} (A^{(i)}_j \rho) \right) A^{(i)}_m\,.
    \end{equation*}

    Cancelling the factorials $\big( (i-1)! / i! = 1/i \big)$ and substituting this trace evaluation back into the sum yields the geometric form of the first-order marginal as
    \begin{equation*}
        \widetilde{K}_{i,1} = \frac{1}{i} \sum_{m=1}^i \left( \prod_{l=m+1}^{i} (A^{(i)}_l \rho) \right) \left( \prod_{j=1}^{m-1} (A^{(i)}_j \rho) \right) A^{(i)}_m\,.
    \end{equation*}

    Multiplying the marginal kernel by $i$ recovers the sum over all cyclic splits as
    \begin{equation*}
        i \widetilde{K}_{i,1} = \sum_{m=1}^i \left( \prod_{l=m+1}^{i} (A^{(i)}_l \rho) \right) \left( \prod_{j=1}^{m-1} (A^{(i)}_j \rho) \right) A^{(i)}_m\,.
    \end{equation*}

    As established in Lemma \ref{lemma:poly_gradient_construction}, this right-hand side is identically the analytical gradient of the monomial, $\nabla P_i(\rho)$. Substituting $i \widetilde{K}_{i,1} = \nabla P_i(\rho)$ back into the padded marginal kernel equation and scaling by $k$ gives
    \begin{equation*}
        k O_{(i),1} = \nabla P_i(\rho) + (k-i) P_i(\rho) \mathbb{I}\,.
    \end{equation*}

    Summing these components over all degrees generates the scaled marginal for the total global kernel, given by
    \begin{align*}
        k O^{\text{sym}}_{k,1} &= k c_0 \mathbb{I} + \sum_{i=1}^k \Big( \nabla P_i(\rho) + (k-i) P_i(\rho) \mathbb{I} \Big) \\
        &= \nabla f(\rho) + C(\rho) \mathbb{I}\,,\quad \text{where } C(\rho) = k c_0 + \sum_{i=1}^k (k-i) P_i(\rho)\,.
    \end{align*}
\end{proof}

The additive shift $C(\rho)\mathbb{I}$ occurs since we extend lower-degree polynomials into the $k$-partite tensor space. For valid quantum states, the trace is conserved (as $\mathrm{Tr}[\rho] = 1$). Therefore, multiplying any term by $\mathrm{Tr}[\rho]$ does not change its value. 

When constructing the $k$-copy permutation-invariant kernel $O^{\text{sym}}_k$, we pad every lower-degree monomial $P_i(\rho)$ with $(k-i)$ copies of the identity matrix. Evaluating this padded observable is equivalent to multiplying the monomial by $(\mathrm{Tr}[\rho])^{k-i}$. This transforms the generic polynomial into a homogeneous polynomial of degree $k$ given by
\begin{equation}
    f_{\text{hom}}(\rho) = c_0 \big(\mathrm{Tr}[\rho]\big)^k + \sum_{i=1}^k P_i(\rho) \big(\mathrm{Tr}[\rho]\big)^{k-i}\,.
\end{equation}
The marginal kernel evaluates the gradient of this homogenized functional $f_{\text{hom}}(\rho)$, rather than the original $f(\rho)$. Applying the product rule to find the gradient of $f_{\text{hom}}(\rho)$ generates trace derivatives (i.e., $\nabla \mathrm{Tr}[\rho] = \mathbb{I}$). Differentiating these padding terms produces the aforementioned $C(\rho)\mathbb{I}$ shift.

Despite this shift, the scaled marginal and the analytical gradient behave identically when evaluated along valid physical trajectories. This is due to the fact that any such valid physical trajectory is formalized as a \textit{traceless} Hermitian perturbation. We state this fact formally as a corollary.

\begin{corollary}[Marginal--Gradient Equivalence for Polynomial Functionals]\label{lemma:generic_poly_equivalence}
    The scaled first-order marginal kernel $kO^{\text{sym}}_{k,1}$ and the analytical matrix gradient $\nabla f(\rho)$ are equivalent along valid physical trajectories over the manifold of valid density matrices. Specifically, for any traceless Hermitian perturbation $X$, their directional derivatives coincide as 
    \begin{equation}
        \mathrm{Tr}\big[X \big(k O^{\text{sym}}_{k,1}\big)\big] = \mathrm{Tr}\big[X \big(\nabla f(\rho)\big)\big]\,.
    \end{equation}
\end{corollary}

\begin{remark}[On the Necessity of Traceless Perturbations]
    \textit{Note that for the generic polynomial case, the equivalence between the scaled first-order marginal kernel $k O^{\text{sym}}_{k,1}$ and the functional gradient $\nabla f(\rho)$ requires the perturbation $X$ to be traceless. This differs from the homogeneous monomial case. For a purely homogeneous functional $f(\rho) = \mathrm{Tr}[\mathcal{B}\rho^m]$, no identity padding is used ($k-m = 0$), so the scalar shift $C(\rho)$ is zero. Thus, for homogeneous monomials, the equality $m O^{\mathrm{sym}}_{\mathcal{B},1} = \nabla f(\rho)$ holds for any Hermitian matrix $X$, even if it is not traceless. However, for a mixed-degree polynomial, the $C(\rho)\mathbb{I}$ shift means the equivalence only holds for traceless physical perturbations; if we evaluated a non-traceless deviation ($\mathrm{Tr}[X] \neq 0$), the directional derivatives would differ.}
\end{remark}

\subsection{Statistical Consequences}

The equivalence obtained in Lemma \ref{lemma:generic_poly_algebraic} indeed has direct statistical implications. Let $O_n^{\text{sym}}$ be any permutation-invariant observable with the property that $\mathrm{Tr}[O_n^{\text{sym}} \rho^{\otimes n}] = f(\rho)$. Later in the manuscript, we show that 
\begin{align}
    \mathrm{Var}_{\rho^{\otimes n}}(O_n^{\text{sym}}) &= \frac{1}{n} \mathrm{Var}_{\rho}\big(k O^{\text{sym}}_{k,1}\big) + \mathcal{O}\left(\frac{1}{n^2}\right)\label{eq:implication1}\\
    &= \frac{1}{n} \mathrm{Var}_{\rho}\big(\nabla f(\rho)\big) + \mathcal{O}\left(\frac{1}{n^2}\right)\label{eq:implication2},
\end{align}
where in Eq.~\eqref{eq:implication1}, $k$ is the maximum polynomial degree of the input to the functional. Eq.~\eqref{eq:implication1} follows since we show in our manuscript (see Subsection \ref{subsec:U_stat_univ}) that for any order $n$, a globally unbiased permutation-invariant kernel $O_n^{\text{sym}}$ is uniquely identified by U-statistics. Eq.~\eqref{eq:implication2} follows from Lemma \ref{lemma:generic_poly_algebraic} above. Moreover, the quantity $\nabla f(\rho)$ is closely related to QFI, and therefore these unbiased and permutation-invariant kernels are actually Cramér-Rao efficient. We prove this formally in Section \ref{sec:variance_optimality}.

\section{Quantum U-Statistics and Its Uniqueness}\label{sec:U_stat_unique}
Let $f(\rho)$ be as defined in Eq.~\eqref{eq:generic_poly} with the property that the degree of the input polynomial to $f(\rho)$ is $k$, and consider a permutation-invariant observable $O_n^{\text{sym}}$ with the property that $\mathrm{Tr}[O_n^{\text{sym}} \rho^{\otimes n}] = f(\rho)$, where $n>k$. In this section, we show that this observable $O_n^{\text{sym}}$ is uniquely identified by a U-statistic observable, which is an $ n$-copy permutation-invariant observable obtained by extending an unbiased permutation-invariant $k$-copy observable of $f(\rho)$, also popularly known as \textit{kernel} in the literature \cite{GB2010,BOW19}. Before proving this uniqueness, let us first discuss quantum U-statistics.

\subsection{Quantum U-Statistics}

Suppose we have access to $n$ copies of $\rho$ and we want to estimate $f(\rho)$. Then, the quantum U-statistic observable acting on $\rho^{\otimes n}$ is defined as follows:

\begin{definition}[{\cite[Definition $6.3$]{GB2010}}]\label{def:symm_u_stat}
    Given any product state $\rho^{\otimes n}$ with $\rho\in \mathcal{D}(\mathcal{H})$, let $O_k^{\text{sym}}$ denote a permutation-invariant $k$-copy kernel of a polynomial functional $f(\rho)$, i.e., $\mathrm{Tr}[O_k^{\text{sym}} \rho^{\otimes k}] = f(\rho)$, with $0<k < n$. Then, the $n$-copy, $k^{\text{th}}$-order quantum U-statistic observable $U_{n,k}$ acting on $\rho^{\otimes n}$ is defined as
    \begin{equation}
        U_{n,k}:=\frac{1}{\binom{n}{k}}\sum_{1\le i_1<i_2<\ldots<i_k\le n}\big(O_k^{\text{sym}}\big)_{(i_1,i_2,\ldots,i_k)},\,
    \end{equation}
    where the observable $\big(O_k^{\text{sym}}\big)_{(i_1,i_2,\ldots,i_k)}$ is the permutation-invariant $k$-copy kernel $O_k^{\text{sym}}$ acting on the $i_1^{\text{th}}, i_2^{\text{th}}, \ldots, i_k^{\text{th}}$ subsystems with $1\leq i_1< i_2<\ldots< i_k\leq n$.
\end{definition}

While a local $k$-copy kernel provides an unbiased estimate, evaluating it on a single subset wastes the remaining $(n-k)$ copies. The U-statistic resolves this by uniformly averaging the local measurement across all $\binom{n}{k}$ subsystem combinations.

To illustrate, consider the functional $f(\rho) = \mathrm{Tr}[\rho^2\sigma^{-1}]$ with a fully known reference state $\sigma$. To construct the permutation-invariant kernel for it, in Eq.~\eqref{eq:mono_kernel_cons}, we choose $\mathcal{B}$ to be $\sigma^{-1}$ and set $\gamma_m$ to the transposition $(1,2)$, making the cyclic permutation operator the $\mathrm{SWAP}$ operator $\mathbb{S}$. We then recover the kernel given in Eq.~\eqref{eq:mono_kernel_cons}. Evaluating this for $m=2$, the permutation-invariant, two-copy local kernel takes the explicit form 
\begin{equation}
    O^{\mathrm{sym}}_{\sigma^{-1}} = \frac{1}{2}\big(\mathbb{S}(\sigma^{-1}\otimes \mathbb{I}) + (\sigma^{-1}\otimes \mathbb{I})\mathbb{S}\big)\,.
\end{equation}

Given an $n$-copy product state $\rho^{\otimes n}$ ($n > 2$), the corresponding $n$-copy, $2^{\text{nd}}$-order quantum U-statistic observable $U_{n,2}$ averages the above kernel $O^{\mathrm{sym}}_{\sigma^{-1}}$ over all distinct pairs chosen from $\rho^{\otimes n}$, i.e.,
\begin{equation}
    U_{n,2} = \frac{1}{\binom{n}{2}} \sum_{1 \le i < j \le n} \big(O^{\mathrm{sym}}_{\sigma^{-1}}\big)_{(i,j)}\,,
\end{equation}
where the observable $\big(O^{\mathrm{sym}}_{\sigma^{-1}}\big)_{(i,j)}$ applies the permutation-invariant, two-copy local kernel to the $i^{\text{th}}$ and $j^{\text{th}}$ subsystems, and acts as the global identity elsewhere. By structurally enforcing permutation invariance, this construction strictly minimizes statistical variance in accordance with the quantum Rao-Blackwell theorem. Since its non-commutative formalization by Gu\c{t}\u{a} and Butucea \cite{GB2010}, quantum U-statistics have remained the standard framework for asymptotically efficient multi-copy estimation.

\subsection{Uniqueness of Quantum U-Statistic}\label{subsec:U_stat_univ}

Having established Quantum U-statistics as a symmetrization method, we now prove its universality. In this subsection, we show that for any polynomial functional, the Quantum U-statistic is the mathematically unique unbiased permutation-invariant $n$-copy extension. We establish this by demonstrating that any arbitrary permutation-invariant unbiased estimator inevitably collapses to the exact U-statistic.

To formalize this, let $U_{n,k}$ be the global $n$-copy, $k^{\text{th}}$-order quantum U-statistic observable for a polynomial functional $f(\rho)$. Suppose there exists another arbitrary, permutation-invariant $n$-copy kernel $O_n^{\text{sym}}$ that also provides an unbiased estimate of $f(\rho)$ for all quantum states $\rho \in \mathcal{D}(\mathcal{H})$. Since both observables unbiasedly estimate the same functional, their expectation values on any product state $\rho^{\otimes n}$ are identical, i.e.,
\begin{equation}
    \mathrm{Tr}\big[O_n^{\text{sym}} \rho^{\otimes n}\big] = f(\rho) \quad \text{and} \quad \mathrm{Tr}\big[U_{n,k} \rho^{\otimes n}\big] = f(\rho).
\end{equation}

We define the difference operator $D_n := O_n^{\text{sym}} - U_{n,k}$. Since both $O_n^{\text{sym}}$ and $U_{n,k}$ are permutation-invariant Hermitian operators, their linear difference $D_n$ is necessarily a permutation-invariant Hermitian operator as well. Evaluating the expectation value of this difference operator yields
\begin{equation}
    \mathrm{Tr}\big[D_n \rho^{\otimes n}\big] = \mathrm{Tr}\big[O_n^{\text{sym}} \rho^{\otimes n}\big] - \mathrm{Tr}\big[U_{n,k} \rho^{\otimes n}\big] = f(\rho) - f(\rho) = 0\,,
\end{equation}
which holds for every $\rho \in \mathcal{D}(\mathcal{H})$. We will use this observation to conclude that $D_n=0.$ 

Towards this, consider the following subspace 
\begin{equation}
    \mathcal{K}(M_d) := \big\{ Y \in M_d^{\otimes n} \mid [P_{\pi},Y]=0, \;\forall \pi \in S_n \big\}\,, \label{eq:commutant_subspace}
\end{equation}
where $M_d$ is the vector space of all $d\times d$ matrices. Observe that the operator $D_n$ is an element of this subspace $\mathcal{K}(M_d)$. In Lemma \ref{lemma:operator_span} below, we show that $\mathcal{K}(M_d) = \mathrm{span}_{\mathbb{C}} \big\{ A^{\otimes n} \mid A \in M_d \big\}$. This structural property will allow us to conclude that $D_n = 0$. 

Before formally stating the lemma, we note that an equivalent characterization of the commutant subspace of Eq.~\eqref{eq:commutant_subspace} was established in \cite[Corollary 4]{Harrow2013}. However, our constructive algebraic proof is comparatively simpler and more direct than Harrow's approach. We discuss these methodological differences in detail in Remark \ref{rem:compare_haarrow}.

\begin{lemma}[Permutation-Invariant Basis of Operator Tensor Space]\label{lemma:operator_span}
    Let $M_d$ be the vector space of all $d\times d$ matrices. Define the following permutation-invariant subspace
    \begin{equation}
        \mathcal{K}(M_d) := \big\{ Y \in M_d^{\otimes n} \mid [P_{\pi},Y]=0, \;\forall \pi \in S_n \big\}\,,
    \end{equation}
    where $P_{\pi}$ denotes the unitary permutation operator corresponding to the permutation $\pi\in S_n$. Then, 
    \begin{equation}
        \mathcal{K}(M_d) = \mathrm{span}_{\mathbb{C}} \big\{ A^{\otimes n} \mid A \in M_d \big\}\,.
    \end{equation}
\end{lemma}

\begin{proof}
    See Appendix \ref{app:operator_span} for the proof.
\end{proof}

In fact, to prove that $D_n=0$, we do not need to invoke Lemma \ref{lemma:operator_span} in its full generality. Indeed, one can easily verify that for the vector space $L_d$ of all $d\times d$ Hermitian matrices, the same proof of Lemma \ref{lemma:operator_span} naturally yields 
\begin{equation}
    \mathcal{K}(L_d)=\mathrm{span}_{\mathbb{R}}\big\{X^{\otimes n}\mid X\in L_d\big\}\,,
\end{equation}
where $\mathcal{K}(L_d)$ is as defined in Lemma \ref{lemma:operator_span} (with the Hermitian tensor space $L_d^{\otimes n}$ replacing the general tensor space $M_d^{\otimes n}$). We state this direct fact formally as a corollary.

\begin{corollary}\label{col:operator_span}
    Let $L_d$ be the vector space of all $d\times d$ Hermitian matrices. Define the following permutation-invariant subspace
    \begin{equation}
        \mathcal{K}(L_d) := \big\{ Y \in L_d^{\otimes n} \mid [P_{\pi},Y]=0, \;\forall \pi \in S_n \big\}\,,
    \end{equation}
    where $P_{\pi}$ denotes the unitary permutation operator corresponding to the permutation $\pi\in S_n$. Then, 
    \begin{equation}
        \mathcal{K}(L_d) = \mathrm{span}_{\mathbb{R}} \big\{ X^{\otimes n} \mid X \in L_d \big\}\,.
    \end{equation}
\end{corollary}

Using the above span property, one can prove the following corollary.

\begin{corollary}[Uniqueness of Permutation-Invariant Operators]\label{col:symmetric_uniqueness}
    Let $D_n \in \mathcal{B}(\mathcal{H}^{\otimes n})$ be a Hermitian operator which is permutation-invariant. If the expectation value $\mathrm{Tr}[D_n \rho^{\otimes n}] = 0, \forall \rho \in \mathcal{D}(\mathcal{H})$, then $D_n$ is identically the null operator (i.e., $D_n = 0$).
\end{corollary}

\begin{proof}
    Consider any arbitrary non-zero Hermitian operator $X$ ($\in L_d$). For any $\varepsilon \in \left(0, 1/{\|X\|}_\infty\right)$, the operator $\gamma = (\mathbb{I}+\varepsilon X)/\mathrm{Tr}[\mathbb{I}+\varepsilon X]$ is positive and has unit trace, thus constituting a valid density operator in $\mathcal{D}(\mathcal{H})$. Evaluating the expectation value of $D_n$ on the product state $\gamma^{\otimes n}$ gives
    \begin{equation*}
        \mathrm{Tr}[D_n\gamma^{\otimes n}] = 0\,.
    \end{equation*}
    Expanding $\gamma$ into the left-hand side of the above expression, we get
    \begin{equation}\label{eq:trace_polynomial}
        \frac{1}{\left(\mathrm{Tr}[\mathbb{I}+\varepsilon X]\right)^n} \mathrm{Tr}\left[D_n(\mathbb{I}+\varepsilon X)^{\otimes n}\right] = 0\,.
    \end{equation}
    For the above equality to hold across the continuous interval $\varepsilon \in \left(0, 1/{\|X\|}_\infty\right)$, the numerator $\mathrm{Tr}\left[D_n(\mathbb{I}+\varepsilon X)^{\otimes n}\right]$, which is a polynomial in $\varepsilon$ of degree at most $n$, must be identically equal to zero. Consequently, all coefficients of this polynomial must independently vanish. Extracting the coefficient of the highest-order term $\varepsilon^n$ dictates that $\mathrm{Tr}[D_n X^{\otimes n}] = 0$, for any arbitrary Hermitian operator $X$.

    By Corollary \ref{col:operator_span}, the subspace $\mathcal{K}(L_d)$ (of the Hermitian tensor space $L_d^{\otimes n}$) is linearly spanned by elements of the form $X^{\otimes n}$ where $X \in L_d$. Since $\mathrm{Tr}[D_n X^{\otimes n}] = 0$ for all $X \in L_d$, linearity ensures that $D_n$ is orthogonal to the entire subspace $\mathcal{K}(L_d)$. Since $D_n$ itself is a permutation-invariant Hermitian operator, it resides entirely within the subspace $\mathcal{K}(L_d)$. Therefore, its inner product with itself must be zero ($\mathrm{Tr}[D_n^2] = 0$), which implies $D_n$ is identically the null operator ($D_n = 0$).
\end{proof}

Thus, the above corollary immediately implies that $O_n^{\text{sym}} = U_{n,k}$, establishing that the Quantum U-statistic is the sole valid permutation-invariant extension. We state this consequence as a corollary.

\begin{corollary}[Uniqueness of Quantum U-Statistic]\label{lemma:u_stat_uniqueness}
    Given an unknown state $\rho\in\mathcal{D}(\mathcal{H})$, let $U_{n,k}$ be the global $n$-copy, $k^{\text{th}}$-order quantum U-statistic observable for a polynomial functional $f(\rho)$ with $0<k<n$, and let $O_n^{\text{sym}} \in \mathcal{B}(\mathcal{H}^{\otimes n})$ be any arbitrary permutation-invariant kernel that unbiasedly estimates $f(\rho)$. Then, $O_n^{\text{sym}}$ is identically equal to $U_{n,k}\,,$ for all $n\in\mathbb{Z}_{>0}$.
\end{corollary}

\begin{remark}[Comparison between Lemma \ref{lemma:operator_span} and Corollary $4$ of Harrow \cite{Harrow2013}]\label{rem:compare_haarrow}
    Lemma \ref{lemma:operator_span} is precisely the operator-spanning result proved in Corollary 4 of Harrow \cite{Harrow2013}. Both results characterize the commutant of the permutation representation on ($M_d^{\otimes n}$) as the linear span of identical tensor powers ($X^{\otimes n}$). Thus, we do not claim the spanning property itself as a new result.

    Our proof, however, is arguably simpler and more direct. Harrow derives the operator result through the corresponding spanning property of the permutation-invariant vector subspace, whereas we work directly in the operator tensor space. Our proof uses only elementary algebra: permutation-invariant tensors are written as linear combinations of symmetrized elementary tensors, and an induction based on the expansion of $(X+tY)^{\otimes n}$ shows that each such tensor is a linear combination of identical tensor powers. Thus, while Lemma \ref{lemma:operator_span} gives the same statement as Harrow's Corollary $4$, our proof provides a short, self-contained, and elementary derivation directly suited to the estimation framework of this manuscript.
\end{remark}

\section{Cram\'er-Rao Efficiency of Quantum U-Statistic}\label{sec:variance_optimality}

From Corollary \ref{lemma:u_stat_uniqueness} of the previous section, any unbiased, permutation-invariant $n$-copy kernel $O^{\mathrm{sym}}_{n}$ for estimating a polynomial functional 
\begin{equation*}
    f(\rho) = c_0 + \sum\limits_{i=1}^{k} \mathrm{Tr}\left[ \prod\limits_{j=1}^{i} \big(A^{(i)}_j \rho\big) \right]\,,
\end{equation*}
is the same as the U-statistic estimator $U_{n,k}$. Therefore, we have
\begin{align}
    \mathrm{Var}_{\rho^{\otimes n}}(O_n^{\text{sym}})
    &= \frac{1}{n} \mathrm{Var}_{\rho}\big(k O_{k,1}^{\text{sym}}\big) + \mathcal{O}\left(\frac{1}{n^2}\right)\label{eq:hoeffding_decomp_var}\\
    &= \frac{1}{n} \mathrm{Var}_{\rho}\big(\nabla f(\rho)\big) + \mathcal{O}\left(\frac{1}{n^2}\right)\,,\label{eq:grad_marg_connection}
\end{align}
where $O_{k,1}^{\text{sym}}$ is the first-order marginal of the permutation-invariant kernel $O_{k}^{\text{sym}}$ for estimating $f(\rho)$, as constructed in Eq.~\eqref{eq:sym_k_cons}. Here, Eq.~\eqref{eq:hoeffding_decomp_var} follows from \cite[Lemma $6.4$]{GB2010}\footnote{For more details on this, see Appendix \ref{app:proof_grad_marg}.}, and Eq.~\eqref{eq:grad_marg_connection} follows from Lemma \ref{lemma:generic_poly_algebraic} above.

From Eq.~\eqref{eq:grad_marg_connection}, we have identified the leading-order term in the variance of $O_n^{\text{sym}}$ to be  $\mathrm{Var}_{\rho}\big(\nabla f(\rho)\big)$.
The natural question is whether one can construct an estimator whose asymptotic variance attains this benchmark, in the sense of Cram\'er-Rao efficiency. Before exploring this possibility, let us first discuss the quantum Cram\'er-Rao type bound.

\subsection{The Quantum Cram\'er-Rao Type Bound}

In quantum estimation theory, the fundamental limit on estimation precision is governed by the quantum Cram\'er-Rao bound (QCRB) \cite{Helstrom1967,Boixo_2007}. To formulate these bounds in local quantum estimation, Suzuki et al. \cite{Suzuki_2020} developed a comprehensive statistical framework based on QFI. In this manuscript, we employ a specialized, Symmetric Logarithmic Derivative (SLD)-based version of their framework, as presented in \cite{Hayashi_2025}.

In \cite{Hayashi_2025}, the authors, under certain regularity assumptions, considered a quantum statistical model with a parameter vector $\theta = (\theta^1,\cdots,\theta^k) \in \mathbb{R}^k$, which generates a family of density operators $\rho_\theta$. Below, we first discuss these regularity conditions in Assumption \ref{assump:regularity}. 

\begin{assumption}[Standard Regularity Conditions {\cite{Suzuki_2020}}]\label{assump:regularity}
For the parameterization mapping $\theta \mapsto \rho_\theta$, we assume:
    \begin{enumerate}
        \item \textbf{Positivity (Full-Rank):} The quantum states are positive (i.e., $\rho_\theta > 0$), and
        \item \textbf{Smoothness and Independence:} The mapping is smooth, and the tangent operators $\partial \rho_\theta/\partial \theta^j$ are linearly independent.
    \end{enumerate}
\end{assumption}

These regularity conditions are essential for the following reasons:
\begin{itemize}
    \item \textbf{Resolving the Infinitesimal Perturbation Issue:} If the true state $\rho_*$ is rank-deficient and resides on the boundary of the density matrix space, moving an infinitesimal distance in an arbitrary traceless Hermitian direction instantly produces negative eigenvalues. This violates positive semi-definiteness and renders local parameterizations invalid. Restricting the analysis to strictly positive states guarantees a valid, fully dimensional neighborhood of quantum states.
    \item \textbf{Resolving the Unique Lyapunov Solution Existence Issue:} Defining the SLD operators requires uniquely solving the Lyapunov equation (Eq.~\eqref{eq:sld_lyapunov}). When $\rho_\theta$ is strictly positive, the linear super-operator governing the Lyapunov equation is invertible on the space of Hermitian matrices, ensuring the unique existence of bounded SLD operators. More broadly, the guaranteed existence and uniqueness of solutions to continuous-time Lyapunov systems are foundational in matrix analysis, where unique solutions are given by convergent integral representations of the form 
    \begin{equation*}
        X = \int_0^\infty e^{\tau A} Q e^{\tau A^\dagger} d\tau\,,
    \end{equation*}
    when $A$ is Hurwitz stable \cite{BR97, Heinz1951} (see also Proposition \ref{prop:lyapunov}).
\end{itemize}

For each parameter $\theta^j$, the Symmetric Logarithmic Derivative (SLD) operator $L_j$ is the self-adjoint solution to the Lyapunov equation, given by
\begin{equation}
    \frac{\partial \rho_\theta}{\partial \theta^j} = \frac{1}{2}(\rho_\theta L_j + L_j \rho_\theta) \equiv \rho_\theta \circ L_j\,, \label{eq:sld_lyapunov}
\end{equation}
where the operation $\circ$ is known as the matrix Jordan product.

The distinguishability of these states under small parameter shifts is given by the SLD quantum Fisher information (QFI) matrix $J(\theta)$. The elements of the matrix are defined by the SLD inner product given by
\begin{equation}
    J_{ij}(\theta) := \langle L_i, L_j \rangle_{\rho_\theta}\,.\label{eq:SLD_QFI_matrix}
\end{equation}

Since probabilities must sum to one ($\mathrm{Tr}[\rho_\theta] = 1$), the expectation value of the SLD operator is always zero, i.e., $\mathrm{Tr}[\rho_\theta L_j] = 0$. This holds since taking the partial derivative of the unit trace gives 
\begin{equation*}
    \frac{\partial}{\partial \theta^j} \mathrm{Tr}[\rho_\theta] = \mathrm{Tr}\left[ \frac{\partial \rho_\theta}{\partial \theta^j} \right] = 0\,.
\end{equation*}
Substituting the Lyapunov equation (Eq.~ \eqref{eq:sld_lyapunov}) into this trace and applying the linearity and cyclic property of the trace ($\mathrm{Tr}[AB] = \mathrm{Tr}[BA]$) simplifies the expression to the expectation value of the SLD as
\begin{equation*}
    \mathrm{Tr}\left[ \frac{1}{2}(\rho_\theta L_j + L_j \rho_\theta) \right] = \frac{1}{2} \Big( \mathrm{Tr}[\rho_\theta L_j] + \mathrm{Tr}[\rho_\theta L_j] \Big) = \mathrm{Tr}[\rho_\theta L_j] = 0\,.
\end{equation*}

Let a Positive Operator-Valued Measure (POVM) $\mathcal{M} = \{M_x\}$ (with $\sum_x M_x = \mathbb{I}$) define the estimation strategy. Each outcome $x$ occurs with probability $p(x|\theta) = \mathrm{Tr}[\rho_\theta M_x]$. The estimator $\hat{\theta}$ is a classical function $\hat{\theta}(x) \in \mathbb{R}^k$ that maps an outcome to a parameter estimate. Thus, $\hat{\theta} = (\hat{\theta}^1, \hat{\theta}^2, \dots, \hat{\theta}^k)^T$ is a random vector, where each component $\hat{\theta}^i$ estimates the $i^{\text{th}}$ parameter $\theta^i$.

For the estimator $\hat{\theta}$ to be locally unbiased at the true parameter $\theta_0$ (corresponding to the true unknown state $\rho_{\theta_0} = \rho_*$), its expected value must equal the true parameter, meaning $\mathbb{E}_{\theta_0}[\hat{\theta}] = \sum_x \hat{\theta}(x) \mathrm{Tr}[\rho_{\theta_0} M_x] = \theta_0$. Additionally, its local Jacobian matrix $J_E$ must equal the identity matrix at $\theta_0$ (since the elements of this Jacobian measure how the expectation of each estimator component $\hat{\theta}^i$ changes with respect to each parameter $\theta^j$), i.e.,
\begin{equation}
    (J_E)_{ij} = \left. \frac{\partial}{\partial \theta^j} \mathbb{E}_{\theta}[\hat{\theta}^i] \right|_{\theta_0} = \delta_{ij}.\label{eq:loc_jac}
\end{equation}

This condition places two geometric requirements on the estimator:
\begin{itemize}
    \item \textbf{Local Sensitivity (Diagonal Elements):} For $i = j$, the condition $\left. \frac{\partial}{\partial \theta^i} \mathbb{E}_{\theta}[\hat{\theta}^i] \right|_{\theta_0} = 1$ ensures that the estimator correctly scales with infinitesimal changes in the target parameter. If the underlying parameter shifts by a differential amount $\varepsilon$, the expected value of the estimator must shift by exactly $\varepsilon$.
    
    \item \textbf{Parameter Decoupling (Off-Diagonal Elements):} For $i \neq j$, the condition $\left. \frac{\partial}{\partial \theta^j} \mathbb{E}_{\theta}[\hat{\theta}^i] \right|_{\theta_0} = 0$ enforces independence between different parameters, i.e., if an independent parameter $\theta^j$ fluctuates, then it must not alter the expected value of the estimator for $\theta^i$. This prevents estimation bias caused by variations in other orthogonal directions of the parameter space.
\end{itemize}

In quantum parameter estimation, the fundamental variance limits are derived for the broad class of \textit{locally unbiased} estimators. This requires that at a specific operating point $\theta_0$, the estimator is properly centered and its first-order Taylor expansion precisely mirrors the parameter space. While we explicitly construct a \textit{globally unbiased} estimator in this work (where $\mathbb{E}_\theta[\hat{\theta}^i] = \theta^i, \forall i\in[k]$, holds for all $\theta$), the quantum Cram\'er-Rao framework remains formulated locally since the QFI is a local geometric quantity. Since any globally unbiased estimator inherently satisfies the local unbiasedness conditions at every possible $\theta_0$, the local variance bound establishes a theoretical floor that our global estimator must obey.

Under these conditions, the fundamental multiparameter quantum Cram\'er-Rao inequality, originally established by Helstrom and Holevo \cite{Helstrom1968, helstrom1976quantum, Holevo1982}, bounds the covariance matrix of any arbitrary estimator $\hat{\theta}$ as
\begin{equation*}
    \mathrm{Cov}(\hat{\theta}) \ge J_E J^{-1}(\theta_0) J_E^T\,,
\end{equation*}
where $J(\theta_0)$ is the SLD QFI matrix (see Eq.~ \eqref{eq:SLD_QFI_matrix}) evaluated at the true parameter $\theta_0$ and $J_E$ is the local Jacobian matrix (see Eq.~ \eqref{eq:loc_jac}) evaluated at $\theta_0$. 

This structure highlights a direct statistical trade-off: an estimator could artificially achieve a lower variance by under-responding to parameter shifts (i.e., elements of $J_E$ being less than $1$), but this noise reduction comes at the cost of estimation bias. To establish a lower bound without bias, the standard framework restricts the analysis to locally unbiased estimators. Enforcing local sensitivity and parameter decoupling ensures the Jacobian is exactly the identity matrix ($J_E = \mathbb{I}$). This collapses the generalized inequality to its standard form, which is
\begin{equation*}
    \mathrm{Cov}(\hat{\theta}) \ge J^{-1}(\theta_0)\,.
\end{equation*}

To estimate a scalar functional $f(\rho_\theta)$, we use a parameterization where the first component $\theta^1$ is the target functional (i.e., $\theta^1 = f(\rho_\theta)$). The remaining components $\theta^2, \dots, \theta^k$ are said to be \textit{nuisance parameters}, which are basically unknown degrees of freedom that define the state but are not the target of the estimation. 

In multiparameter estimation, the minimum variance for the target estimator is bounded by the $(1,1)$ component of the inverse QFI matrix, $[J^{-1}(\theta_0)]_{11}$, rather than just its isolated Fisher information. To understand the orthogonality condition, we partition the QFI matrix $J(\theta_0)$ into block-submatrices as
\begin{equation*}
    J = \begin{pmatrix} J_{11} & \mathbb{J}^T \\ \mathbb{J} & J_{\mathrm{nuis}} \end{pmatrix}\,,
\end{equation*}
where $J_{11} = \langle L_1, L_1 \rangle_{\rho_*}$ is the isolated Fisher information of the target parameter $\theta^1$, $J_{\mathrm{nuis}}$ is the $(k-1) \times (k-1)$ QFI submatrix corresponding to the nuisance parameters, and $\mathbb{J}$ is a column vector containing the cross-information terms $J_{j1} = \langle L_j, L_1 \rangle_{\rho_*}\,$, where $2\le j\le k$.

Using the block matrix inversion formula, the $(1,1)$ component of the inverse QFI matrix is the reciprocal of the Schur complement of the submatrix $J_{\mathrm{nuis}}$ corresponding to the nuisance parameters, i.e.,
\begin{equation*}
    [J^{-1}(\theta_0)]_{11} = \frac{1}{J_{11} - \mathbb{J}^T J_{\mathrm{nuis}}^{-1} \mathbb{J}}\,.
\end{equation*}
Since the QFI matrix is positive semi-definite, the term $\mathbb{J}^T J_{\mathrm{nuis}}^{-1} \mathbb{J}$ is non-negative. Subtracting this term in the denominator represents a statistical penalty: it increases the minimum variance due to the overlap (or cross-talk) between the target parameter and the nuisance parameters.

To find the tightest bound and decouple the target functional from the unknown background state, we must eliminate this penalty. We do this by constructing a local parameterization where the tangent operators for the nuisance parameters correspond to SLDs $L_j$ that are orthogonal to the target SLD $L_1$ under the SLD inner product. 

By enforcing the orthogonality condition $\langle L_1, L_j \rangle_{\rho_*} = 0$ for all $j \ge 2$, every element of the cross-information vector $\mathbb{J}$ becomes zero. This makes the QFI matrix block-diagonal. In this orthogonal geometry, the penalty term $\mathbb{J}^T J_{\mathrm{nuis}}^{-1} \mathbb{J}$ vanishes. This simplifies the single-copy multiparameter Cramér-Rao type bound $V_{\mathrm{CR}}$ for the functional to the inverse variance of its isolated SLD as
\begin{equation*}
    V_{\mathrm{CR}} = [J^{-1}(\theta_0)]_{11} = \frac{1}{J_{11}} = \langle L_1, L_1 \rangle_{\rho_*}^{-1}\,.
\end{equation*}

This relation, also stated below in Proposition \ref{prop:qcrb_loca}, provides the basis for deriving the Cram\'er-Rao variance bound of our permutation-invariant kernel $O_n^{\mathrm{sym}}$ for quantum polynomial functionals.

\subsection{Additivity of Quantum Cram\'er--Rao Bound and $n$-Copy Scaling Limit \cite{Suzuki_2020}}

When provided with $n$ independent and identically distributed (i.i.d.) copies of the quantum state, the joint statistical model is governed by the product state $\rho_\theta^{\otimes n}$. A fundamental property of the QFI is its additivity over independent subsystems. Consequently, as originally established by Helstrom and Holevo \cite{helstrom1976quantum, Holevo1982}, the additivity of the QFI over independent subsystems ensures that the QFI matrix for the collective $n$-copy state scales linearly with the number of copies as $J_{ij}^{(n)}(\theta) = n J_{ij}(\theta)\,.$

This additivity establishes the origin of the $1/n$ statistical scaling limit. By taking the matrix inverse of the $n$-copy QFI matrix, we obtain $[J^{(n)}(\theta)]^{-1} = J^{-1}(\theta)\big/n\,.$

Therefore, the minimum achievable variance for estimating the target parameter using any locally unbiased measurement strategy across $n$ copies is exactly the single-copy bound divided by $n$, i.e., $V_{\mathrm{CR}}^{(n)} = V_{\mathrm{CR}}\big/n\,.$

As we formally prove later in this section of the manuscript, the single-copy Cram\'er-Rao bound for a generic polynomial functional $f$ computed at $\rho$ evaluates to $V_{\mathrm{CR}} = \mathrm{Var}_{\rho_*}\big(\nabla f(\rho_*)\big)$, where the variance is taken over the functional's gradient. Thus, the fundamental asymptotic variance limit for any unbiased $n$-copy estimator is identically equal to $\mathrm{Var}_{\rho_*}\big(\nabla f(\rho_*)\big)/n$.

\subsection{Parameterization and Local Unbiasedness}

\textit{To formalize the Cram\'er-Rao bound, we use the estimation framework with nuisance parameters designed in \cite{Suzuki_2020}. However, for better readability, we will utilise the version presented in \cite[Section 2]{Hayashi_2025}.} A density matrix on a $d$-dimensional Hilbert space is Hermitian and has a unit trace, so the space of quantum states has exactly $d^2-1$ independent real parameters. We choose a parameterization $\rho_\theta$ with parameters $\theta = (\theta^1, \theta^2, \dots, \theta^{d^2-1}) \in \mathbb{R}^{d^2-1}$. The first parameter is the value of the function to be estimated, and the remaining parameters are orthogonal to the first at the true state $\rho = \rho_*$.

Specifically, setting $\theta_0 := (\theta_0^1, 0, \dots, 0)$ where $\theta_0^1 := f(\rho_*)$, our parameterization must satisfy the following local unbiasedness conditions:
\begin{align}
    \rho_{\theta_0} &= \rho_*\,, \label{local_framework:eq1} \\
    \left.\frac{\partial f(\rho_\theta)}{\partial \theta^1}\right|_{\theta=\theta_0} &= 1\,, \label{local_framework:eq2} \\
    \left.\frac{\partial f(\rho_\theta)}{\partial \theta^j}\right|_{\theta=\theta_0} &= 0\,, \quad \text{for } j = 2, \dots, d^2-1\,. \label{local_framework:eq3}
\end{align}

We define the SLD inner product for two Hermitian matrices $X$ and $Y$ as $\langle X,Y \rangle_\rho := \mathrm{Tr}[\rho(X \circ Y)]$, using the symmetrization $X \circ Y := (XY+YX)/2$. We then choose SLD operators $L_1, \dots, L_{d^2-1}$ to generate the tangent vectors of the parameterization as
\begin{align}
    \left.\frac{\partial \rho_\theta}{\partial \theta^j}\right|_{\theta=\theta_0} &= \rho_* \circ L_j\,, \quad \text{for } j = 1, \dots, d^2-1\,. \label{local_framework:eq4}
\end{align}
To separate the nuisance parameters from the target parameter, we enforce the SLD orthogonality condition
\begin{align}
    \langle L_1, L_j \rangle_{\rho_*} &= \mathrm{Tr}[\rho_{\theta_0} (L_1 \circ L_j)] \notag\\
    &\overset{(a)}{=}\mathrm{Tr}[(\rho_{\theta_0} \circ L_j) L_1] \notag\\
    &=\mathrm{Tr}\left[ \left.\frac{\partial \rho_{\theta}}{\partial \theta^j}\right|_{\theta=\theta_0} L_1 \right] = 0\,, \quad \text{for } j = 2, \dots, d^2-1\,. \label{local_framework:eq5}
\end{align}
where equality ($a$) above follows from the trace duality of the matrix Jordan product.

Under this parameterization, the minimum variance bound is given by the following proposition.

\begin{proposition}[{\cite[Theorem $5.3$]{Suzuki_2020}}]\label{prop:qcrb_loca}
    Let the local parameterization of the quantum statistical model satisfy the unbiasedness and orthogonality conditions defined in \eqref{local_framework:eq1} through \eqref{local_framework:eq5}. Then, the Cram\'er-Rao type bound for estimating $f(\rho_\theta)$ is $\langle L_1, L_1 \rangle_{\rho_*}^{-1}$ at $\theta = \theta_0$.
\end{proposition}

\subsection{Derivation of the Cram\'er-Rao Bound for Generic Polynomial Functionals}

We now use Proposition \ref{prop:qcrb_loca} to find the exact variance limit for the quantum U-Statistic estimator for any polynomial functional.

\begin{lemma}[Quantum Cram\'er-Rao Bound for Polynomial Functionals]\label{lemma:Cramer-Rao_Bound_Polynomial}
    Let $\rho_* \in \mathcal{D}(\mathcal{H})$ be the true unknown quantum state satisfying Assumption \ref{assump:regularity}. Consider a generic real-valued polynomial functional $f$ with a well-defined gradient such that it is not proportional to the identity matrix at the true state (i.e., $\nabla f(\rho_*) \neq c\mathbb{I}$ for any scalar $c \in \mathbb{R}$). Using a local parameterization with nuisance parameters that satisfy the unbiasedness and orthogonality conditions, the single-copy Cram\'er-Rao bound for estimating $f(\rho)$ evaluates to the intrinsic quantum variance of the gradient of the functional, i.e.,
    \begin{equation*}
        V_{\mathrm{CR}}(\rho_*) = \mathrm{Var}_{\rho_*}\big(\nabla f(\rho_*)\big) = \mathrm{Tr}\big[\rho_* \big(\nabla f(\rho_*)\big)^2\big] - \big(\mathrm{Tr}\big[\rho_* \nabla f(\rho_*)\big]\big)^2\,.
    \end{equation*}
\end{lemma}

\begin{proof}
    To calculate the variance limit for $f(\rho)$ at $\rho = \rho_*$, we define the centered gradient operator as
    \begin{equation*}
        L^1_{f} := \nabla f(\rho_*) - c\mathbb{I},
    \end{equation*}
    where the scalar shift is the expectation value $c = \mathrm{Tr}[\rho_* \nabla f(\rho_*)]$. By design, the expectation value of this centered operator is zero: $\mathrm{Tr}[\rho_* L^1_{f}] = \mathrm{Tr}[\rho_* \nabla f(\rho_*)] - c\mathrm{Tr}[\rho_*] = 0$.
    
    We define the primary tangent operator for our parameterization as
    \begin{equation*}
        X_1 = \frac{1}{\mathrm{Var}_{\rho_*}(\nabla f(\rho_*))} \rho_* \circ L^1_{f}.
    \end{equation*}
    Since $\rho_*$ and $L^1_f$ are Hermitian, $X_1$ is a traceless Hermitian operator. Since Assumption \ref{assump:regularity} ensures $\rho_*$ is positive (full-rank), the tangent space of valid physical perturbations has exactly $d^2-1$ degrees of freedom. This dimensionality guarantees we can choose $d^2-2$ other linearly independent, traceless Hermitian operators $X_j$ (for $j \ge 2$) that are orthogonal to the gradient, meaning $\mathrm{Tr}[X_j \nabla f(\rho_*)] = 0$. We then define the full parameterization of the state space as
    \begin{equation*}
        \rho_\theta := \rho_* + (\theta^1 - \theta_0^1)X_1 + \sum_{j=2}^{d^2-1} \theta^j X_j.
    \end{equation*}
    This ensures that at $\theta = \theta_0$, the parameterization gives the true state, satisfying Condition \eqref{local_framework:eq1}. Crucially, under Assumption \ref{assump:regularity}, any infinitesimal perturbation along the traceless directions $X_j$ avoids negative eigenvalues, ensuring $\rho_\theta$ remains a valid positive semi-definite density matrix locally. Taking the derivative, the tangent vectors are $$\left.\frac{\partial \rho_\theta}{\partial \theta^j}\right|_{\theta=\theta_0} = X_j\,.$$
    
    Using Eq.~\eqref{eq:operator_derivative_chain_rule}, the directional derivative of the polynomial functional $f(\rho)$ around $\rho_*$ simplifies to
    \begin{equation}
        \left.\frac{\partial f(\rho_\theta)}{\partial \theta^j}\right|_{\theta=\theta_0} = \mathrm{Tr}\left[ \left.\frac{\partial \rho_\theta}{\partial \theta^j}\right|_{\theta=\theta_0} \nabla f(\rho_*) \right] = \mathrm{Tr}[X_j \nabla f(\rho_*)].\label{eq:f_par_generic}
    \end{equation}
    
    For the nuisance parameters ($j \ge 2$), our choice of $X_j$ ensures $\mathrm{Tr}[X_j \nabla f(\rho_*)] = 0$. Substituting this into Eq.~\eqref{eq:f_par_generic} gives $0$, satisfying Condition \eqref{local_framework:eq3}.
    
    For the target parameter ($j = 1$), we substitute the definition of $X_1$ into the derivative to get
    \begin{align}
        \left.\frac{\partial f(\rho_\theta)}{\partial \theta^1}\right|_{\theta=\theta_0} &= \mathrm{Tr}[X_1 \nabla f(\rho_*)] \notag\\
        &= \frac{1}{\mathrm{Var}_{\rho_*}(\nabla f(\rho_*))} \mathrm{Tr}\big[ (\rho_* \circ L^1_{f}) \nabla f(\rho_*) \big] \notag\\
        &\overset{(a)}{=} \frac{1}{\mathrm{Var}_{\rho_*}(\nabla f(\rho_*))} \mathrm{Tr}\big[ \rho_* (L^1_{f} \circ \nabla f(\rho_*)) \big] \notag\\
        &\overset{(b)}{=} \frac{1}{\mathrm{Var}_{\rho_*}(\nabla f(\rho_*))} \langle L^1_{f}, \nabla f(\rho_*) \rangle_{\rho_*}\,,\label{eq:f_par_generic_1}
    \end{align}
    where the equality ($a$) uses the trace duality of the Jordan product, $\mathrm{Tr}[(A \circ B) C] = \mathrm{Tr}[A (B \circ C)]$, and step ($b$) applies the definition of the SLD inner product. 
    
    To calculate the numerator in Eq.~\eqref{eq:f_par_generic_1}, we substitute $L^1_{f} = \nabla f(\rho_*) - c\mathbb{I}$ and expand the SLD inner product as
    \begin{align*}
        \langle L^1_{f}, \nabla f(\rho_*) \rangle_{\rho_*} &= \big\langle \nabla f(\rho_*) - c\mathbb{I}, \, \nabla f(\rho_*) \big\rangle_{\rho_*} \\
        &= \mathrm{Tr}\big[ \rho_* \big(\nabla f(\rho_*)\big)^2 \big] - c \, \mathrm{Tr}[ \rho_* (\mathbb{I} \circ \nabla f(\rho_*)) ] \\
        &\overset{(a)}{=} \mathrm{Tr}\big[ \rho_* \big(\nabla f(\rho_*)\big)^2 \big] - \big(\mathrm{Tr}[ \rho_* \nabla f(\rho_*) ]\big)^2 \\
        &= \mathrm{Var}_{\rho_*}(\nabla f(\rho_*)).
    \end{align*}
    where the equality ($a$) uses the definition of the shift $c = \mathrm{Tr}[\rho_* \nabla f(\rho_*)]$. This shows that the inner product equals the variance of the gradient. Substituting this variance back into Eq.~\eqref{eq:f_par_generic_1} cancels the denominator, making the derivative exactly $1$. This proves perfect tracking and satisfies Condition \eqref{local_framework:eq2}.
    
    To connect this parameterization to the SLD Fisher information framework, we write the tangent vectors as $X_j = \rho_* \circ L_j$. By Assumption \ref{assump:regularity}, the positivity of $\rho_*$ ensures the Lyapunov super-operator is invertible, guaranteeing that a unique, bounded Hermitian SLD operator $L_j$ exists for each chosen tangent direction. For the target parameter, this gives
    \begin{equation*}
        L_1 = \frac{1}{\mathrm{Var}_{\rho_*}(\nabla f(\rho_*))} L^1_{f}.
    \end{equation*}
    This defines the primary SLD and satisfies Condition \eqref{local_framework:eq4}.
    
    To separate the target parameter from the nuisance parameters, we must show $\langle L_1, L_j \rangle_{\rho_*} = 0$ for all $j \ge 2$. Expanding the SLD inner product gives
    \begin{align}
        \langle L_1, L_j \rangle_{\rho_*} &= \frac{1}{\mathrm{Var}_{\rho_*}(\nabla f(\rho_*))} \mathrm{Tr}\big[ \rho_* (L^1_{f} \circ L_j) \big] \notag\\
        &\overset{(c)}{=} \frac{1}{\mathrm{Var}_{\rho_*}(\nabla f(\rho_*))} \mathrm{Tr}\big[ (\rho_* \circ L_j) L^1_{f} \big] \notag\\
        &\overset{(d)}{=} \frac{1}{\mathrm{Var}_{\rho_*}(\nabla f(\rho_*))} \mathrm{Tr}\big[ X_j (\nabla f(\rho_*) - c\mathbb{I}) \big] \notag\\
        &= \frac{1}{\mathrm{Var}_{\rho_*}(\nabla f(\rho_*))} \big( \mathrm{Tr}[X_j \nabla f(\rho_*)] - c\mathrm{Tr}[X_j] \big) \notag\\
        &\overset{(e)}{=} \frac{1}{\mathrm{Var}_{\rho_*}(\nabla f(\rho_*))} \mathrm{Tr}[X_j \nabla f(\rho_*)]\,,\label{eq:orthogonal_SLDs_generic}
    \end{align}
    Here, step ($c$) uses the trace duality of the Jordan product. Step ($d$) substitutes $X_j = \rho_* \circ L_j$ and the definition of $L^1_{f}$. Step ($e$) uses the conservation of probability: since density matrices have a unit trace, any tangent vector must be traceless ($\mathrm{Tr}[X_j] = 0$). This eliminates the scalar shift $c$. 
    
    Since we chose $X_j$ so that $\mathrm{Tr}[X_j \nabla f(\rho_*)] = 0$ for $j \ge 2$, Eq.~\eqref{eq:orthogonal_SLDs_generic} becomes zero. Thus, $\langle L_1, L_j \rangle_{\rho_*} = 0$, satisfying the orthogonality requirement of Condition \eqref{local_framework:eq5}.
    
    With all the conditions met, we use Proposition \ref{prop:qcrb_loca} to find the Cram\'er-Rao bound. This is the inverse variance of the primary SLD operator $L_1$. We calculate the inner product
    \begin{align*}
        \langle L_1, L_1 \rangle_{\rho_*} &= \left(\frac{1}{\mathrm{Var}_{\rho_*}(\nabla f(\rho_*))}\right)^2 \langle L^1_{f}, L^1_{f} \rangle_{\rho_*} \\
        &= \frac{1}{\big(\mathrm{Var}_{\rho_*}(\nabla f(\rho_*))\big)^2} \mathrm{Var}_{\rho_*}(\nabla f(\rho_*)) \\
        &= \frac{1}{\mathrm{Var}_{\rho_*}(\nabla f(\rho_*))}.
    \end{align*}
    Inverting this scalar gives the variance of the gradient operator (which is our required fundamental bound) as
    \begin{equation*}
        V_{\mathrm{CR}}(\rho_*) = \langle L_1, L_1 \rangle_{\rho_*}^{-1} = \mathrm{Var}_{\rho_*}(\nabla f(\rho_*))\,.
    \end{equation*}
\end{proof}

\begin{remark}[Extensions to Non-Degenerate, Higher-Order Estimation Geometries]
    The condition $\nabla f(\rho_*) \neq c\mathbb{I}$ ensures that the first-order estimation geometry is non-degenerate. If the gradient evaluates to a scalar multiple of the identity matrix, then its variance vanishes (i.e., $V_{\mathrm{CR}}(\rho_*) = c^2 - c^2 = 0$). Geometrically, this corresponds to a stationary point of the functional, such as evaluating a quantum divergence between two identical states. At this point, the functional experiences no first-order change in any physical direction, rendering the standard first-order Quantum Cram\'er-Rao bound uninformative. To determine the fundamental estimation limits in this specific regime, one must rely on second-order (or higher-order) Quantum Cram\'er-Rao bounds. These bounds capture the curvature of the state space using the second-order Symmetric Logarithmic Derivative, which shifts the variance scaling from $\mathcal{O}(1/n)$ to $\mathcal{O}(1/n^2)$.
\end{remark}

\subsection{Asymptotic Efficiency}

The main result is that a permutation-invariant observable $O_n^{\mathrm{sym}}$ (acting on $\rho^{\otimes n}$ and unbiasedly estimating $f(\rho)$) is \emph{asymptotically efficient}. In other words, it extracts the maximum amount of information per copy of $\rho$ allowed by quantum mechanics.

To show this, we write the Cram\'er-Rao bound for $n$ independent copies (which is the single-copy bound divided by $n$) as $V_{\mathrm{CR}}^{(n)} = V_{\mathrm{CR}}(\rho_*)\big/n = \mathrm{Var}_{\rho_*}\big(\nabla f(\rho_*)\big)\big/n\,.$

This variance limit highlights a practical challenge. Achieving the Cram\'er-Rao bound requires measuring in the eigenbasis of the primary SLD operator $L_1$. Since $L_1 \propto \nabla f(\rho_*) - c\mathbb{I}$, the optimal local measurement must be projected onto the eigenbasis of the functional's gradient $\nabla f(\rho_*)$.

However, this creates an experimental circularity. Since $\nabla f(\rho_*)$ depends on the unknown true state $\rho_*$, building the optimal measurement setup requires knowing the state before actually measuring it. In practice, this forces the use of adaptive estimation protocols. For example, an experimenter must consume a portion of the $n$ copies for initial state tomography to approximate $\rho_*$, and then use this estimate to calibrate the measurement hardware for the remaining copies.

Our unbiased, global, permutation-invariant estimator $O_n^{\mathrm{sym}}$ avoids this issue entirely. Since $O_n^{\mathrm{sym}}$ is constructed directly from the algebraic definition of $f$, it is completely state-independent. It evaluates the functional without requiring prior knowledge of $\rho_*$, initial tomography, or adaptive measurements.

We now compare this to the variance of our permutation-invariant observable. According to Eq.~\eqref{eq:hoeffding_decomp_var}, the variance of any unbiased permutation-invariant estimator $O_n^{\mathrm{sym}}$ on $\rho^{\otimes n}$ can be decomposed such that the leading-order term is determined by its first-order marginal kernel, which corresponds to the functional gradient $\nabla f(\rho_*)$. Thus, the exact variance is
\begin{align*}
    \mathrm{Var}_{\rho_*^{\otimes n}}\big(O_n^{\mathrm{sym}}\big) &= \frac{\mathrm{Var}_{\rho_*}\big(\nabla f(\rho_*)\big)}{n} + \mathcal{O}\left(\frac{1}{n^2}\right)\\
    &= V_{\mathrm{CR}}^{(n)} + \mathcal{O}\left(\frac{1}{n^2}\right)\,.
\end{align*}

As the sample size $n \to \infty$, the higher-order $\mathcal{O}(1/n^2)$ terms vanish much faster than the leading $\mathcal{O}(1/n)$ term. Taking the limit gives
\begin{equation*}
    \lim_{n \to \infty} n \mathrm{Var}_{\rho_*^{\otimes n}}\big(O_n^{\mathrm{sym}}\big) = \mathrm{Var}_{\rho_*}\big(\nabla f(\rho_*)\big) = n V_{\mathrm{CR}}^{(n)}.
\end{equation*}
This shows that the permutation-invariant estimator $O_n^{\mathrm{sym}}$ achieves the quantum Cram\'er-Rao bound, and is thereby asymptotically efficient.

\begin{remark}[State-Independent, Multi-Copy Observables vs. State-Dependent, Single-Copy Measurements]
    The higher-order $\mathcal{O}(1/n^2)$ term in the variance expansion highlights a key structural difference between local and global estimation for non-linear functionals. The local Cram\'er-Rao bound assumes an idealized single-copy observable, i.e., the SLD operator, perfectly tuned to the unknown state $\rho_*$. The variance of $n$ independent single-copy measurements scales exactly as $1/n$, with no higher-order terms.

    However, for a non-linear functional $f(\rho)$, no single-copy observable can be globally unbiased. To estimate $f(\rho)$ without prior knowledge of the state, a globally unbiased estimator must use collective, multi-copy measurements (such as SWAP operators acting across different copies). In the Hoeffding decomposition, the independent first-order marginal kernel achieves the fundamental Cram\'er-Rao limit, while the remaining multi-copy interactions generate the $\mathcal{O}(1/n^2)$ variance terms.

    Therefore, this higher-order term represents a finite-sample ``cost of universality''. It is the statistical penalty for using a state-independent, multi-copy observable instead of a perfectly tuned single-copy measurement. Since this penalty vanishes as $n \to \infty$, the global permutation-invariant estimator provides an advantage, as it guarantees robustness and state-independence at finite-sample sizes while achieving optimal local efficiency in the asymptotic limit.
\end{remark}

Having established the asymptotic Cram\'er-Rao efficiency of quantum U-statistic estimators, we now explore several practical applications and theoretical discussions related to our presented variance analysis framework. In the subsequent subsection, we demonstrate the versatility of this framework by applying it to existing estimation paradigms and estimation of a few fundamental quantum information-theoretic quantities, with a particular focus to the Bures $\chi^2$-divergence (see Subsection \ref{subsec:app_bures_chi}).

\subsection{Discussion About Variance Bound for Estimating Relative Entropy \cite{Hayashi_2025}}

Another interesting problem on estimating a functional was studied by Hayashi in \cite{Hayashi_2025}, where there is an unknown quantum state $\rho$ and a fully known reference state $\sigma$, and the objective is to estimate the relative entropy $D(\rho||\sigma)$ defined as
\begin{equation*}
    D(\rho||\sigma) := \mathrm{Tr}\big[\rho(\log\rho-\log \sigma)\big]\,.
\end{equation*}

For this specific divergence, the gradient with respect to the unknown state evaluates to $\nabla D(\rho||\sigma) = (\log \rho - \log \sigma + \mathbb{I})$. To solve this estimation problem, the work constructs a global $n$-copy observable to estimate $D(\rho||\sigma)$, effectively serving as a surrogate for the ideal functional observable $(\log \rho - \log \sigma)$.

The performance of this estimator is bounded in \cite[Theorem 2]{Hayashi_2025}, which establishes that for any states $\rho$ and $\sigma$, we have
\begin{equation*}
    \mathrm{MSE}_n(\rho||\sigma) \le \left( \frac{1}{\sqrt{n}} \sqrt{V(\rho||\sigma)} + \frac{1}{n}\log d_{n,d} \right)^2\,.
\end{equation*}

In this theorem, $\mathrm{MSE}_n(\rho||\sigma)$ basically represents the variance (mean squared error) of the global $n$-copy observable. The term $V(\rho||\sigma)$, known as the relative varentropy \cite{Tomamichel_2013, Li_2014}, corresponds exactly to the variance of the functional gradient. Since shifting an operator by a scalar multiple of the identity matrix does not alter its variance, the relative varentropy evaluates to
\begin{equation*}
    V(\rho||\sigma) := \mathrm{Var}_{\rho}(\nabla D(\rho||\sigma)) = \mathrm{Tr}\left[\rho\big(\log\rho - \log\sigma - D(\rho||\sigma)\mathbb{I}\big)^2\right]\,.
\end{equation*}

In light of the Cram\'er-Rao optimality subsequently established by the author for the proposed global $n$-copy estimator, the variance analysis demonstrates a structural correspondence to Eq.~\eqref{eq:hoeffding_decomp_var} of our manuscript. Specifically, it underscores that the first-order asymptotic term of the global estimator's variance scales precisely with the relative varentropy $V(\rho||\sigma)$, which is identically the intrinsic state variance of the target functional's gradient.

\subsection{Application I - Bures $\chi^2$-Divergence}\label{subsec:app_bures_chi}

As a non-trivial application of Eq.~\eqref{eq:hoeffding_decomp_var}, we introduce the Bures $\chi^2$-divergence between two quantum states and analyse the statistical variance inherent in estimating this quantity for an unknown density operator $\rho$ with respect to a fully known reference state $\sigma$.

\begin{definition}[Bures $\chi^2$-Divergence \cite{Bures1969A,Uhlmann1976}]\label{def:bures}
    Given two density operators $\rho, \sigma\in \mathcal{D}(\mathcal{H})$, the Bures $\chi^2$-divergence between the states $\rho$ and $\sigma$ is expressed as 
    \begin{align*}
        \chi^2_{\mathrm{B}}(\rho\|\sigma) &:= \mathrm{Tr}\left[(\rho-\sigma) \Omega_\sigma(\rho-\sigma)\right] = \mathrm{Tr}[\rho \Omega_{\sigma}(\rho)] - 1\,,
    \end{align*}
    where $\Omega_\sigma$ is a super-operator whose inverse is given by
    \begin{equation*}
        \Omega_\sigma^{-1}(X)=\frac{\sigma X+X \sigma}{2}\,.
    \end{equation*}
\end{definition}

Let us derive an analytically tractable expression for the super-operator $\Omega_{\sigma}(\rho)$. To that end, we invoke the following proposition concerning the integral solution of the continuous-time Lyapunov equation in order to establish an integral representation for the Bures $\chi^2$-divergence between two quantum states.

\begin{proposition}[Continuous-Time Lyapunov Solution {\cite[Theorem 9.4]{BR97}, \cite{Heinz1951}}]
\label{prop:lyapunov}
    Let $A$ be a Hurwitz stable matrix (i.e., all eigenvalues of $A$ have strictly negative real parts) and let $Q$ be a Hermitian matrix. The unique solution $X$ to the continuous-time Lyapunov equation $AX + XA^\dagger + Q = 0$ is given by the following convergent continuous-time integral,
    \begin{equation*}
        X = \int_0^\infty e^{\tau A} Q e^{\tau A^\dagger} d\tau.
    \end{equation*}
\end{proposition}

Leveraging the preceding Proposition, the integral representation for the Bures $\chi^2$-divergence is as given in the following lemma.

\begin{lemma}[Integral Representation of Bures $\chi^2$-Divergence]\label{lem:integral_measured_chi2}
    Given two density operators $\rho, \sigma\in \mathcal{D}(\mathcal{H})$, where $\rho$ is fully unknown whereas $\sigma$ is a fully known reference operator, the Bures $\chi^2$-divergence $\chi^2_{\mathrm{B}}(\rho \| \sigma)$ between the two states admits the following continuous-time integral representation,
    \begin{equation*}
       \chi^2_{\mathrm{B}}(\rho \| \sigma) = 2 \int_0^\infty \mathrm{Tr}[(\rho e^{-\tau\sigma})^2] d\tau - 1.
    \end{equation*}
\end{lemma}

\begin{proof}
    See Appendix \ref{subsec:integral_measured_chi2} for the proof.
\end{proof}

Building upon this novel continuous-time integral representation, we can evaluate the integral directly in the eigenbasis of the reference state $\sigma$ to recover the well-known spectral form of the inverse symmetric logarithmic derivative super-operator. This explicit spectral expression has been extensively studied in the foundational literature of quantum estimation theory and quantum information geometry \cite{braunstein1994statistical, hayashi2006quantum, BOW19}.

\begin{corollary}[Spectral Expression of $\Omega_{\sigma}(\rho)$ \cite{braunstein1994statistical, hayashi2006quantum, BOW19}]\label{cor:spectral_omega}
    Suppose the density operator $\sigma$ admits the spectral decomposition $\sigma = \sum\limits_{i=1}^d \lambda_i |i\rangle \langle i|$. Then, the super-operator $\Omega_{\sigma}(\rho)$ can be explicitly expressed in the eigenbasis of $\sigma$ as
    \begin{equation}
        \Omega_{\sigma}(\rho) = \sum_{j,k=1}^d \frac{2}{\lambda_j + \lambda_k} |j\rangle \langle j| \rho |k\rangle \langle k|\,.
    \end{equation}
\end{corollary}

\begin{proof}
    See Appendix \ref{subsec:spectral_omega} for the proof.
\end{proof}

Now, suppose there exists a $n$-copy permutation-invariant observable $O^{\mathrm{sym}}_{n,\chi^2_{\mathrm{B}}}$ estimating the Bures $\chi^2$-divergence between the two quantum states $\rho$ and $\sigma$ (\cite{BOW19} presents such an estimator based on quantum U-statistics). Then, as per Eq.~\eqref{eq:hoeffding_decomp_var}, the variance of the estimator exactly scales as 
\begin{equation}\label{eq:bures_variance}
    \mathrm{Var}_{\rho^{\otimes n}}\left(O^{\mathrm{sym}}_{n,\chi^2_{\mathrm{B}}}\right) = \frac{1}{n} \mathrm{Var}_{\rho}\big(\nabla \chi^2_{\mathrm{B}}(\rho||\sigma)\big) + \mathcal{O}\left(\frac{1}{n^2}\right)\,.
\end{equation}

Considering a Hermitian perturbation $\delta\rho$, and exploiting the linearity of the super-operator $\Omega_\sigma$, the first-order directional variation of the functional comes out to be 
\begin{equation}\label{eq:first_order_frechet}
    \delta f = \mathrm{Tr}[\delta\rho \, \Omega_\sigma(\rho)] + \mathrm{Tr}[\rho \, \Omega_\sigma(\delta\rho)]\,.
\end{equation}

To combine the two terms on the right-hand side of Eq.~\eqref{eq:first_order_frechet}, we first need to observe that $\Omega_\sigma$ is self-adjoint under the Hilbert-Schmidt inner product. Since for any Hermitian operators $A$ and $B$, assuming $X = \Omega_\sigma(A)$ and $Y = \Omega_\sigma(B)$, by the definition of the super-operator $\Omega_\sigma$, we have $A = \frac{1}{2}(\sigma X + X\sigma)$ and $B = \frac{1}{2}(\sigma Y + Y\sigma)$. Using the cyclic property of the trace, we have
\begin{align*}
    \mathrm{Tr}[A\Omega_\sigma(B)] &= \frac{1}{2}\mathrm{Tr}\big[(\sigma X + X\sigma)Y\big] \\
    &= \frac{1}{2}\mathrm{Tr}\big[X(Y\sigma + \sigma Y)\big] \\
    &= \mathrm{Tr}[X B] \\
    &= \mathrm{Tr}[\Omega_\sigma(A)B]\,.
\end{align*}

Setting $A = \rho$ and $B = \delta\rho$, and applying this self-adjointness to the second term of our variation yields $\mathrm{Tr}[\rho \, \Omega_\sigma(\delta\rho)] = \mathrm{Tr}[\Omega_\sigma(\rho)\delta\rho]$. Since the trace of a product of two Hermitian matrices is commutative, the total first-order variation simplifies to $\delta f = 2\mathrm{Tr}[\Omega_\sigma(\rho)\delta\rho]$.

By the standard definition of the matrix gradient $\delta f = \langle \nabla f(\rho), \delta\rho \rangle_{\mathrm{HS}} = \mathrm{Tr}[\nabla f(\rho)\delta\rho]$, we immediately identify the gradient of the functional as $\nabla \chi^2_{\mathrm{B}}(\rho \Vert{} \sigma) = 2\,\Omega_\sigma(\rho)\,.$ Substituting the value of the functional gradient in Eq.~\eqref{eq:bures_variance}, we see that the variance of the estimator $O^{\mathrm{sym}}_{n,\chi^2_{\mathrm{B}}}$ scales exactly as 
\begin{align*}
     \mathrm{Var}_{\rho^{\otimes n}}\left(O^{\mathrm{sym}}_{n,\chi^2_{\mathrm{B}}}\right) &= \frac{1}{n} \mathrm{Var}_{\rho}\big(2\,\Omega_\sigma(\rho)\big) + \mathcal{O}\left(\frac{1}{n^2}\right)\\
      &= \frac{4}{n} \mathrm{Var}_{\rho}\big(\Omega_\sigma(\rho)\big) + \mathcal{O}\left(\frac{1}{n^2}\right)\,.
\end{align*}

We note that the minimum singular value assumption $\lambda_{\min}(\sigma) \ge \delta > 0$ introduced by \cite{BOW19} to guarantee bounded estimator variance is sufficient, but not necessary. Boundedness of the first-order variance term can hold even when $\lambda_{\min}(\sigma)$ becomes arbitrarily small, as shown by the following theorem.

\begin{theorem}[Sufficiency and Non-Necessity of Spectral Bounds]\label{thm:variance_spectral_bound}
    Let $\rho, \sigma \in \mathcal{D}(\mathcal{H})$. If $\lambda_{\min}(\sigma) \ge \delta$, for a fixed $\delta > 0$ then both $\chi^2_{\mathrm{B}}(\rho\Vert{}\sigma)$ and $\mathrm{Var}_{\rho}\big(\nabla \chi^2_{\mathrm{B}}(\rho\Vert{}\sigma)\big)$ are uniformly bounded. However, the converse does not hold, i.e., bounded first-order variance and/or $\chi^2_{\mathrm{B}}(\rho\Vert{}\sigma)$ does not imply $\lambda_{\min}(\sigma) \ge \delta$.
\end{theorem}

\begin{proof}
    For the forward implication, $\lambda_{\min}(\sigma) \ge \delta > 0$ implies $\Vert{}\Omega_\sigma(\rho)\Vert{}_\infty \le 1/\delta$. Bounding the second moment via the spectral norm yields 
    \[\mathrm{Var}_\rho\big(\Omega_\sigma(\rho)\big) \le \mathrm{Tr}\left[\rho \Omega_\sigma(\rho)^2\right] \le \Vert{}\Omega_\sigma(\rho)\Vert{}_\infty^2 \le \frac{1}{\delta^2}\,,\]
    which directly bounds the gradient variance as $\mathrm{Var}_\rho\big(\nabla \chi^2_{\mathrm{B}}(\rho\Vert{}\sigma)\big) = 4\mathrm{Var}_\rho\big(\Omega_\sigma(\rho)\big) \le 4/\delta^2$. 
    
    To show the non-necessity of the minimum singular value assumption of \cite{BOW19}, let us construct the parametric family of $2 \times 2$ density operators, for $n \ge 2$, as 
    \[\rho_n = \begin{pmatrix} 1-\frac{1}{n^{2/3}} & 1/n \\ 1/n & \frac{1}{n^{2/3}} \end{pmatrix}\,,\quad \text{and}\quad \sigma_n = \begin{pmatrix} 1-1/n & 0 \\ 0 & 1/n \end{pmatrix}\,.\]
    Note that $\lim_{n\to\infty} \lambda_{\min}(\sigma_n) = 0\,.$ Now, via direct substitutions into Lemma \ref{lem:integral_measured_chi2}, we have 
    \[\Omega_{\sigma_n}(\rho_n) = \begin{pmatrix} \frac{1-\frac{1}{n^{2/3}}}{1-1/n} & 1/n \\ 1/n & {n^{1/3}} \end{pmatrix}.\]
    Evaluating the divergence and variance explicitly gives us \[\chi^2_{\mathrm{B}}(\rho_n\Vert{}\sigma_n) = \frac{\left(1 - \frac{1}{n^{2/3}}\right)^3}{\left(1 - 1/n\right)^2} + \frac{1}{n^2} + \frac{2}{n^2}\left(\frac{1 - \frac{1}{n^{2/3}}}{1 - 1/n}\right) + \frac{2}{n^{5/3}}\,,\] 
    \begin{align*}
        \text{and, }\qquad
    \mathrm{Var}_{\rho_n}\big(\Omega_{\sigma_n}(\rho_n)\big) &= \frac{n^{-2/3}\left(1 - n^{-2/3}\right)^3}{\left(1 - 1/n\right)^2} + 1 - n^{-2/3} + n^{-2} - 4n^{-4} \\& + 2n^{-2} \frac{\left(1 - n^{-2/3}\right)\left(2n^{-2/3} - 1\right)}{1 - 1/n} + 2n^{-5/3} - 4n^{-7/3} - 2n^{-1/3} \frac{\left(1 - n^{-2/3}\right)^2}{1 - 1/n}\,.
    \end{align*}
    
    Therefore, $\lim_{n\to\infty}\lambda_{\min}(\sigma_n) = 0$, while $\lim_{n\to\infty} \chi^2_{\mathrm{B}}(\rho_n\Vert{}\sigma_n)=\lim_{n\to\infty}\mathrm{Var}_{\rho_n}\big(\Omega_{\sigma_n}(\rho_n)\big)=1$. Thus, both terms remain uniformly bounded despite $\lim_{n\to\infty} \lambda_{\mathrm{min}}(\sigma_n)= 0$.
\end{proof}

Consequently, assuming the boundedness of $\mathrm{Var}_{\rho}\big(\nabla \chi^2_{\mathrm{B}}(\rho\Vert{}\sigma)\big)$ provides a weaker and more natural condition for statistical analysis than requiring a uniform lower bound on $\lambda_{\min}(\sigma)$. In fact, as discussed in section \ref{sec:variance_optimality}, the boundedness of $\mathrm{Var}_{\rho}\big(\nabla \chi^2_{\mathrm{B}}(\rho\Vert{}\sigma)\big)$ is required solely to ensure a convergent first-order QFI, and thereby to exhibit the first-order Cram\'er-Rao optimality of our analysis. Nevertheless, also pointed out in Corollary \ref{cor:degenerate_variance} of Appendix \ref{app:proof_grad_marg}, our analysis can be naturally extended to account for vanishing variances up to any $k^{\text{th}}$-order tensor gradient (and a non-vanishing variance for the $(k+1)^{\text{th}}$-order tensor gradient) of the functional, in which case we need to study the $(k+1)^{\text{th}}$-order convergent QFI, and the $(k+1)^{\text{th}}$-order Cram\'er-Rao optimality of our analysis. In the remaining case where the variance of every $k^{\text{th}}$-order tensor gradient of the functional vanishes, the estimation problem itself becomes fundamentally non-viable.

\subsection{Application II - Purity}

Let $\rho \in \mathcal{D}(\mathcal{H})$. We define the \emph{purity} of the density operator $\rho$ as $$\operatorname{Pur}(\rho):=\tr(\rho^2)\,.$$

Similar to the variance analysis before, suppose there exists an $n$-copy permutation-invariant observable $O^{\mathrm{sym}}_{n,\mathrm{Pur}}$ estimating the purity $\mathrm{Pur(\rho)}$ of a given quantum state $\rho$ (refer to the discussion in Subsection \ref{subsec:homo_monomial}). Then, as per Eq.~\eqref{eq:hoeffding_decomp_var}, the variance of the estimator exactly scales as 
\begin{equation}\label{eq:purity_variance}
    \mathrm{Var}_{\rho^{\otimes n}}\left(O^{\mathrm{sym}}_{n,\mathrm{Pur}}\right) = \frac{1}{n} \mathrm{Var}_{\rho}\big(\nabla \mathrm{Pur}(\rho)\big) + \mathcal{O}\left(\frac{1}{n^2}\right)\,.
\end{equation}

Let us compute the gradient of the purity functional $\mathrm{Pur}(\rho)$. Considering a Hermitian perturbation $\delta\rho$, we expand the functional to extract its first-order directional variation as
\begin{align*}
    f(\rho + \delta\rho) &= \tr\big[(\rho + \delta\rho)^2\big] \\
    &= \tr\big[\rho^2 + \rho(\delta\rho) + (\delta\rho)\rho + (\delta\rho)^2\big]\,.
\end{align*}

By discarding the quadratic term $\tr\big[(\delta\rho)^2\big]$, which corresponds to $\mathcal{O}(\Vert{}\delta\rho\Vert{}^2)$, the first-order variation isolates to 
\begin{equation*}
    \delta f = \tr[\rho(\delta\rho)] + \tr[(\delta\rho)\rho]\,.
\end{equation*}

Exploiting the cyclic property of the trace, we observe that $\tr[(\delta\rho)\rho] = \tr[\rho(\delta\rho)]$. Consequently, the two terms combine seamlessly, simplifying the total first-order variation to $\delta f = 2\tr[\rho(\delta\rho)]$. By the standard definition of the matrix gradient under the Hilbert-Schmidt inner product $\delta f = \langle \nabla f(\rho), \delta\rho \rangle_{\mathrm{HS}} = \tr[\nabla f(\rho)\delta\rho]$, we immediately identify the gradient of the purity functional as $\nabla \operatorname{Pur}(\rho) = 2\rho$. Substituting the value of the functional gradient in Eq.~\eqref{eq:purity_variance}, we see that the variance of the estimator $O^{\mathrm{sym}}_{n,\mathrm{Pur}}$ scales exactly as 
\begin{align*}
     \mathrm{Var}_{\rho^{\otimes n}}\left(O^{\mathrm{sym}}_{n,\mathrm{Pur}}\right) &= \frac{1}{n} \mathrm{Var}_{\rho}(2\rho) + \mathcal{O}\left(\frac{1}{n^2}\right)\\
      &= \frac{4}{n} \mathrm{Var}_{\rho}(\rho) + \mathcal{O}\left(\frac{1}{n^2}\right)\,.
\end{align*}

\subsection{Application III - (Squared) Hilbert-Schmidt Distance}

\begin{definition}\label{def:hil_smith}
    Given two density operators $\rho, \sigma\in \mathcal{D}(\mathcal{H})$, the Hilbert-Schmidt distance $D_{\mathrm{HS}}(\rho, \sigma)$ between the two density operators $\rho$ and $\sigma$ is the Schatten 2-norm (also called the Hilbert-Schmidt norm) of their difference, that is \[D_{\mathrm{HS}}(\rho, \sigma) := \Vert{}\rho - \sigma\Vert{}_{\mathrm{HS}} = \sqrt{\mathrm{Tr}\big[(\rho - \sigma)^2\big]}\,.\]
\end{definition}

Again, assuming the existence of a $n$-copy permutation-invariant observable $O^{\mathrm{sym}}_{n,\mathrm{HS}^2}$ estimating the (squared) Hilbert-Schmidt distance $D^2_{\mathrm{HS}}(\rho, \sigma)$ between an unknown density operator $\rho$, and a fully-known reference density operator $\sigma$, as a direct consequence of Eq.~\eqref{eq:hoeffding_decomp_var}, the variance of the estimator scales as 
\begin{align*}
    \mathrm{Var}_{\rho^{\otimes n}}\left(O^{\mathrm{sym}}_{n,\mathrm{HS}^2}\right) &= \frac{1}{n} \mathrm{Var}_{\rho}\big(\nabla D^2_{\mathrm{HS}}(\rho, \sigma)\big) + \mathcal{O}\left(\frac{1}{n^2}\right)\\
    &=\frac{4}{n} \mathrm{Var}_{\rho}(\rho-\sigma) + \mathcal{O}\left(\frac{1}{n^2}\right)\,.
\end{align*}

\begin{table}[h]
    \centering
    \resizebox{1.0\textwidth}{!}{
    \renewcommand{\arraystretch}{3.0}
    \begin{tabular}{@{}|l|l|l|l|@{}}
        \hline
        \hline
        \thead{\textbf{Quantum Measure / Divergence} \\ \textbf{(Assuming Full-Rank Reference State $\sigma$)}} & \thead{\textbf{Functional Gradient} \\ $\nabla f(\rho)$} & \thead{\textbf{Permutation-Invariant (Local) Kernel} \\ $O^{\mathrm{sym}}_k$} & \thead{\textbf{First-Order Marginal Kernel} \\ $O^{\mathrm{sym}}_{k,1}$} \\
        \hline
        \hline
        \makecell[l]{Hilbert-Schmidt Inner Product\\ $\big(f = \mathrm{Tr}[\rho\sigma]\big)$} & $\sigma$ & $\sigma$ & $\sigma$ \\
        \hline
        \makecell[l]{(Squared) Hilbert-Schmidt Distance \\ $(f = \mathrm{Tr}[(\rho-\sigma)^2])$} & $2(\rho - \sigma)$ & $\mathbb{S} - \sigma \otimes \mathbb{I} - \mathbb{I} \otimes \sigma + \mathrm{Tr}\big[\sigma^2\big]\mathbb{I}^{\otimes 2}$ & $\rho - \sigma + \left(\mathrm{Tr}\big[\sigma^2\big] - \mathrm{Tr}[\rho\sigma]\right)\mathbb{I}$ \\
        \hline
        Purity $\left(f = \mathrm{Tr}\big[\rho^2\big]\right)$ & $2\rho$ & $\mathbb{S}$ & $\rho$ \\
        \hline
        \makecell[l]{$k^{\text{th}}$ Order Purity / (Exponential) Integer \\ Rényi-$k$ Entropy $\left(f = g^{1-k}=\mathrm{Tr}\big[\rho^k\big]\right)$} & $k\rho^{k-1}$ & $\dfrac{1}{k!}\displaystyle \sum\limits_{\pi \in S_k}P_\pi P_{\gamma_k} P_\pi^\dagger$ & $\rho^{k-1}$ \\
        \hline
        Bures $\chi^2$-Divergence & $2\,\Omega_\sigma(\rho)$ & $\mathbb{S}_\sigma - \mathbb{I}^{\otimes 2}$ & $\Omega_\sigma(\rho) - \mathbb{I}$ \\
        \hline
        \makecell[l]{Maximal $\chi^2$-Divergence\\ $\left(f=\mathrm{Tr}\left[\rho^2\sigma^{-1}\right]-1\right)$} & $\rho\sigma^{-1} + \sigma^{-1}\rho$ & $\dfrac{1}{2}\left(\mathbb{S}(\mathbb{I} \otimes \sigma^{-1}) + (\mathbb{I} \otimes \sigma^{-1})\mathbb{S}\right) - \mathbb{I}^{\otimes 2}$ & $\dfrac{1}{2}\left(\rho\sigma^{-1} + \sigma^{-1}\rho\right) - \mathbb{I}$ \\
        \hline
        \makecell[l]{Sandwiched $\chi^2$-Divergence \\ $\left(f=\mathrm{Tr}\left[ \rho \sigma^{-1/2} \rho \sigma^{-1/2} \right] - 1\right)$} & $2\sigma^{-1/2}\rho\sigma^{-1/2}$ & $\mathbb{S}\left(\sigma^{-1/2} \otimes \sigma^{-1/2}\right) - \mathbb{I}^{\otimes 2}$ & $\sigma^{-1/2}\rho\sigma^{-1/2} - \mathbb{I}$\\
        \hline
        \hline
    \end{tabular}
    }
    \caption{Functional Gradients, permutation-invariant (local) kernels, and their corresponding first-order marginal kernels for a few fundamental quantum polynomial functionals.}
    \label{tab:quantum_functional_kernels_hoeffding}
\end{table}

Table \ref{tab:quantum_functional_kernels_hoeffding} catalogues the functional gradients, permutation-invariant local kernels, and their corresponding first-order marginal kernels for a few fundamental quantum information theoretic polynomial functionals. In the table, the notation $\mathbb{S}$ denotes the two-copy $\mathrm{SWAP}$ operator $\mathbb{S} = \sum_{i,j} \vert{}j,i\rangle\langle i,j\vert{}\,,$ and $P_{\gamma_m}$ is the $m$-copy backward cyclic shift operator defined as 
\begin{equation*}
    P_{\gamma_m} := \sum\limits_{i_1, i_2, \dots, i_m} \vert{}i_2, i_3 \dots, i_{m-1}, i_m, i_1\rangle \langle i_1, i_2, i_3, \dots, i_m\vert{}\,.
\end{equation*}
Also, assuming the reference state $\sigma$ to be full rank, the $\sigma$-weighted $\mathrm{SWAP}$ operator $\mathbb{S}_\sigma$ is defined as a two-copy observable acting on $\mathcal{H}^{\otimes 2}$ that generalizes the 2-copy $\mathrm{SWAP}$ operator by incorporating the spectral geometry of $\sigma$ as 
\begin{equation*}
    \mathbb{S}_\sigma := \sum_{i,j} \frac{2}{\lambda_i + \lambda_j} \vert{}j, i\rangle\langle i, j\vert{}\,,
\end{equation*}
where $\sigma = \sum\limits_i \lambda_i \vert{}i\rangle\langle i\vert{}$ denotes the spectral decomposition of the operator $\sigma$.

\section{Concentration Bounds and the MDP}
\label{sec:concentration}

While the asymptotic Cram\'er--Rao efficiency established in Section \ref{sec:variance_optimality} characterizes the of quantum U-statistics in the large-sample limit ($n \to \infty$), practical quantum parameter estimation often operates under strict finite-sample constraints. In such regimes, merely characterizing the variance is insufficient; high-probability tail bounds are required to guarantee reliability, along with a clear characterization of the finite-sample complexity. The rest of the manuscript addresses these issues.

\subsection{Variance-Sensitive Exponential Concentration Bounds}

Standard concentration inequalities often rely solely on the spectral norm of the observable, leading to loose bounds that fail to capture the true statistical fluctuations of the state. To rectify this, we establish an exponential two-sided probability concentration bound that retains the finite-$n$ variance term of the quantum U-statistic in its quadratic term. The ensuing argument follows the general moment-expansion philosophy utilized for local quantum observables, but crucially exploits the specific combinatorial structure of permutation-invariant U-statistic kernels. Through this approach, the second-order variance contribution is cleanly isolated, while the intractable higher-order contributions are systematically controlled.

To formalize this setting, let $\rho\in\mathcal{D}(\mathcal{H})$ be an unknown state, and $O_{k}^{\text{sym}}\in\mathcal{B}(\mathcal{H}^{\otimes k})$ be a permutation-invariant $k$-copy local kernel corresponding to the real-valued polynomial functional $f$. Since by definition, the kernel is an unbiased estimator, we have 
\begin{equation}\label{eq:centred_local_operator}
    \mathrm{Tr}\left[\rho^{\otimes n}\big(O^{\text{sym}}_k-f(\rho)\mathbb{I}^{\otimes k}\big)_{\Lambda}\right]=0\,, \quad \forall\Lambda\in\binom{[n]}{k}\,,
\end{equation} 
where $\big(O^{\text{sym}}_k-f(\rho)\mathbb{I}^{\otimes k}\big)_\Lambda$ denotes an $n$-copy operator that acts as the $k$-copy local centred operator $\big(O^{\text{sym}}_k-f(\rho)\mathbb{I}^{\otimes k}\big)$ on all the copies $\rho_{(i)}$ such that $i\in \Lambda$, and as the identity operator on all of the remaining copies $\rho_{(j)}$ such that $j\in [n]\backslash \Lambda$. Let us define the centred quantum U-statistic $\widetilde{U}_{n,k}$ as 
\begin{equation}\label{eq:centred_u_stat_definition}
    \widetilde{U}_{n,k}:=\frac{1}{\binom{n}{k}}\sum\limits_{\Lambda\in \binom{[n]}{k}}\big(O^{\text{sym}}_k-f(\rho)\mathbb{I}^{\otimes k}\big)_{\Lambda}\,.
\end{equation}

We observe that the quantum variance of the centred quantum U-statistic $\widetilde{U}_{n,k}$ is the same as that of the original quantum U-statistic $U_{n,k}$. In other words, analysing the variance of the original quantum U-statistic is equivalent to analysing the second-order moment of the centred quantum U-statistic, i.e., 
\begin{equation}\label{eq:var_equivalence}
    \mathrm{Var}_{\rho^{\otimes n}}\left(U_{n,k}\right)=\mathrm{Tr}\left[\rho^{\otimes n}\widetilde{U}_{n,k}^2\right]=\mathrm{Var}_{\rho^{\otimes }}\left(\widetilde{U}_{n,k}\right)\,.
\end{equation}

Analyzing the second-order moment of $\widetilde{U}_{n,k}$ requires evaluating the covariances between the $k$-copy local centred operators $\big(O^{\text{sym}}_k-f(\rho)\mathbb{I}^{\otimes k}\big)$ acting on different subsets. If two subsets share no common tensor factors, their operators act independently on the product state, yielding zero covariance. When they do intersect, their covariance contributes to the total variance. The following lemma guarantees that these pairwise overlap contributions are non-negative, a vital property for preserving the exact variance in our concentration bounds.

\begin{lemma}[Positivity of Pair Covariances]\label{lemma:pair_covariance_positivity}
    For any subset of indices $\Lambda,\Gamma\in\binom{[n]}{k}$, the $n$-copy operators $\big(O^{\text{sym}}_k-f(\rho)\mathbb{I}^{\otimes k}\big)_{\Lambda}$ and $\big(O^{\text{sym}}_k-f(\rho)\mathbb{I}^{\otimes k}\big)_{\Gamma}$ exhibit non-negative covariance, i.e., 
    \begin{equation*}
        \mathrm{Tr}\left[\rho^{\otimes n} \big(O^{\text{sym}}_k-f(\rho)\mathbb{I}^{\otimes k}\big)_{\Lambda} \big(O^{\text{sym}}_k-f(\rho)\mathbb{I}^{\otimes k}\big)_{\Gamma}\right]\ge 0\,.
    \end{equation*} 
    In particular, if the subsets are strictly disjoint (i.e., $\Lambda\cap\Gamma=\emptyset$), then their covariance vanishes, i.e., 
    \begin{equation*}
        \mathrm{Tr}\left[\rho^{\otimes n} \big(O^{\text{sym}}_k-f(\rho)\mathbb{I}^{\otimes k}\big)_{\Lambda} \big(O^{\text{sym}}_k-f(\rho)\mathbb{I}^{\otimes k}\big)_{\Gamma}\right]=0\,.
    \end{equation*}
\end{lemma}

\begin{proof}
    The disjoint case follows immediately from the zero-trace property of Eq.~\eqref{eq:centred_local_operator}, and the independent product structure of $\rho^{\otimes n}$. Suppose now that the intersection between the subsets has a non-zero cardinality $j\in\{1,...,k\}$. By invoking the permutation invariance of the local $k$-copy centred kernel $\big(O^{\text{sym}}_k-f(\rho)\mathbb{I}^{\otimes k}\big)$, we may relabel the tensor factors without loss of generality so that $\Lambda=\{1,...,k\}$ and $\Gamma=\{1,...,j,k+1,...,2k-j\}$. Then, the $j^{\text{th}}$-order centred marginal kernel can be defined by tracing out the non-overlapping factors of the local $k$-copy kernel as 
    \[\big(O^{\text{sym}}_k-f(\rho)\mathbb{I}^{\otimes k}\big)_j:=\mathrm{Tr}_{j+1,...,k}\left[\big(O^{\text{sym}}_k-f(\rho)\mathbb{I}^{\otimes k}\big)(\mathbb{I}^{\otimes j}\otimes\rho^{\otimes(k-j)})\right]\,.\] 
   Now, using the fact that both $\big(O^{\text{sym}}_k-f(\rho)\mathbb{I}^{\otimes k}\big)$ and $\rho$ are Hermitian matrices, alongside the cyclicity of the partial trace with respect to operators acting solely on the traced tensor factors, we verify that 
    \begin{align*}
        \big(O^{\text{sym}}_k-f(\rho)\mathbb{I}^{\otimes k}\big)_j^{\dagger}&= \mathrm{Tr}_{j+1,...,k}\left[(\mathbb{I}^{\otimes j}\otimes\rho^{\otimes(k-j)})\big(O^{\text{sym}}_k-f(\rho)\mathbb{I}^{\otimes k}\big)\right]\\
        &=\mathrm{Tr}_{j+1,...,k}\left[\big(O^{\text{sym}}_k-f(\rho)\mathbb{I}^{\otimes k}\big)(\mathbb{I}^{\otimes j}\otimes\rho^{\otimes(k-j)})\right]\\
        &=\big(O^{\text{sym}}_k-f(\rho)\mathbb{I}^{\otimes k}\big)_j\,,
    \end{align*}
    that is, the $j^{\text{th}}$-order centred marginal kernel is Hermitian.

    Since the tensor factors residing outside the common $j$ positions are mutually disjoint, evaluating the expectation over the product structure of $\rho^{\otimes n}$ directly gives 
    \[\mathrm{Tr}\left[\rho^{\otimes n}\big(O^{\text{sym}}_k-f(\rho)\mathbb{I}^{\otimes k}\big)_{\Lambda} \big(O^{\text{sym}}_k-f(\rho)\mathbb{I}^{\otimes k}\big)_{\Gamma}\right]=\mathrm{Tr}\left[\rho^{\otimes j}\big(O^{\text{sym}}_k-f(\rho)\mathbb{I}^{\otimes k}\big)_j^2\right]\ge 0\,.\]
\end{proof}

Now, to track the overlapping subsets across higher-order moments, we map these intersections onto a graph structure. For an integer $r\ge 1$, and an ordered $r$-tuple $\mathbf{\Lambda}=(\Lambda_{1},...,\Lambda_{r})\in\binom{[n]}{k}^r$, let $G(\mathbf{\Lambda})$ denote the \textit{intersection graph} associated with the vertex set $[r]$. In this graph, an edge connects vertex $i$ and $j$ precisely when their respective subsets overlap, i.e., $\Lambda_{i}\cap\Lambda_{j}\ne\emptyset$. We classify the tuple as \textit{connected} whenever its intersection graph $G(\mathbf{\Lambda})$ forms a connected component. 

Leveraging this graph-theoretic perspective, we define the \textit{sum of the absolute moments of the centred quantum U-statistic $\widetilde{U}_{n,k}$ over all connected configurations} as 
\begin{equation}\label{eq:connected_absolute_moment}
    A_{r}:=\frac{1}{\binom{n}{k}^r}\sum_{\substack{\Lambda_{1},...,\Lambda_{r}\in\binom{[n]}{k}:\\ G(\mathbf{\Lambda})\text{ is connected}}} \left\lvert{}\mathrm{Tr}\left[\rho^{\otimes n}\prod_{i=1}^r  \big(O^{\text{sym}}_k-f(\rho)\mathbb{I}^{\otimes k}\big)_{\Lambda_{i}}\right]\right\rvert{}\,.
\end{equation} 
Since individual operators $\big(O^{\text{sym}}_k-f(\rho)\mathbb{I}^{\otimes k}\big)_{\Lambda_{i}}$ inside the product have zero expectation, isolated vertices yield zero, i.e., $A_{1}=0$. For the second-order term, the pair-positivity established in Lemma \ref{lemma:pair_covariance_positivity} allows us to identify this absolute connected moment exactly as the finite-sample variance of the quantum U-statistic $U_{n,k}$.

\begin{lemma}[Second-Order Connected Absolute Moment]\label{lemma:second_moment_variance}
    With the second-order connected absolute moment $A_{2}$ as given by Definition \eqref{eq:connected_absolute_moment}, we precisely recover the variance of a quantum U-statistic, i.e., 
    \begin{equation}
        A_{2}=\mathrm{Var}_{\rho^{\otimes n}}\left(U_{n,k}\right)\,.
    \end{equation}
\end{lemma} 

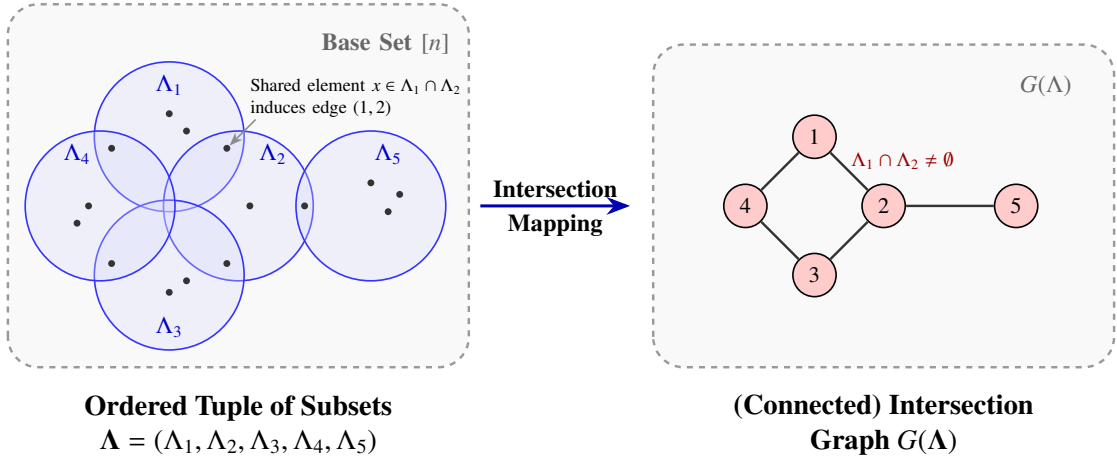
\begin{figure}
    \centering
    \resizebox{0.8\textwidth}{!}{
        \begin{tikzpicture}[
            subset/.style={circle, draw=blue!80, fill=blue!15, fill opacity=0.3, thick, minimum size=2.6cm},
            element/.style={circle, fill=black!80, inner sep=1.2pt},
            vertex/.style={circle, draw=black, fill=red!20, thick, minimum size=7.5mm, font=\bfseries},
            edge/.style={thick, draw=black!80, line width=1.2pt},
            maparrow/.style={-{Stealth[scale=1.2]}, line width=1.5pt, blue!70!black}
        ]
        \begin{scope}[local bounding box=leftpanel]
            \draw[dashed, draw=gray!80, rounded corners=15pt, very thick, fill=gray!5] (-1.8, -1.8) rectangle (6.2, 4.5);
            \node[font=\large\bfseries, text=gray!80!black, anchor=south east] at (6.0, 3.5) {Base Set $[n]$};
        
            \node[subset] (L1) at (1, 2.2) {};
            \node[subset] (L2) at (2.2, 1) {};
            \node[subset] (L3) at (1, -0.2) {};
            \node[subset] (L4) at (-0.2, 1) {};
            \node[subset] (L5) at (4.5, 1) {};
            
            \node[blue!80!black, font=\large\bfseries] at (1, 3.1) {$\Lambda_1$};
            \node[blue!80!black, font=\large\bfseries] at (2.8, 1.9) {$\Lambda_2$};
            \node[blue!80!black, font=\large\bfseries] at (1, -1.1) {$\Lambda_3$};
            \node[blue!80!black, font=\large\bfseries] at (-0.6, 1.9) {$\Lambda_4$};
            \node[blue!80!black, font=\large\bfseries] at (4.8, 1.9) {$\Lambda_5$};
        
            \node[element] (x14) at (0, 2) {};
            \node[element] (x12) at (2, 2) {};
            \node[element] (x23) at (2, 0) {};
            \node[element] (x34) at (0, 0) {};
            \node[element] (x25) at (3.35, 1) {};
            
            \node[element] (u1a) at (1, 2.6) {};   \node[element] (u1b) at (1.3, 2.3) {};
            \node[element] (u2a) at (2.4, 1) {};   \node[element] (u3a) at (1, -0.5) {};  \node[element] (u3b) at (1.3, -0.3) {};
            \node[element] (u4a) at (-0.4, 1) {};  \node[element] (u4b) at (-0.6, 0.7) {};
            \node[element] (u5a) at (4.5, 1.4) {}; \node[element] (u5b) at (4.8, 0.9) {}; \node[element] (u5c) at (5.0, 1.2) {};
            
            \draw[Stealth-, draw=gray!90, thick] (x12) -- ++(0.5, 0.5) 
                node[above right, font=\footnotesize, align=left, inner sep=1pt, text=black, xshift=-4pt] {Shared element $x \in \Lambda_1 \cap \Lambda_2$\\ induces edge $(1,2)$};
        \end{scope}
        
        \node[font=\Large\bfseries, align=center, text=black] at (2.2, -2.8) {Ordered Tuple of Subsets\\ $\mathbf{\Lambda} = (\Lambda_{1}, \Lambda_{2}, \Lambda_{3}, \Lambda_{4}, \Lambda_{5})$};
        
        \draw[maparrow] (6.4, 1) -- (9.0, 1) 
            node[midway, above, font=\large\bfseries, text=black] {Intersection}
            node[midway, below, font=\large\bfseries, text=black] {Mapping};
        
        \begin{scope}[shift={(11.2, 0)}, local bounding box=rightpanel]
            \draw[dashed, draw=gray!80, rounded corners=15pt, very thick, fill=gray!5] (-1.8, -1.8) rectangle (6.2, 3.8);
            \node[font=\large\bfseries, text=gray!80!black, anchor=north east] at (5.6, 3.4) {$G(\Lambda)$};
        
            \node[vertex] (v1) at (1, 2.2) {$1$};
            \node[vertex] (v2) at (2.2, 1) {$2$};
            \node[vertex] (v3) at (1, -0.2) {$3$};
            \node[vertex] (v4) at (-0.2, 1) {$4$};
            \node[vertex] (v5) at (4.5, 1) {$5$};
        
            \draw[edge] (v1) -- (v2) node[midway, above right, font=\small, red!60!black, xshift=-3pt, yshift=-3pt] {$\Lambda_1 \cap \Lambda_2 \neq \emptyset$};
            \draw[edge] (v2) -- (v3);
            \draw[edge] (v3) -- (v4);
            \draw[edge] (v4) -- (v1);
            \draw[edge] (v2) -- (v5);
            
        \end{scope}
        
        \node[font=\Large\bfseries, align=center, text=black] at (13.4, -2.8) {(Connected) Intersection\\ Graph $G(\mathbf{\Lambda})$};
        
        \end{tikzpicture}
    }
    \caption{An Example of an Intersection Graph on a Tuple of Five Subsets}
    \label{fig:intersection_graph_example}
\end{figure}

\begin{proof}
    For the case where $r=2$, the graph connectedness condition trivially reduces to the requirement that the sets overlap, i.e., $\Lambda_{1}\cap\Lambda_{2}\ne\emptyset$. By Lemma \ref{lemma:pair_covariance_positivity}, every corresponding overlapping covariance is non-negative, whereas the covariance vanishes for all disjoint pairs. This allows us to drop the modulo from the definition of $A_2$ and reintroduce the disjoint pairs into the summation without altering its value, yielding 
    \begin{align*}
        A_{2}&=\frac{1}{\binom{n}{k}^2}\sum_{\Lambda,\Gamma\in\binom{[n]}{k}}\mathrm{Tr}\left[\rho^{\otimes n} \big(O^{\text{sym}}_k-f(\rho)\mathbb{I}^{\otimes k}\big)_{\Lambda} \big(O^{\text{sym}}_k-f(\rho)\mathbb{I}^{\otimes k}\big)_{\Gamma}\right]\\
        &=\mathrm{Tr}\left[\rho^{\otimes n}\Biggg(\frac{1}{\binom{n}{k}}\sum_{\Lambda\in \binom{[n]}{k}} \big(O^{\text{sym}}_k-f(\rho)\mathbb{I}^{\otimes k}\big)_{\Lambda}\Biggg)^{2}\right]\\
        &\overset{(a)}{=} \mathrm{Var}_{\rho^{\otimes n}}\big(\widetilde{U}_{n,k}\big)\\
        &\overset{(b)}{=} \mathrm{Var}_{\rho^{\otimes n}}\left(U_{n,k}\right)\,,
    \end{align*}
    where equality ($a$) above follows from the definition of a centred quantum U-statistic, and equality ($b$) follows from equality \eqref{eq:var_equivalence}.
\end{proof}

While the second-order moment of $\widetilde{U}_{n,k}$ is exact, moments of order $r \ge 3$ become highly complex due to the combinatorial explosion of intersecting subsets. To bound these higher-order connected moments, we employ a structural argument: since every connected graph contains at least one spanning tree, we can bound the complex connected absolute moments by summing exclusively over the simpler tree topologies.

\begin{lemma}[Spanning-Tree Bound for Connected Absolute Moments]\label{lemma:spanning_tree_connected_moment}
    For every $r\in\mathbb{Z}_{\ge 3}$, the connected absolute moments of $\widetilde{U}_{n,k}$ are bounded above by 
    \begin{equation*}
        A_{r}\le \dfrac{2^rr^{r-2}k^{2r-2}\left\lVert O^{\text{sym}}_k\right\rVert_{\infty}^r}{n^{r-1}}\,.
    \end{equation*}
\end{lemma}

\begin{proof}
    Let $\mathcal{T}_{r}$ denote the set of all labelled trees constructed on the vertex set $[r]$. We observe that if $G(\mathbf{\Lambda})$ is a connected graph, then it must contain at least one spanning tree. Therefore, we can bound the connected absolute moments by summing over all such trees. Since the individual weights are non-negative, we establish 
    \begin{align}
        A_{r}&=\frac{1}{\binom{n}{k}^r}\sum_{\substack{\Lambda_{1},...,\Lambda_{r}\in\binom{[n]}{k}:\\ G(\mathbf{\Lambda})\text{ is connected}}}\left\lvert \mathrm{Tr}\left[\rho^{\otimes n}\prod_{i=1}^r  \big(O^{\text{sym}}_k-f(\rho)\mathbb{I}^{\otimes k}\big)_{\Lambda_{i}}\right]\right\rvert{}\notag \\
        &\leq\frac{1}{\binom{n}{k}^r}\sum_{T\in \mathcal{T}_{r}}\sum_{\substack{\Lambda_{1},...,\Lambda_{r}\in\binom{[n]}{k}:\\\Lambda_{i}\cap\Lambda_{j}\ne\emptyset\,,\,\forall(i,j)\in E(T)}} \left\lvert \mathrm{Tr}\left[\rho^{\otimes n}\prod_{i=1}^r  \big(O^{\text{sym}}_k-f(\rho)\mathbb{I}^{\otimes k}\big)_{\Lambda_{i}}\right]\right\rvert \notag\\
        &\overset{(a)}{\leq}\frac{1}{\binom{n}{k}^r}\sum_{T\in \mathcal{T}_{r}}\sum_{\substack{\Lambda_{1},...,\Lambda_{r}\in\binom{[n]}{k}:\\\Lambda_{i}\cap\Lambda_{j}\ne\emptyset\,,\,\forall(i,j)\in E(T)}} \prod_{i=1}^r {\left\lVert{}\big(O^{\text{sym}}_k-f(\rho)\mathbb{I}^{\otimes k}\big)_{\Lambda_{i}}\right\rVert{}}_{\infty} \notag\\
        &\overset{(b)}{\leq}\left(\frac{2{\left\lVert{}O^{\text{sym}}_k\right\rVert{}}_{\infty}}{\binom{n}{k}}\right)^r\sum_{T\in \mathcal{T}_{r}}\sum_{\substack{\Lambda_{1},...,\Lambda_{r}\in\binom{[n]}{k}:\\\Lambda_{i}\cap\Lambda_{j}\ne\emptyset\,,\,\forall(i,j)\in E(T)}} 1\qquad,\label{eq:double_sum_inequality}
    \end{align}
    where inequality ($a$) follows from an application of H\"older's inequality to the trace expression, followed by the submultiplicativity of matrix infinity norms, while inequality ($b$) follows from the fact that 
    \begin{align*}
        \left\lVert O^{\text{sym}}_k - f(\rho)\mathbb{I}^{\otimes k} \right\rVert_\infty &\leq \left\lVert O^{\text{sym}}_k \right\rVert_\infty + \left\lVert f(\rho)\mathbb{I}^{\otimes k} \right\rVert_\infty\\
        &=\left\lVert O^{\text{sym}}_k \right\rVert_\infty + \left\lvert \mathrm{Tr}\left[O^{\text{sym}}_k \rho^{\otimes k}\right]\right\rvert\\
        &\overset{(c)}{\leq} 2\left\lVert O^{\text{sym}}_k \right\rVert_\infty\,,
    \end{align*}
    where the last step ($c$) again follows from an application of H\"older's inequality to the expression $\left\lvert \mathrm{Tr}\left[O^{\text{sym}}_k \rho^{\otimes k}\right]\right\rvert\,$.

    We now present Algorithm \ref{alg:spanning_tree_leaf_pruning} to upper-bound the complicated-looking double-sum presented in inequality \eqref{eq:double_sum_inequality}. Figure \ref{fig:spanning_tree_leaf_pruning} presents a visual working of the algorithm.

    \begin{algorithm}
    \caption{Spanning Tree Leaf-Pruning Method for Upper-Bounding $r^{\text{th}}$ Absolute Connected Moments}
    \label{alg:spanning_tree_leaf_pruning}
    \begin{algorithmic}[1]
        \Require Base set $[n]$, subset size $1<k\leq n$, subset tuple cardinality $r \ge 3$
        \Ensure Upper bound $R$ on the $r^{\text{th}}$ absolute connected moment
        \State $R \gets 0$ \Comment{Initialize the $r^{\text{th}}$ absolute connected moment}
        \State $\mathcal{T}_r \gets $ All $r^{r-2}$ labelled spanning trees on vertex set $V =[n]$
        \For{\textbf{each} tree $T = (V_T, E_T) \in \mathcal{T}_r$}
            \State $M_T \gets 1$ \Comment{Initialize the absolute connected moment for the tree $T$}
            \While{$|V_T| > 1$}
                \State Select a peripheral leaf node $i \in V_T$ (i.e., a node with degree 1)
                \State Let $p \in V_T$ be the unique parent node of $i$ in $T$
                \State $M_T \gets M_T \times k \binom{n-1}{k-1}$ \Comment{Bound the free choices for $\Lambda_i$ given a fixed parent subset $\Lambda_p$}
                \State $V_T \gets V_T \setminus \{i\}$ \Comment{Prune leaf $i$ from the tree $T$}
                \State $E_T \gets E_T \setminus \{(i, p)\}$ \Comment{Remove the incident edge on leaf $i$}
            \EndWhile
            \State $M_T \gets M_T \times \binom{n}{k}$ \Comment{Only the root remains; subset choices unconstrained for the root}
            \State $R \gets R + M_T$ \Comment{Accumulate the absolute connected moment for the tree $T$}
        \EndFor
        \State \Return $R$
    \end{algorithmic}
    \end{algorithm}
    
    In other words, to upper-bound the double sum presented in inequality \eqref{eq:double_sum_inequality}, the algorithm picks a tree $T$ from $\cT_r$ in every iteration of the \textbf{for} loop of lines 3-14. Without loss of generality, let a peripheral leaf node labelled $i$ be connected to its parent node labelled $p$ in the tree $T$. For any fixed configuration of the parent subset $\Lambda_{p}$, line 8 bounds the sum over the subsets corresponding to the leaf node $i$ as 
    \[\sum_{\Lambda_{i}\in\binom{[n]}{k} : \Lambda_{i}\cap\Lambda_{p}\ne\emptyset} 1 \le\sum_{v\in\Lambda_{p}}\sum_{\Lambda_{i}\in\binom{[n]}{k} : v\in\Lambda_{i}} 1 \le k \binom{n-1}{k-1}= \frac{k^2}{n} \binom{n}{k}\,,\]
    while lines 9-10 prune the leaf node $i$ off the tree $T$.
        
    By iteratively applying this bound from the leaves inward, the algorithm eliminates the non-root vertices one at a time. After successfully eliminating all $(r-1)$ non-root nodes, line 12 accounts for the remaining isolated sum over the root vertex of at most $$\sum_{\Lambda_{1}\in\binom{[n]}{k}} 1=\binom{n}{k}\,.$$ Thus, for every specifically fixed labeled tree $T$, line 13 upper-bounds the associated inner sum by an overall factor of at most $\binom{n}{k}\left(\frac{k^2}{n} \binom{n}{k}\right)^{r-1}$. 
    
    By Cayley's formula, the total number of distinct labeled trees is exactly $|\mathcal{T}_{r}|=r^{r-2}$. Multiplying this combinatorial factor, and the multiplicative factor of $\Big(2{\big\lVert{}O^{\text{sym}}_k\big\rVert{}}_{\infty}\big/\binom{n}{k}\Big)^r$ with the inner-sum upper-bound completes the proof.
\end{proof}

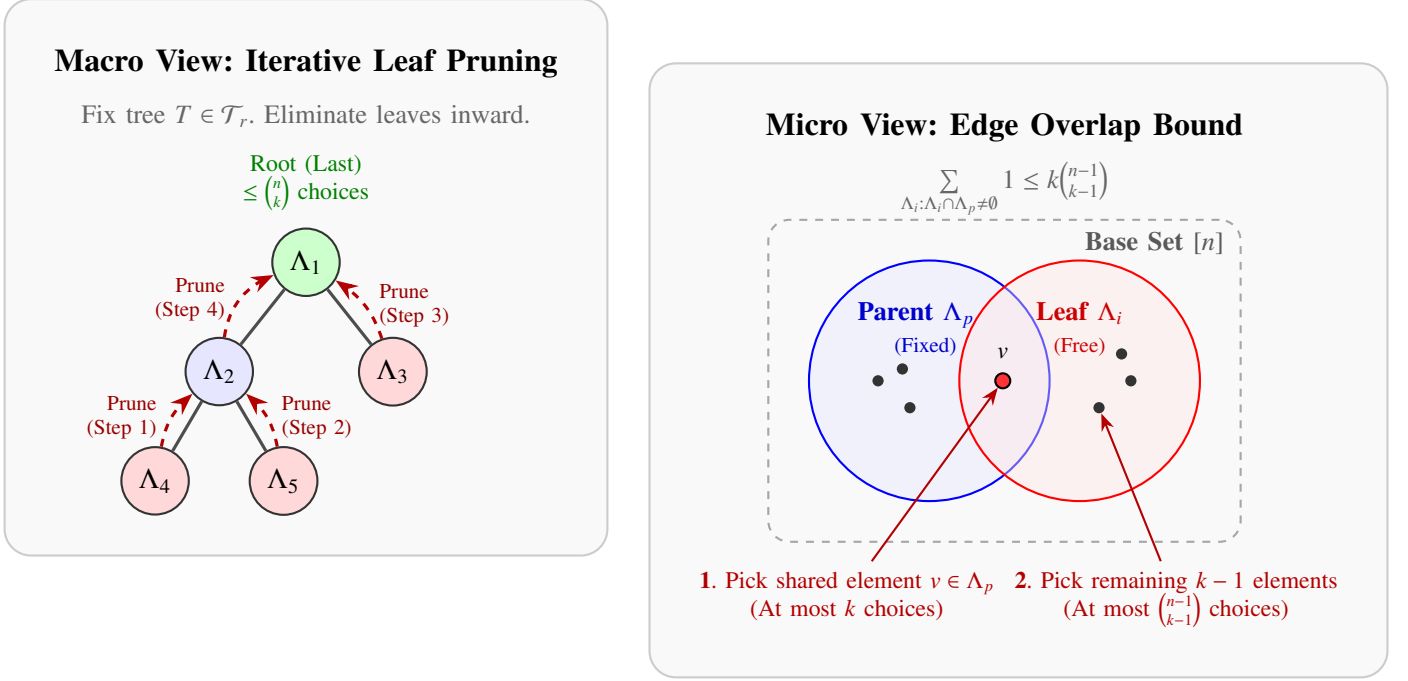
\begin{figure}
    \centering
    \resizebox{\textwidth}{!}{
        \begin{tikzpicture}
    \tikzset{
        treenode/.style={circle, draw=black!80, fill=blue!10, thick, minimum size=8mm, font=\bfseries},
        rootnode/.style={circle, draw=black!80, fill=green!20, thick, minimum size=8mm, font=\bfseries},
        leafnode/.style={circle, draw=black!80, fill=red!15, thick, minimum size=8mm, font=\bfseries},
        treeedge/.style={thick, draw=black!70, line width=1.2pt},
        prunearrow/.style={-Stealth, line width=1.2pt, red!70!black, dashed},
        setcircle/.style={circle, draw=#1, fill=#1!10, fill opacity=0.4, thick, minimum size=3.2cm},
        dot/.style={circle, fill=black!80, inner sep=1.5pt},
        highlightdot/.style={circle, fill=red!80, inner sep=2pt, draw=black, thick},
        panelbox/.style={rounded corners=10pt, fill=gray!5, draw=gray!40, thick, inner sep=15pt},
        paneltitle/.style={font=\large\bfseries, text=black},
        panelsubtitle/.style={font=\small, text=gray!80!black, align=center}
    }

    \node[paneltitle] (titleA) {Macro View: Iterative Leaf Pruning};
    \node[panelsubtitle, below=0.1cm of titleA] (subA) {Fix tree $T \in \mathcal{T}_r$. Eliminate leaves inward.};

    \node[rootnode, below=1.2cm of subA] (L1) {$\Lambda_1$};
    \node[font=\footnotesize, text=green!50!black, align=center, above=0.1cm of L1] (L1_label) {Root (Last)\\ $\le \binom{n}{k}$ choices};

    \node[treenode, below left=0.8cm and 0.5cm of L1] (L2) {$\Lambda_2$};
    \node[leafnode, below right=0.8cm and 0.5cm of L1] (L3) {$\Lambda_3$};

    \node[leafnode, below left=0.8cm and 0.2cm of L2] (L4) {$\Lambda_4$};
    \node[leafnode, below right=0.8cm and 0.2cm of L2] (L5) {$\Lambda_5$};

    \draw[treeedge] (L1) -- (L2);
    \draw[treeedge] (L1) -- (L3);
    \draw[treeedge] (L2) -- (L4);
    \draw[treeedge] (L2) -- (L5);

    \draw[prunearrow] (L4) to[bend left=20] node[midway, left=0.05cm, font=\scriptsize, text=red!60!black, align=right] {Prune\\(Step 1)} (L2);
    \draw[prunearrow] (L5) to[bend right=20] node[midway, right=0.05cm, font=\scriptsize, text=red!60!black, align=left] {Prune\\(Step 2)} (L2);
    \draw[prunearrow] (L3) to[bend right=20] node[midway, right=0.05cm, font=\scriptsize, text=red!60!black, align=left] {Prune\\(Step 3)} (L1);
    \draw[prunearrow] (L2) to[bend left=30] node[midway, left=0.05cm, font=\scriptsize, text=red!60!black, align=right] {Prune\\(Step 4)} (L1);

    \node[paneltitle, above right=-1.5cm and 2.5cm of titleA] (titleB) {Micro View: Edge Overlap Bound};
    \node[panelsubtitle, below=0cm of titleB] (subB) {$\sum\limits_{\Lambda_{i}: \Lambda_{i}\cap\Lambda_{p}\ne\emptyset} 1 \le k \binom{n-1}{k-1}$};

    \node[below=1.9cm of subB] (B_center) {};

    \node[setcircle=blue, left=-0.75cm of B_center] (Lp) {};
    \node[setcircle=red, right=-0.75cm of B_center] (Li) {};

    \node[highlightdot, right=-0.75cm of Lp.east] (v) {};
    \node[font=\footnotesize, text=black, above=0.05cm of v] (v_lbl) {$v$};

    \node[font=\small\bfseries, text=blue!80!black, align=center, above=0.2cm of Lp.center] (Lp_label) {Parent $\Lambda_p\quad$\\ \normalfont\scriptsize(Fixed)\ \ };
    \node[font=\small\bfseries, text=red!80!black, align=center, above=0.2cm of Li.center] (Li_label) {Leaf $\Lambda_i$\\ \normalfont\scriptsize(Free)};

    \node[rounded corners=8pt, dashed, draw=gray!70, thick, fit=(Lp) (Li) (Lp_label) (Li_label), inner sep=15pt] (baseset) {};
    \node[font=\small\bfseries, text=gray!70!black, above left=-0.6cm and 0.1cm of baseset.north east] {Base Set $[n]$};

    \node[dot, left=0.6cm of Lp.center] (p1) {};
    \node[dot, above left=0.1cm and 0.3cm of Lp.center] (p2) {};
    \node[dot, below left=0.3cm and 0.2cm of Lp.center] (p3) {};

    \node[dot, right=0.6cm of Li.center] (i1) {};
    \node[dot, above right=0.3cm and 0.5cm of Li.center] (i2) {};
    \node[dot, below right=0.3cm and 0.2cm of Li.center] (i3) {};

    \node[font=\footnotesize, text=red!70!black, align=center, below left=2.4cm and -1cm of Lp.center] (v_text) {$\mathbf{1.}$ Pick shared element $v \in \Lambda_p$\\ (At most $k$ choices)};
    \draw[-Stealth, thick, red!70!black] (v_text) -- (v);

    \node[font=\footnotesize, text=red!70!black, align=center, below right=2.4cm and -1cm of Li.center] (i_text) {$\mathbf{2.}$ Pick remaining $k-1$ elements\\ (At most $\binom{n-1}{k-1}$ choices)};
    \draw[-Stealth, thick, red!70!black] (i_text) -- (i3);

    \begin{scope}[on background layer]
        \node[panelbox, fit=(titleA) (subA) (L1_label) (L4) (L5) (L3)] (boxA) {};
        \node[panelbox, fit=(titleB) (subB) (baseset) (v_text) (i_text)] (boxB) {};

        \node[right=0.75cm of titleA] (mid_title) {};
        \node[below=3.5cm of mid_title] (mid_arrow_center) {};
        \node[left=0.2cm of mid_arrow_center] (arr_start) {};
        \node[right=0.2cm of mid_arrow_center] (arr_end) {};
    \end{scope}

\end{tikzpicture}
    }
    \caption{Visualizing Algorithm \ref{alg:spanning_tree_leaf_pruning} to Upper Bound the $r^{\text{th}}$ Absolute Connected Moments ($r\in\mathbb{Z}_{\ge 3}$)}
    \label{fig:spanning_tree_leaf_pruning}
\end{figure}

To transition from bounding individual connected absolute moments of $\widetilde{U}_{n,k}$ to establishing a global concentration inequality for $U_{n,k}$, we first need to evaluate the associated moment-generating function (MGF) of $\widetilde{U}_{n,k}$. The following lemma achieves this by leveraging a cluster expansion — a technique prominent in statistical mechanics — which allows us to factorize the macroscopic MGF into independent connected components.

\begin{lemma}[Centred Quantum U-Statistic MGF Upper-Bound]\label{lemma:MGF}
    For every formal parameter $t\in\mathbb{R}$, the moment generating function of the centred quantum U-statistic $\widetilde{U}_{n,k}$ is bounded above by its exponentiated connected absolute moments as
    \begin{equation*}
        \mathrm{Tr}[\rho^{\otimes n}e^{t\widetilde{U}_{n,k}}]\le \exp\left[\sum_{r=2}^{\infty}\frac{|t|^{r}}{r!}A_{r}\right]\,,
    \end{equation*} 
    provided that the sum inside the exponent on the right-hand side of the above inequality is finite, and where $A_r$ is as given by Definition~\eqref{eq:connected_absolute_moment}. 
\end{lemma}

\begin{proof}
    We begin by taking the Taylor series expansion of the matrix exponential as 
    \begin{equation}\label{eq:MGF_centred_u_statistic}
        \mathrm{Tr}[\rho^{\otimes n}e^{t\widetilde{U}_{n,k}}] \le \sum_{r=0}^{\infty}\frac{|t|^{r}}{r!}B_{r}\,,
    \end{equation} 
    where $B_{r}$ is the total unconstrained absolute moment defined as 
    $$B_{r}:=\frac{1}{\binom{n}{k}^r}\sum_{\Lambda_{1},...,\Lambda_{r}\in \binom{[n]}{k}}\left\lvert{}\mathrm{Tr}\left[\rho^{\otimes n} \prod_{i=1}^r \big(O^{\text{sym}}_k-f(\rho)\mathbb{I}^{\otimes k}\big)_{\Lambda_i}\right]\right\rvert{}\,.$$ 
    For any fixed tuple $\mathbf{\Lambda}=(\Lambda_{1},...,\Lambda_{r})\in \binom{[n]}{k}$, let $\pi(\mathbf{\Lambda})$ denote the partition of the index set $[r]$ grouped according to the isolated connected components of the intersection graph $G(\mathbf{\Lambda})$. Operators that are assigned to entirely distinct disconnected components fundamentally act on disjoint, non-overlapping subsets of the tensor factors. As a direct consequence, these operators commute with one another, allowing the global product state expectation to factorize cleanly across the distinct components as 
    $$\left\lvert{}\mathrm{Tr}\left[\rho^{\otimes n} \prod_{i=1}^r \big(O^{\text{sym}}_k-f(\rho)\mathbb{I}^{\otimes k}\big)_{\Lambda_i}\right]\right\rvert{} = \prod_{C\in\pi(\mathbf{\Lambda})}\left\lvert{}\mathrm{Tr}\left[\rho^{\otimes n}\prod_{i\in C}\big(O^{\text{sym}}_k-f(\rho)\mathbb{I}^{\otimes k}\big)_{\Lambda_i}\right]\right\rvert{}\,.$$ 
    Now, if we first fix a specific set-partition and sum over all the tuples obeying the fixed set-partition, while also dropping the requirement that different blocks utilize non-overlapping tensor supports, then we only increase the sum. Doing so, we get 
    \begin{equation}\label{eq:B_r_upper_bound}
        B_{r}\le\sum_{\pi\in\mathfrak{P}([r])}\prod_{C\in\pi}A_{|C|}\,,
    \end{equation} 
    where $\mathfrak{P}([r])$ represents the set of all possible partitions of the index set $[r]$. 
    
    Let us now group the partitions of the $r$-element set by their specific block sizes. Let $k_m$ denote the number of blocks of exactly size $m$ in a given partition $\pi\in \mathfrak{P}([r])$, subject to the total cardinality constraint $\sum_{m=1}^r m k_m = r$. The total number of distinct ways to partition an $r$-element set into exactly $k_m$ blocks of size $m$ is given by $$\frac{r!}{\prod\limits_{m=1}^r k_m! (m!)^{k_m}}\,.$$ Since each block of size $m$ contributes a scalar weight of $A_m$, the total weight of all partitions corresponding to the configuration $\{k_m\}_{m=1}^r$ is $\prod_{m=1}^r A_m^{k_m}$.
    
    Therefore, we can instead rewrite the sum in Eq.~\eqref{eq:B_r_upper_bound} by summing over all valid block-size configurations as
    {$$ \displaystyle\sum_{\pi\in\mathfrak{P}([r])}\prod_{C\in\pi}A_{|C|} = r!\Bigg(\sum_{\substack{\{k_1,\cdots,k_r\} :\\ \sum\limits_{m\in [r]} m k_m = r}} \prod_{m=1}^{r} \frac{1}{k_m!} \left( \frac{A_m}{m!} \right)^{k_m}\Bigg)\,. $$}
    
    Substituting this into the right-hand side of Ineq.~\eqref{eq:MGF_centred_u_statistic}, the $r!$ in the numerator perfectly cancels the $r!$ in the denominator of the series. Furthermore, summing over all $r$ from $0$ to $\infty$ completely removes the size constraint $\sum m k_m = r$, thereby allowing the number of blocks $k_m$ of each size to vary independently from $0$ to $\infty$. This simplifies the summation as
    \begin{align*}
        \sum_{r=0}^{\infty}\frac{|t|^{r}}{r!}B_{r} &\le \sum_{r=0}^{\infty} \sum_{\{k_m\} : \sum m k_m = r} \prod_{m=1}^{r} \frac{1}{k_m!} \left( \frac{A_m |t|^m}{m!} \right)^{k_m} \\
        &=\sum_{k_1=1}^\infty \sum_{k_2=1}^\infty\ldots \prod_{m=1}^{\infty} \frac{1}{k_m!} \left( \frac{A_m |t|^m}{m!} \right)^{k_m}\\
        &=\left(\sum_{k_1=1}^\infty \frac{1}{k_1!} \left( \frac{A_1 |t|^1}{1!} \right)^{k_1}\right)\left(\sum_{k_2=1}^\infty \frac{1}{k_2!} \left( \frac{A_2 |t|^2}{2!} \right)^{k_2}\right)\ldots \\
        &= \prod_{m=1}^{\infty} \left( \sum_{k_m=0}^{\infty} \frac{1}{k_m!} \left( \frac{A_m |t|^m}{m!} \right)^{k_m} \right)\,.
    \end{align*}
    
    Observe that the inner sum over $k_m$ is exactly the standard Taylor series expansion for the exponential function $\mathrm{exp}\big(A_m |t|^m/m!\big)$. Since the product of exponentials equates to the exponential of the sum, this directly yields
    $$ \prod_{m=1}^{\infty} \exp\left( \frac{A_m |t|^m}{m!} \right) = \exp\left( \sum_{m=1}^{\infty} \frac{|t|^{m}}{m!} A_{m} \right)\,. $$
    
    Recalling that $A_1=0$ completes the proof.
\end{proof}

Having isolated the variance in Lemma \ref{lemma:second_moment_variance}, bounded the higher-order moments via a spanning tree argument in Lemma \ref{lemma:spanning_tree_connected_moment}, and integrated them into the centred quantum U-statistic MGF upper-bound in Lemma \ref{lemma:MGF}, we have the following variance-sensitive MGF expansion as a corollary.

\begin{corollary}[Variance-Sensitive Exponential MGF Upper-Bound]\label{col:var_sensitive_exponential_bound}
    Let $U_{n,k}$ denote the global $n$-copy, order-$k$ quantum U-statistic corresponding to the $k$-copy local kernel $O_{k}^{\text{sym}}$, and let $\widetilde{U}_{n,k}$ be its centred deviation as defined in Eq.~\eqref{eq:centred_u_stat_definition}. Assuming the kernel is non-trivial such that 
    \begin{equation*}
        0<{\big\lVert{}O^{\text{sym}}_k-f(\rho)\mathbb{I}^{\otimes k}\big\rVert{}}_\infty\leq 2{\left\lVert{}O^{\text{sym}}_k\right\rVert{}}_\infty\,, 
    \end{equation*} 
    for every arbitrary scaling parameter $t\in\mathbb{R}$ that resides within the radius of convergence 
    \begin{equation*}
        \dfrac{2ek^{2}|t|{\left\lVert{}O^{\text{sym}}_k\right\rVert{}}_\infty}{n}<1\,,
    \end{equation*} 
    the moment-generating function is bounded above by a variance-sensitive converging exponent given by 
    \begin{equation*}
        \mathrm{Tr}[\rho^{\otimes n}e^{t\widetilde{U}_{n,k}}]\le \exp\left(\frac{t^{2}\mathrm{Var}_{\rho^{\otimes n}}\left(U_{n,k}\right)}{2}+\sum_{r=3}^{\infty}\dfrac{|t|^r2^rr^{r-2}k^{2r-2}\left\lVert O^{\text{sym}}_k\right\rVert_{\infty}^r}{r!n^{r-1}}\right)\,.
    \end{equation*}
\end{corollary}

\begin{proof}
    By invoking the component expansion bound derived in Lemma \ref{lemma:MGF}, we have
    $$\mathrm{Tr}[\rho^{\otimes n}e^{t\widetilde{U}_{n,k}}]\le \exp\left(\frac{t^{2}}{2}A_{2}+\sum_{r=3}^{\infty}\frac{|t|^{r}}{r!}A_{r}\right)\,.$$ 
    We then substitute the structurally exact variance identity $A_{2}=\mathrm{Var}_{\rho^{\otimes n}}\left(U_{n,k}\right)$ authenticated by Lemma \ref{lemma:second_moment_variance}, and replace the generic higher-order terms $A_r$ (for all $r\ge 3$) with the explicit spanning-tree bounds derived in Lemma \ref{lemma:spanning_tree_connected_moment}, thereby producing the desired exponential bound. Finally, executing the standard ratio test on the resulting infinite series, we get 
    \begin{align*}
        \dfrac{a_{n+1}}{a_n}&=\dfrac{\frac{|t|^{r+1}2^{r+1}{(r+1)}^{r-1}k^{2r}\left\lVert O^{\text{sym}}_k\right\rVert_{\infty}^{r+1}}{(r+1)!n^r}}{\frac{|t|^r2^rr^{r-2}k^{2r-2}\left\lVert O^{\text{sym}}_k\right\rVert_{\infty}^r}{r!n^{r-1}}}\\
        &=\frac{2k^2\vert{}t\vert{} \left\lVert O^{\text{sym}}_k\right\rVert_{\infty}}{n} \left(1+\frac{1}{r}\right)^{r-2}\\
        &<\frac{2k^2\vert{}t\vert{} \left\lVert O^{\text{sym}}_k\right\rVert_{\infty}}{n} \left(1+\frac{1}{r}\right)^r\\
        &<\frac{2ek^2\vert{}t\vert{} \left\lVert O^{\text{sym}}_k\right\rVert_{\infty}}{n}\,,
    \end{align*}
    which guarantees that the exponent converges from the assumption that $2ek^2\vert{}t\vert{} \left\lVert O^{\text{sym}}_k\right\rVert_{\infty}\big/n<1$.
\end{proof}

\begin{remark}[Comparison of MGF Bound of Corollary \ref{col:var_sensitive_exponential_bound} with {\cite[Theorem 3.1]{De_Palma_2025}}]\label{rem:mgf_comparison}
    We now compare our above-derived MGF upper bound with Theorem 3.1 from the work of De Palma and Pastorello \cite{De_Palma_2025}. Their work establishes an upper bound on the moment-generating function for local observables, which is highly useful for quantum spin systems. To handle the complicated overlaps between local operators, they use a star graph decomposition of a minimal subgraph. Motivated by this graph method, we handle these overlaps by using a spanning tree argument to bound the connected absolute moments, as shown in Lemma \ref{lemma:spanning_tree_connected_moment}. By adding up only the tree shapes and using Cayley's formula, we control the huge number of intersecting subsets directly without needing to build a minimal subgraph.

    Additionally, De Palma and Pastorello's result uses a quantum local norm metric that measures the total interaction strength of the local terms. While this fits well with optimal mass transport methods in thermodynamics, it is less helpful for our statistical estimation framework because it does not isolate the exact statistical fluctuations of the state. Instead, Lemma \ref{lemma:MGF} cleanly separates the second-order contribution found in Lemma \ref{lemma:second_moment_variance}. This isolates the finite-sample variance $\mathrm{Var}_{\rho^{\otimes n}}(U_{n,k})$ to make the leading term of the bound explicitly quadratic, while bounding the worst-case extreme errors using only the standard spectral norm ${\big\lVert O^{\text{sym}}_k \big\rVert}_{\infty}$.
\end{remark}

With the variance-sensitive MGF upper bound established, we can now derive a corresponding two-sided probability concentration bound for the quantum U-statistic measurement deviation from the actual functional value.

\begin{corollary}[Two-sided Exponential Variance-Sensitive Concentration Bound]\label{col:two_sided_concentration}
    For every specified precision $\varepsilon>0$, the total absolute failure probability obeys
    \begin{equation*}
        \mathrm{Pr}\left(\left\lvert Z_{U_{n,k}} - f(\rho) \right\rvert \ge \varepsilon\right) \le 2 \exp\left(-nI_{n}(\varepsilon)\right)\,, 
    \end{equation*}
    where $Z_{U_{n,k}}$ denotes the outcome of globally measuring $\rho^{\otimes n}$ in the eigenbasis of the quantum U-statistic $U_{n,k}$, and the optimal finite-sample rate function $I_{n}(\varepsilon)$ is defined as 
    \begin{equation}\label{eq:rate_function}
        I_{n}(\varepsilon) := \sup_{0\le s<\frac{1}{2ek^2 \left\lVert O^{\text{sym}}_k\right\rVert_{\infty}}} \left\{ s\varepsilon - \frac{n\mathrm{Var}_{\rho^{\otimes n}}\left(U_{n,k}\right)}{2}s^{2} - \Psi(s) \right\}\,, 
    \end{equation}
    where $\Psi(s)$ defines an auxiliary penalty function that isolates the higher-order combinatorial error terms as
    \begin{equation}\label{eq:Psi_definition}
        \Psi(s) := \sum_{r=3}^{\infty}\dfrac{s^r 2^r r^{r-2} k^{2r-2} \left\lVert O^{\text{sym}}_k\right\rVert_{\infty}^r}{r!}\,.
    \end{equation}
\end{corollary}

\begin{proof}
    Because the MGF bound of Corollary \ref{col:var_sensitive_exponential_bound} depends only on the absolute value $|t|$, we restrict our attention to non-negative parameters $t \ge 0$. Applying Markov's inequality to the moment-generating function yields a generic concentration inequality for any precision parameter $\varepsilon > 0$, given by
    $$ \mathrm{Pr}\left(Z_{\widetilde{U}_{n,k}} \ge \varepsilon\right) \le e^{-t\varepsilon}\mathrm{Tr}\left[\rho^{\otimes n}e^{t\widetilde{U}_{n,k}}\right] \le \exp\left(-t\varepsilon + \frac{t^{2}\mathrm{Var}_{\rho^{\otimes n}}\left(U_{n,k}\right)}{2} + \sum_{r=3}^{\infty}\dfrac{t^r 2^r r^{r-2} k^{2r-2} \left\lVert O^{\text{sym}}_k\right\rVert_{\infty}^r}{r! n^{r-1}}\right)\,. $$
    
    To guarantee that this upper bound decays exponentially with respect to the sample size $n$, the optimization parameter $t$ must scale linearly with $n$. Furthermore, for the overall exponent to remain negative and physically meaningful, the linear decay term $-t\varepsilon$ must overpower the positive variance and higher-order penalty terms. We therefore reparameterize the bound by choosing $t = sn$, where $s \ge 0$ acts as a scale-free optimization parameter. 
    
    Substituting $t = sn$ into the original convergence radius $2ek^2\vert{}t\vert{} \big\lVert O^{\text{sym}}_k\big\rVert_{\infty}\big/n$ (which is smaller than 1) restricts our new parameter $s$ to the domain $0 \le s < 1\big/\big(2ek^2 \big\lVert O^{\text{sym}}_k\big\rVert_{\infty}\big)$. Under this substitution, the sample size $n$ elegantly factors out of the infinite series, allowing us to define an $n$-independent auxiliary penalty function that isolates the higher-order combinatorial errors as
    \begin{equation}\label{eq:Psi_definition}
        \Psi(s) := \sum_{r=3}^{\infty}\dfrac{s^r 2^r r^{r-2} k^{2r-2} \left\lVert O^{\text{sym}}_k\right\rVert_{\infty}^r}{r!}\,.
    \end{equation}
    
    By substituting $t = sn$ and the definition of $\Psi(s)$ back into our generic Markov bound, we factor out $n$ from the entire exponent. Noting that the deviation event $Z_{U_{n,k}} - f(\rho) \ge \varepsilon$ is identical to $Z_{\widetilde{U}_{n,k}} \ge \varepsilon$ for the centered operator, and taking the supremum over the valid range of $s$ to optimize the decay rate, we directly obtain the finite-sample upper-tail concentration inequality, i.e., $$\mathrm{Pr}\left(Z_{U_{n,k}} - f(\rho) \ge \varepsilon\right) \le \exp\left(-nI_{n}(\varepsilon)\right)\,.$$
    
    Now, we observe that the probability bound established above is symmetric with respect to sign inversions of the local $k$-copy operators. In other words, replacing the operator $\left(O^{\text{sym}}_k-f(\rho)\mathbb{I}^{\otimes k}\right)$ with $-\left(O^{\text{sym}}_k-f(\rho)\mathbb{I}^{\otimes k}\right)$ perfectly preserves the spectral bound ${\left\lVert{}O^{\text{sym}}_k-f(\rho)\mathbb{I}^{\otimes k}\right\rVert{}}_\infty\leq 2{\left\lVert{}O^{\text{sym}}_k\right\rVert{}}_\infty$, and identically maintains all underlying pair covariances. This symmetry allows us to naturally extend the bound to both the tails simultaneously, thus completing the proof.
\end{proof}

\begin{remark}[Comparison of Concentration Inequality of Corollary \ref{col:two_sided_concentration} with {\cite[Theorem 3.2]{De_Palma_2025}}]
    As discussed in Remark \ref{rem:mgf_comparison}, choosing our spanning tree method over the star graph and local norm techniques used by De Palma and Pastorello \cite{De_Palma_2025} directly shapes our final probability bounds. Theorem 3.2 in the latter manuscript gives a concentration bound based on the total interaction strength of the local terms. In contrast, turning our MGF bound into a tail probability yields a matrix Bernstein-type inequality. This formula perfectly separates the normal Gaussian errors, which depend on the exact finite-sample variance, from the worst-case extreme errors, which depend on the standard spectral norm. Because of this, we do not need to calculate any extra local norm metrics.
\end{remark}

Finally, to better contextualize these finite-sample probabilistic bounds within the framework of quantum estimation theory, we link the optimal rate function back to the underlying Fisher geometry. We combine the concentration bound with the variance expansion proved earlier in Section \ref{sec:variance_optimality}. Recall that, evaluated at a fixed physical state $\rho$, the total variance factors into an intrinsic leading term and a suppressed multi-copy interaction error as
\begin{align}
    &\mathrm{Var}_{\rho^{\otimes n}}\left(U_{n,k}\right)=\frac{\mathrm{Var}_{\rho}\left(\nabla f(\rho)\right)}{n}+\mathcal{O}\left(\frac{1}{n^2}\right)\,,\nonumber\\
    \text{or,}\qquad &n\mathrm{Var}_{\rho^{\otimes n}}\left(U_{n,k}\right)=\mathrm{Var}_{\rho}\left(\nabla f(\rho)\right)+\mathcal{O}\left(\frac{1}{n}\right)\,.\label{eq:variance_limit}
\end{align}

Combining Corollary \ref{col:two_sided_concentration} with Eq.~\eqref{eq:variance_limit} directly yields the following lower bound on the quantum U-statistic concentration convergence rate in the large-$n$ limit.

\begin{theorem}[Variance-Sensitive Concentration Convergence Rate Large-$n$ Limit Lower Bound]\label{thm:large_deviation_lower_bound}
    Assuming that the estimation geometry is non-degenerate (i.e., $\mathrm{Var}_{\rho}\left(\nabla f(\rho)\right)>0$), evaluating the $n$-copy, order-$k$ quantum U-statistic $U_{n,k}$ performance for every $\varepsilon>0$ yields the large-limit deviation lower bound as 
    \begin{equation*}
        \liminf_{n\to\infty}-\frac{1}{n}\log \mathrm{Pr}\left(\left|Z_{U_{n,k}}-f(\rho)\right|\ge\varepsilon\right)\ge I(\varepsilon)\,,
    \end{equation*} 
    where $I(\varepsilon)$ denotes an optimal large-limit two-sided rate function defined as 
    \begin{equation*}
        I(\varepsilon):=\sup_{0\le s<\frac{1}{2ek^2 \left\lVert O^{\text{sym}}_k\right\rVert_{\infty}}} \left\{s\varepsilon-\frac{\mathrm{Var}_{\rho}\left(\nabla f(\rho)\right)}{2}s^{2}-\Psi(s)\right\}\,,
    \end{equation*} 
    and where $Z_{U_{n,k}}$ denotes the outcome of measuring the product state $\rho^{\otimes n}$ globally in the eigenbasis of the quantum U-statistic $U_{n,k}$, and $\Psi(s)$ is as given by Definition~\eqref{eq:Psi_definition}. 
    
    In particular, the above lower bound guarantees a positive exponential decay rate, since $I(\varepsilon)>0$ for every $\varepsilon>0$.
\end{theorem}

\begin{proof}
    Fix any parameter $s$ within the half-open interval $\left[0,1\big/\big(2ek^2 \big\lVert O^{\text{sym}}_k\big\rVert_{\infty}\big)\right)$. By substituting the variance expansion identity \eqref{eq:variance_limit} directly into the two-sided exponent derived in Corollary \ref{col:two_sided_concentration}, we get 
    $$-\frac{1}{n}\log \mathrm{Pr}\left(\left|Z_{U_{n,k}}-f(\rho)\right| \ge\varepsilon\right)\ge s\varepsilon-\frac{n\mathrm{Var}_{\rho^{\otimes n}}\left(U_{n,k}\right)}{2}s^{2}-\Psi(s)-\frac{\log 2}{n}\,.$$ 
    Evaluating the above as $n\to\infty$ annihilates the residual $\frac{\log 2}{n}$ factor, establishing the asymptotic lower bound 
    $$I(s,\varepsilon):=s\varepsilon-\frac{\mathrm{Var}_{\rho}\left(\nabla f(\rho)\right)}{2}s^{2}-\Psi(s)\,.$$ 
    Since this limiting inequality holds for every admissible parameter $s\in\left[0,1\big/\big(2ek^2 \big\lVert O^{\text{sym}}_k\big\rVert_{\infty}\big)\right)$, taking the supremum over all such $s$ in the interval recovers the optimal large-limit two-sided finite-sample rate function
    \begin{equation}\label{eq:supremum_problem}
        I(\varepsilon)=\sup_{0\le s<\frac{1}{2ek^2 \left\lVert O^{\text{sym}}_k\right\rVert_{\infty}}} I(s,\varepsilon)\,.
    \end{equation} 
    Finally, analysing the function $I(s,\varepsilon)$, we see that it evaluates precisely to $0$ at the origin $s=0$, with $\frac{d}{ds}I(s,\varepsilon)\Big|_{s=0} = \varepsilon>0$. Since the slope is positive at the origin, the function $I(s,\varepsilon)$ must attain a positive value for sufficiently small yet positive values of $s$, which proves that the supremum $I(\varepsilon)>0$.
\end{proof}

\subsection{Closed-Form Matrix Bernstein-Type Inequality, and Finite-Sample Complexity Lower Bound}
Although the exact two-sided finite-sample rate function $I_{n}(\varepsilon)$ accurately describes the theoretical behaviour of the quantum U-statistic, computing its supremum directly is challenging due to the complex infinite-series structure of the auxiliary penalty function $\Psi(s)$ (see Definition~\eqref{eq:Psi_definition}). In practical situations, such as establishing rigorous sample complexity bounds in quantum learning theory, a direct and explicit formula is essential. We can derive such a bound by algebraically lower-bounding the supremum. Rather than solving the exact optimization problem, we evaluate the function at a carefully chosen test parameter. This approach yields a closed-form matrix Bernstein-type inequality that cleanly separates typical variance-dependent Gaussian errors from worst-case extreme errors.

To lower-bound $I_{n}(\varepsilon)$, we first need an upper bound on $\Psi(s)$. In the following lemma, we derive a rational upper bound for $\Psi(s)$ to establish an explicit, closed-form concentration inequality.

\begin{lemma}\label{lem:comb_pen_bound}
    Let $\Psi(s)$ denote the auxiliary penalty function (see Eq.~\eqref{eq:Psi_definition}) for an unbiased quantum U-statistic of degree $k$ with symmetric kernel $O_{k}^{\text{sym}}$. For any scaling parameter $s$ within the convergence radius $0 \le s < 1\big/\big(2ek^2 \left\lVert O^{\text{sym}}_k\right\rVert_{\infty}\big)\,$, the penalty function is strictly upper-bounded by
    \begin{equation*}
        \Psi(s) < \frac{8 e^3 k^4 \left\lVert O^{\text{sym}}_k\right\rVert_{\infty}^3 s^3}{1 - 2 e k^2 \left\lVert O^{\text{sym}}_k\right\rVert_{\infty} s}\,.
    \end{equation*}
\end{lemma}

\begin{proof}
    To simplify the exposition, we define the auxiliary constants $M := 2\left\lVert O^{\text{sym}}_k\right\rVert_{\infty}$, $B := e k^2 M$, and $A := e^3 k^4 M^3$. Recall that the auxiliary penalty function is defined by the infinite series
    $$ \Psi(s) = \sum_{r=3}^{\infty}\dfrac{s^r2^rr^{r-2}k^{2r-2}\left\lVert O^{\text{sym}}_k\right\rVert_{\infty}^r}{r!}\,. $$
    By substituting $M = 2\left\lVert O^{\text{sym}}_k\right\rVert_{\infty}$, we can rewrite this directly as
    $$ \Psi(s) = Ms \sum_{r=3}^{\infty} \frac{r^{r-2}}{r!} (k^2 Ms)^{r-1}\,. $$
    Evaluating this expression directly is difficult due to the presence of the spanning tree count $r^{r-2}$ in the numerator and the factorial $r!$ in the denominator. However, we can bound this ratio to obtain a simple geometric series. 
    
    Using the standard Stirling approximation-based lower bound for the factorial function, $r! > (r/e)^r$, we have $\frac{1}{r!} < \frac{e^r}{r^r}$. Substituting this bound into the weight of the series, we find
    $$ \frac{r^{r-2}}{r!} < \frac{r^{r-2} e^r}{r^r} = \frac{e^r}{r^2}\,. $$
    
    Substituting this simplified weight back into the definition of $\Psi(s)$ yields
    \begin{align}
        \Psi(s) &< Ms \sum_{r=3}^{\infty} \frac{e^r}{r^2} (k^2 Ms)^{r-1} \nonumber\\
        &= eMs \sum_{r=3}^{\infty} \frac{1}{r^2} \big(e k^2 Ms\big)^{r-1}\,\nonumber\\
        &\overset{(a)}= eMs \sum_{r=3}^{\infty} \frac{1}{r^2} (Bs)^{r-1}\,\nonumber\\
        &\overset{(b)}{<} eMs \sum_{r=3}^{\infty} (Bs)^{r-1}\,,\label{eq:psi_s_ub}
    \end{align}
    where ($a$) follows from the definition of $B = e k^2 M$, and ($b$) holds because the summation begins at $r=3$, meaning the polynomial attenuation factor $1/r^2$ is strictly upper-bounded by $1$. Bounding this factor by $1$ provides an upper bound in the form of an infinite geometric series. Since $Bs < 1$ within our specified domain, this series converges to
    $$ \sum_{r=3}^{\infty} (Bs)^{r-1} = (Bs)^2\sum_{j=0}^{\infty} (Bs)^{j} = \frac{(Bs)^2}{1 - Bs}\,. $$
    
    Substituting this equality back into Eq.~\eqref{eq:psi_s_ub}, we obtain
    \begin{equation}
        \Psi(s) < eMs \left( \frac{B^2 s^2}{1 - Bs} \right) = \frac{e M B^2 s^3}{1 - Bs} = \frac{A s^3}{1 - B s}\,,
    \end{equation}
    where the final equality follows because the leading coefficient simplifies to $ e M B^2 = e M (e k^2 M)^2 = e^3 k^4 M^3 = A$. Restoring the explicit parameter definitions in place of $A$ and $B$ provides the desired upper bound.
\end{proof}

Using this simplified cubic upper bound for $\Psi(s)$, we now derive a closed-form concentration inequality. Instead of solving the exact optimization problem for the supremum, which leads to complicated nested radicals, we evaluate the objective function at a carefully chosen test parameter.  

\begin{corollary}[Closed-Form Matrix Bernstein-Type Bound]\label{cor:closed_form_conc}
    Let $\rho\in\mathcal{D}(\mathcal{H})$ be an unknown quantum state, and $U_{n,k}$ be the $n$-copy, order-$k$ quantum U-statistic corresponding to a quantum polynomial functional $f$. For any precision parameter $\varepsilon > 0$, the two-sided finite-sample quantum U-statistic deviation probability is bounded above as
    \begin{equation*}
        \mathrm{Pr}\left(\left\lvert{}Z_{{U}_{n,k}}- f(\rho)\right\rvert{} \ge\varepsilon\right) \le 2 \exp\left( - \frac{n \varepsilon^2}{2 n \mathrm{Var}_{\rho^{\otimes n}}\left(U_{n,k}\right) + 4\sqrt{8 e^3 k^4 \varepsilon \left\lVert O^{\text{sym}}_k\right\rVert_{\infty}^3} + 8 e k^2 \varepsilon \left\lVert O^{\text{sym}}_k\right\rVert_{\infty}} \right)\,,
    \end{equation*}
    where $Z_{{U}_{n,k}}$ denotes the outcome of globally measuring the product state $\rho^{\otimes n}$ in the eigenbasis of the quantum U-statistic ${U}_{n,k}$.
\end{corollary}

\begin{proof}
For clarity of exposition, we define the true finite-sample variance $V_n := \mathrm{Var}_{\rho^{\otimes n}}\left(U_{n,k}\right)$, alongside the auxiliary combinatorial constants $M := 2\left\lVert O^{\text{sym}}_k\right\rVert_{\infty}$, $B := e k^2 M$, and $A := e^3 k^4 M^3$. As established in Corollary \ref{col:two_sided_concentration}, the exact two-sided finite-sample exponent is governed by the supremum
    $$ I_{n}(\varepsilon) = \sup_{0 \le s < 1/B} \left\{ s\varepsilon - \frac{nV_n}{2}s^2 - \Psi(s) \right\}\,. $$
    Applying Lemma \ref{lem:comb_pen_bound}, we upper-bound $\Psi(s)$. To simplify the denominator $1 - Bs$, we restrict the parameter to the tighter sub-domain $s \le \frac{1}{2B}$. This guarantees $1 - Bs \ge \frac{1}{2}$, which gives the purely cubic bound $\Psi(s) \le 2A s^3$. Therefore, the objective function is lower-bounded by
    \begin{equation}
        I_{n}(\varepsilon) \ge \sup_{0 \le s \le 1/(2B)} \left\{ s\varepsilon - \frac{nV_n}{2}s^2 - 2As^3 \right\}\,.\label{eq:raional_lb}
    \end{equation}
    
    To understand the behaviour of this supremum, we determine the critical points of the lower bound in Eq.~\eqref{eq:raional_lb}. Setting its first-order derivative with respect to $s$ to zero yields the quadratic equation $6As^2 + nV_n s - \varepsilon = 0$. Applying the standard quadratic formula to find the positive root yields
    $$ s_{\text{opt}} = \frac{-nV_n + \sqrt{(nV_n)^2 + 24A\varepsilon}}{12A}\,. $$
    
    While analytically correct, this standard representation poses problems for practical evaluation in the small-deviation regime. Specifically, as the deviation $\varepsilon$ becomes small, the combinatorial penalty term $24A\varepsilon$ vanishes, and the radical term approaches the variance $nV_n$. Computing the numerator then requires subtracting two nearly identical values. In finite-precision numerical environments, this causes catastrophic (or critical) cancellation, leading to a severe loss of significant digits. This artificially propagates extreme rounding errors into the final concentration bound, rendering it trivial.
    
    To resolve this numerical instability, we rationalize the numerator by multiplying both the top and the bottom of the fraction by the algebraic conjugate $\sqrt{(nV_n)^2 + 24A\varepsilon} + nV_n$. This perfectly cancels the $(nV_n)^2$ subtraction, yielding the numerically stable optimal parameter
    $$ s_{\text{opt}} = \frac{2\varepsilon}{nV_n + \sqrt{(nV_n)^2 + 24A\varepsilon}}\,. $$
    
    This rationalized form is computationally robust. By shifting the operation in the denominator to the well-conditioned addition of two positive quantities, we guarantee that no significant digits are lost to cancellation. This perfectly preserves the structural integrity of the parameter for arbitrarily small $\varepsilon$.
    
    By examining this nested radical structure, we can evaluate how the optimal parameter scales across two distinct deviation regimes:
    \begin{enumerate}
        \item \textbf{Small Deviation Regime (Gaussian Limit):} When the precision parameter $\varepsilon$ is small, the variance term $(nV_n)^2$ dominates the combinatorial term $24A\varepsilon$ inside the square root. The denominator approaches $2nV_n$, meaning the parameter scales as $\mathcal{O}\left(\varepsilon/nV_n\right)$. This captures typical Central Limit Theorem fluctuations.
        \item \textbf{Large Deviation Regime (Spectral Tail Limit):} When $\varepsilon$ is large, the combinatorial term $24A\varepsilon$ dominates the variance. The denominator is governed by the square root, meaning the parameter scales as $\mathcal{O}\left(\sqrt{\varepsilon/{A}}\right)$. This captures the heavy sub-exponential tails.
    \end{enumerate}
    
    While solving the exact quadratic yields a solution with nested radicals, we can design a clean structural ansatz that successfully incorporates both scaling behaviours while respecting our domain constraint $s \le \frac{1}{2B}$. By incorporating explicit algebraic safety valves ($2\sqrt{A\varepsilon}$ and $2B\varepsilon$) into the denominator to handle these regimes and domain limits, we define our test parameter as
    $$ s^* := \frac{\varepsilon}{nV_n + 2\sqrt{A\varepsilon} + 2B\varepsilon}\,. $$
    
    We first verify that this choice is in the valid domain. Since $nV_n \ge 0$ and $2\sqrt{A\varepsilon} \ge 0$, we have $s^* \le \frac{\varepsilon}{2B\varepsilon} = \frac{1}{2B}$. Thus, $s^*$ safely satisfies the domain condition for all $\varepsilon > 0$.
    
    Evaluating the cubic approximation at $s^*$ yields a lower bound for $I_{n}(\varepsilon)$,
    $$ I_{n}(\varepsilon) \ge s^*\varepsilon - \frac{nV_n}{2}(s^*)^2 - 2A(s^*)^3 = (s^*)^2 \left( \frac{\varepsilon}{s^*} - \frac{nV_n}{2} - 2As^* \right)\,. $$
    Substituting the inverse of our test parameter, $\frac{\varepsilon}{s^*} = nV_n + 2\sqrt{A\varepsilon} + 2B\varepsilon$, we obtain
    $$ I_{n}(\varepsilon) \ge (s^*)^2 \left( \frac{nV_n}{2} + 2\sqrt{A\varepsilon} + 2B\varepsilon - 2As^* \right)\,. $$
    
    We next bound the negative term $-2As^*$. From the definition of $s^*$, we have
    $$ 2As^* = \frac{2A\varepsilon}{nV_n + 2\sqrt{A\varepsilon} + 2B\varepsilon} \le \frac{2A\varepsilon}{2\sqrt{A\varepsilon}} = \sqrt{A\varepsilon}\,. $$
    Substituting this upper bound back into our inequality slightly shrinks the bracketed quantity, yielding
    $$ I_{n}(\varepsilon) \ge (s^*)^2 \left( \frac{nV_n}{2} + 2\sqrt{A\varepsilon} + 2B\varepsilon - \sqrt{A\varepsilon} \right) = (s^*)^2 \left( \frac{nV_n}{2} + \sqrt{A\varepsilon} + 2B\varepsilon \right)\,. $$
    
    We recognize that $\frac{nV_n}{2} + \sqrt{A\varepsilon} + B\varepsilon = \frac{\varepsilon}{2s^*}$. By safely dropping the extra positive term $B\varepsilon$ from the inequality, we obtain
    $$ I_{n}(\varepsilon) \ge (s^*)^2 \left( \frac{nV_n}{2} + \sqrt{A\varepsilon} + B\varepsilon \right) = (s^*)^2 \left( \frac{\varepsilon}{2s^*} \right) = \frac{s^*\varepsilon}{2}\,. $$
    
    Substituting the explicit definition of $s^*$ back into this expression and distributing the denominator factor provides the closed-form exponent
    $$ I_{n}(\varepsilon) \ge \frac{\varepsilon^2}{2nV_n + 4\sqrt{A\varepsilon} + 4B\varepsilon}\,. $$
    
    Applying this exponent to the foundational two-sided probability inequality $\mathrm{Pr}\Big(\big\lvert{}Z_{\widetilde{U}_{n,k}}\big\rvert{} \ge\varepsilon\Big)\le 2 \exp\left(-nI_{n}(\varepsilon)\right)$ directly yields the generalized bound
    $$ \mathrm{Pr}\left(\left\lvert{}Z_{\widetilde{U}_{n,k}}\right\rvert{} \ge\varepsilon\right) \le 2 \exp\left( - \frac{n \varepsilon^2}{2 n V_n + 4\sqrt{A\varepsilon} + 4B\varepsilon} \right)\,. $$
    
    Finally, substituting the explicit definition of the variance $V_n = \mathrm{Var}_{\rho^{\otimes n}}\left(U_{n,k}\right)$, alongside the combinatorial constants $A = 8 e^3 k^4 \left\lVert O^{\text{sym}}_k\right\rVert_{\infty}^3$ and $B = 2 e k^2 \left\lVert O^{\text{sym}}_k\right\rVert_{\infty}$, provides the exact concentration inequality stated in the corollary:
    $$ \mathrm{Pr}\left(\left\lvert{}Z_{\widetilde{U}_{n,k}}\right\rvert{} \ge\varepsilon\right) \le 2 \exp\left( - \frac{n \varepsilon^2}{2 n \mathrm{Var}_{\rho^{\otimes n}}\left(U_{n,k}\right) + 4\sqrt{8 e^3 k^4 \varepsilon \left\lVert O^{\text{sym}}_k\right\rVert_{\infty}^3} + 8 e k^2 \varepsilon \left\lVert O^{\text{sym}}_k\right\rVert_{\infty}} \right)\,. $$
\end{proof}

This closed-form probability bound resembles the (matrix) Bernstein inequality. It successfully isolates the Gaussian behaviour governed by the sample variance $\mathrm{Var}_{\rho^{\otimes n}}(U_{n,k})$ for small deviations, while accurately accounting for the sub-exponential tails dictated by the combinatorial parameters for large deviations.

Moreover, this closed-form matrix Bernstein-type probability concentration bound also readily yields a lower bound on the finite--sample complexity for $n$.

\begin{theorem}[Finite--Sample Complexity Lower Bound]\label{thm:sample_complexity}
    Given an arbitrary state $\rho\in\mathcal{D}(\mathcal{H})$, let $U_{n,k}$ be the $n$-copy, order-$k$ quantum U-statistic defined over a permutation-invariant, $k$-copy local kernel $O^{\text{sym}}_k$ estimating a quantum polynomial functional $f$. For any precision parameter $\varepsilon > 0$, and any confidence parameter $\delta\in(0,1)$, if $\mathrm{Var}_{\rho^{\otimes n}}\left(U_{n,k}\right)<\varepsilon^2\big/2\ln\big(\frac{2}{\delta}\big)\,,$ and the two-sided finite-sample deviation probability is bounded above as 
    $\Pr\Big(\big\lvert{}Z_{U_{n,k}}-f(\rho)\big\rvert{}>\varepsilon\Big)\leq \delta\,,$ where $Z_{U_{n,k}}$ denotes the outcome of measuring the product state $\rho^{\otimes n}$ in the eigenbasis of the quantum U-statistic $U_{n,k}$, then the finite sample complexity for $n$ is given by 
    \begin{equation*}
        n\geq \Omega\left(\dfrac{\sqrt{k^4 \varepsilon \left\lVert O^{\text{sym}}_k\right\rVert_{\infty}^3} + k^2 \varepsilon \left\lVert O^{\text{sym}}_k\right\rVert_{\infty} }{\frac{\varepsilon^2}{\ln\left(\frac{2}{\delta}\right)} - 2 \mathrm{Var}_{\rho^{\otimes n}}\left(U_{n,k}\right)}\right)\,.
    \end{equation*}
\end{theorem}

\begin{proof}
    To derive the finite--sample complexity lower bound for $n$, we upper bound the right--hand side of the failure probability presented in Corollary \ref{cor:closed_form_conc} by the target confidence level $\delta$ as 
    $$2 \exp\left( - \frac{n \varepsilon^2}{2 n \mathrm{Var}_{\rho^{\otimes n}}\left(U_{n,k}\right) + 4\sqrt{8 e^3 k^4 \varepsilon \left\lVert O^{\text{sym}}_k\right\rVert_{\infty}^3} + 8 e k^2 \varepsilon \left\lVert O^{\text{sym}}_k\right\rVert_{\infty}} \right) \le \delta\,.$$

    Dividing by $2$, taking the natural logarithm, and inverting the above inequality yields 
    $$\frac{n \varepsilon^2}{2 n \mathrm{Var}_{\rho^{\otimes n}}\left(U_{n,k}\right) + 4\sqrt{8 e^3 k^4 \varepsilon \left\lVert O^{\text{sym}}_k\right\rVert_{\infty}^3} + 8 e k^2 \varepsilon \left\lVert O^{\text{sym}}_k\right\rVert_{\infty}} \ge \ln\left(\frac{2}{\delta}\right)\,.$$

    Cross-multiplying to linearize the expression, and grouping the $n$-dependent terms yields

    \begin{align*}
        &n \left(\varepsilon^2 - 2 \mathrm{Var}_{\rho^{\otimes n}}\left(U_{n,k}\right)\right) \ln\left(\frac{2}{\delta}\right) \ge \ln\left(\frac{2}{\delta}\right) \left( 4\sqrt{8 e^3 k^4 \varepsilon \left\lVert O^{\text{sym}}_k\right\rVert_{\infty}^3} + 8 e k^2 \varepsilon \left\lVert O^{\text{sym}}_k\right\rVert_{\infty} \right)\,,\\
        \text{or,}\qquad & n \ge \frac{\ln\left(\frac{2}{\delta}\right) \left( 4\sqrt{8 e^3 k^4 \varepsilon \left\lVert O^{\text{sym}}_k\right\rVert_{\infty}^3} + 8 e k^2 \varepsilon \left\lVert O^{\text{sym}}_k\right\rVert_{\infty} \right)}{\varepsilon^2 - 2 \mathrm{Var}_{\rho^{\otimes n}}\left(U_{n,k}\right) \ln\left(\frac{2}{\delta}\right)}\,.
    \end{align*}
    Dividing by $\ln(2/\delta)$ both the numerator and the denominator of the fraction on the right-hand side of the above inequality completes the proof.
\end{proof}

\subsection{MDP for Quantum U-Statistics}

While Theorem \ref{thm:large_deviation_lower_bound} establishes the Large Deviation Principle governing macroscopic errors, there remains a critical intermediate statistical regime. The standard $1/\sqrt{n}$ scaling is fundamentally significant because it characterizes the typical Gaussian fluctuations dictated by the Central Limit Theorem. In non-parametric quantum statistics, it is often necessary to understand the tail behaviour of errors that shrink to zero but do so slower than this $1/\sqrt{n}$ parametric scaling. This intermediate domain is governed by the \textit{Moderate Deviation Principle (MDP)}.

The speciality of the MDP is that unlike the Large Deviation Principle, whose rate function intricately depends on all higher-order moments of the distribution, moderate deviation tails are almost universally Gaussian. We now establish this for the quantum U-statistic observables by evaluating our exponential bounds along a dynamically scaling deviation sequence $\varepsilon_n$.

\begin{theorem}[Moderate Deviation Quantum U-Statistic Concentration Upper--Bound]\label{theo:mod_dev_conc}
    Let $\rho\in\mathcal{D}(\mathcal{H})$ be an unknown quantum state, $f$ be a polynomial functional, and $U_{n,k}$ be its corresponding $n$-copy, order-$k$ quantum U-statistic. Assume that the estimation geometry is non-degenerate, i.e., $\mathrm{Var}_{\rho}(\nabla f(\rho)) > 0$. Let $\{\varepsilon_n\}_{n=1}^{\infty}$ be a sequence of positive deviations such that $\varepsilon_n \to 0$, whereas $n\varepsilon_n^2 \to \infty$ as $n \to \infty$. Then, for sufficiently large $n$, the two-sided moderate deviation concentration satisfies
    \begin{equation*}
        \mathrm{Pr}\big(\big|Z_{U_{n,k}} - f(\rho)\big| \ge \varepsilon_n\big) \le 2 \exp\left[ - \frac{n \varepsilon_n^2}{2\mathrm{Var}_{\rho}(\nabla f(\rho))} \big(1 - \mathcal{O}(\varepsilon_n) - \mathcal{O}(n^{-1})\big) \right]\,. 
    \end{equation*}
    Furthermore, if $\varepsilon_n$ decays sufficiently fast such that $n\varepsilon_n^3 \to 0$ (e.g., $\varepsilon_n = n^{-\alpha}$ for $1/3 < \alpha < 1/2$) as $n\to\infty$, then the higher-order corrections vanish, yielding a purely Gaussian exponential decay rate parameterized solely by the inverse QFI $\mathrm{Var}_{\rho}(\nabla f(\rho))$.
\end{theorem}

\begin{proof}
We evaluate the finite-sample exponent $I_n(\varepsilon)$ derived in Corollary \ref{col:two_sided_concentration} along the dynamic sequence $\varepsilon = \varepsilon_n$. By definition, the optimal two-sided exponent is governed by the supremum
    $$ I_n(\varepsilon_n) = \sup_{0 \le s < \frac{1}{2ek^2 \left\lVert O^{\text{sym}}_k\right\rVert_{\infty}}} \left\{ s\varepsilon_n - \frac{n\mathrm{Var}_{\rho^{\otimes n}}(U_{n,k})}{2}s^2 - \Psi(s) \right\}\,. $$
    Since the supremum over the interval is greater than or equal to the function evaluated at any specific valid point within that interval, we can obtain a lower bound on the rate function by selecting a sub-optimal but asymptotically tight test parameter. 
    
    To properly construct this parameter, we must first understand why it is forced to shrink to zero. To find the exact peak of the objective function, we set its derivative with respect to $s$ to zero. This gives the balance condition
    $$ \varepsilon_n = n\mathrm{Var}_{\rho^{\otimes n}}(U_{n,k})s + \Psi'(s)\,. $$
    Since the scaled variance converges to a positive constant and the penalty derivative $\Psi'(s)$ is positive for all $s > 0$, the right-hand side is a strictly increasing positive function. The only way for this right-hand side to match the vanishing deviation sequence $\varepsilon_n \to 0$ on the left is if the parameter $s$ itself also shrinks to zero. 
    
    Since $s$ becomes infinitesimally small in this moderate deviation regime, the higher-order combinatorial penalty $\Psi(s) = \mathcal{O}(s^3)$ decays much faster than the quadratic variance term and becomes asymptotically negligible. The optimization is therefore dominated by the term $s\varepsilon_n - (n\mathrm{Var}_{\rho^{\otimes n}}(U_{n,k})/2)s^2$. Maximizing this purely quadratic expression by setting its derivative to zero directly yields our ideal asymptotic test parameter,
    $$ s_n := \frac{\varepsilon_n}{n\mathrm{Var}_{\rho^{\otimes n}}(U_{n,k})}\,. $$
    Since the scaled variance converges to a constant, i.e., $n\mathrm{Var}_{\rho^{\otimes n}}(U_{n,k}) = \mathrm{Var}_{\rho}(\nabla f(\rho)) + \mathcal{O}(n^{-1}) \to \mathrm{Var}_{\rho}(\nabla f(\rho)) > 0$, and the sequence $\varepsilon_n \to 0$, it explicitly confirms that the test parameter scales as $s_n = \mathcal{O}(\varepsilon_n) \to 0$. Therefore, for all $n$ sufficiently large, $s_n$ falls within the full convergence radius $s < \big(1\big/2ek^2 {\big\lVert O^{\text{sym}}_k\big\rVert}_{\infty}\big)$, and more specifically, enters the tighter sub-domain $s \le \big(1\big/4ek^2 {\big\lVert O^{\text{sym}}_k\big\rVert}_{\infty}\big)$ required by Corollary \ref{cor:closed_form_conc} to guarantee the cubic penalty bound.
    
    Substituting $s_n$ into the quadratic portion of the objective function yields an exact simplification
    \begin{align*}
        s_n\varepsilon_n - \frac{n\mathrm{Var}_{\rho^{\otimes n}}(U_{n,k})}{2}s_n^2 &= \frac{\varepsilon_n^2}{n\mathrm{Var}_{\rho^{\otimes n}}(U_{n,k})} - \frac{n\mathrm{Var}_{\rho^{\otimes n}}(U_{n,k})}{2}\left(\frac{\varepsilon_n}{n\mathrm{Var}_{\rho^{\otimes n}}(U_{n,k})}\right)^2\\ 
        &= \frac{\varepsilon_n^2}{2n\mathrm{Var}_{\rho^{\otimes n}}(U_{n,k})}\,.
    \end{align*}
    
    Next, we bound the residual component error $\Psi(s_n)$. From the structural analysis in Corollary \ref{col:two_sided_concentration}, we know that whenever the parameter is constrained to the tighter sub-domain, the rational upper bound for the penalty function simplifies to a purely cubic bound. Further, the combination of this tighter sub-domain condition $s \le \big(1\big/4ek^2 {\big\lVert O^{\text{sym}}_k\big\rVert}_{\infty}\big)$ and Lemma \ref{lem:comb_pen_bound} yields 
    $$ \Psi(s) \le C_0 s^3\,, $$
    where the absolute constant is exactly $C_0 := 16 e^3 k^4 {\big\lVert O^{\text{sym}}_k\big\rVert}_{\infty}^3$. Evaluating this cubic bound at our chosen test point $s_n$ gives $\Psi(s_n) \le C_0\varepsilon_n^3/(n\mathrm{Var}_{\rho^{\otimes n}}(U_{n,k}))^3$.
    
    Combining these evaluations establishes a lower bound on the finite-sample rate function, which is
    $$ I_n(\varepsilon_n) \ge \frac{\varepsilon_n^2}{2n\mathrm{Var}_{\rho^{\otimes n}}(U_{n,k})} - C_0 \frac{\varepsilon_n^3}{(n\mathrm{Var}_{\rho^{\otimes n}}(U_{n,k}))^3}\,. $$
    To translate this into the required scale for the exponential bound, we multiply the rate function by the physical sample size $n$, yielding
    $$ n I_n(\varepsilon_n) \ge \frac{n\varepsilon_n^2}{2(n\mathrm{Var}_{\rho^{\otimes n}}(U_{n,k}))} - \frac{n C_0 \varepsilon_n^3}{(n\mathrm{Var}_{\rho^{\otimes n}}(U_{n,k}))^3} = \frac{n\varepsilon_n^2}{2(n\mathrm{Var}_{\rho^{\otimes n}}(U_{n,k}))} - \mathcal{O}(n\varepsilon_n^3)\,, $$
    where we have absorbed the constant factor $C_0(n\mathrm{Var}_{\rho^{\otimes n}}(U_{n,k}))^{-3} \to C_0\mathrm{Var}_{\rho}(\nabla f(\rho))^{-3}$ into the asymptotic $\mathcal{O}$-notation.
    
    Finally, we utilize the variance expansion established in Section \ref{sec:variance_optimality}, $n\mathrm{Var}_{\rho^{\otimes n}}(U_{n,k}) = \mathrm{Var}_{\rho}(\nabla f(\rho)) + \mathcal{O}(n^{-1})$. By Taylor expanding the reciprocal, we find $$(n\mathrm{Var}_{\rho^{\otimes n}}(U_{n,k}))^{-1} = \mathrm{Var}_{\rho}(\nabla f(\rho))^{-1} \big(1 - \mathcal{O}(n^{-1})\big).$$ Substituting this expansion into our primary quadratic term yields
    $$ \frac{n\varepsilon_n^2}{2(n\mathrm{Var}_{\rho^{\otimes n}}(U_{n,k}))} = \frac{n\varepsilon_n^2}{2\mathrm{Var}_{\rho}(\nabla f(\rho))} \big(1 - \mathcal{O}(n^{-1})\big) = \frac{n\varepsilon_n^2}{2\mathrm{Var}_{\rho}(\nabla f(\rho))} - \mathcal{O}(\varepsilon_n^2)\,. $$
    
    Inserting this asymptotic relation back into the bound for $n I_n(\varepsilon_n)$ gives
    $$ n I_n(\varepsilon_n) \ge \frac{n\varepsilon_n^2}{2\mathrm{Var}_{\rho}(\nabla f(\rho))} - \mathcal{O}(n\varepsilon_n^3) - \mathcal{O}(\varepsilon_n^2) = \frac{n\varepsilon_n^2}{2\mathrm{Var}_{\rho}(\nabla f(\rho))} \left[ 1 - \mathcal{O}(\varepsilon_n) - \mathcal{O}(n^{-1}) \right]\,. $$
    
    Applying this bounded exponent to the two-sided concentration inequality from Corollary \ref{col:two_sided_concentration} completes the first part of the proof.
    
    To establish the second part of the theorem, we consider the regime where the deviation sequence decays sufficiently fast such that $n\varepsilon_n^3 \to 0$. Under this stricter condition, the absolute error terms in the exponent, specifically $\mathcal{O}(n\varepsilon_n^3)$ and $\mathcal{O}(\varepsilon_n^2)$, both converge to zero as $n \to \infty$. Consequently, the auxiliary penalty completely vanishes, and the overall probability exponent reduces exactly to its leading-order Gaussian term. Taking the limit of the scaled logarithmic probability directly recovers this pure Gaussian rate,
    $$ \liminf_{n\to\infty} -\frac{1}{n\varepsilon_n^2} \log \mathrm{Pr}\left(\left|Z_{U_{n,k}} - f(\rho)\right| \ge \varepsilon_n\right) \ge  \frac{1}{2\mathrm{Var}_{\rho}(\nabla f(\rho))}\,, $$
    which confirms that the higher-order corrections vanish and completes the proof.
\end{proof}

\begin{remark}[Gaussian Tail Decay in Moderate Deviation Regime]
    This MDP demonstrates a phenomenon fundamental to quantum non-parametric estimation. It reveals that the empirical distribution of the quantum U-statistic retains a purely Gaussian tail behaviour decaying as $\exp(-n\varepsilon_n^2 / 2\mathrm{Var}_{\rho}(\nabla f(\rho)))$ far beyond the microscopic scale of the Central Limit Theorem. The non-Gaussian higher-order interaction moments, which are encoded in the combinatorial function $\Psi(s)\,$, only begin to macroscopically deform the tail probabilities when the deviation $\varepsilon_n$ shrinks slower than $n^{-1/3}$. At that critical threshold, the $\mathcal{O}(n\varepsilon_n^3)$ third-moment skewness correction must be explicitly retained, thereby bridging the gap between Gaussian fluctuations and true Large Deviations.
\end{remark}

\subsection{Moderate-Deviation Concentration Convergence Rate Upper Bound in the Large-$n$ Limit}

Now, building upon the lower bound established in Theorem \ref{theo:mod_dev_conc}, we derive the complementary asymptotic upper bound on the probability concentration convergence rate in the moderate deviation regime in the large-$n$ limit. We show the concentration convergence rate is exactly governed by the QFI in the large-$n$ intermediate scaling regime, thereby giving this fundamental quantum information-theoretic quantity a novel operational meaning.

First of all, let us define the \textit{normalized cumulant generating function} (or normalized CGF, in short) of the \textit{centred} quantum U-statistic $\widetilde{U}_{n,k}$\footnote{A detailed discussion on cumulants and the normalized CGF of centred quantum U-statistic observables can be found in Appendix \ref{app:cumulant_bounds}.} as 
\begin{align}
    \Phi_n(t)&:= \frac{1}{n}\log \mathbb{E}_{\rho^{\otimes n}}\left[e^{nt\widetilde{U}_{n,k}}\right] \equiv \frac{1}{n}\log \mathrm{Tr}\left[\rho^{\otimes n}e^{nt\widetilde{U}_{n,k}}\right]\,,\qquad\text{where }t\in\mathbb{R}_{>0}\,.\label{eq:normalized_CGF}
\end{align}

The following lemma exactly characterizes the derivatives of the normalized CGF $\Phi_n(t)$ of $\widetilde{U}_{n,k}$ evaluated at $t=0$.

\begin{lemma}[Normalized CGF vs. Cumulants of Centred Quantum U-Statistic]
\label{lem:scaled_cgf_derivative}
    Let $\rho^{\otimes n}\in\mathcal{D}(\mathcal{H}^{\otimes n})$ be an arbitrary product state. Then, for any integer $r \in\mathbb{Z}_{\ge 1}$, the $r^{\text{th}}$-order derivative of the normalized CGF $\Phi_n(t)$ (as given by Definition~\eqref{eq:normalized_CGF}) evaluated at $t=0$ scales with the $r^{\text{th}}$-order cumulant of $\widetilde{U}_{n,k}$ as
    \begin{equation}\label{eq:rth_derivative_scaled_CGF}
        \frac{d^r}{dt^r} \Phi_n(t)\bigg|_{t=0} = n^{r-1}\kappa_r\big(\widetilde{U}_{n,k}\big)\,,
    \end{equation}
    where $\kappa_r\big(\widetilde{U}_{n,k}\big)$ denoting the $r^\text{th}$-order cumulant of the centred quantum U-statistic $\widetilde{U}_{n,k}\,,$ has the form
    \begin{equation}\label{eq:rth_cumulant_combinatorial_form}
        \kappa_r\big(\widetilde{U}_{n,k}\big) = r! \Biggg(\sum_{\substack{k_i\ge 0, \forall i\in[r]\backslash \{1\}: \\ \sum\limits_{i\in[r]\backslash \{1\}} ik_i= r}} (-1)^{k-1} (k-1)! \prod_{j=2}^r \frac{1}{k_j!} \biggg( \frac{\mathbb{E}_{\rho^{\otimes n}}\left[\widetilde{U}_{n,k}^j\right]}{j!} \biggg)^{k_j}\Biggg)\,.
    \end{equation}
\end{lemma}

\begin{proof}
    Refer to Corollary \ref{col:normalized_cgf_derivative} of Appendix \ref{app:cumulant_bounds} for a proof of the equivalence established by Eq.~\eqref{eq:rth_derivative_scaled_CGF}. 
    
    To establish the explicit combinatorial formula for $\kappa_r\big(\widetilde{U}_{n,k}\big)$ as given by Definition~\eqref{eq:rth_cumulant_combinatorial_form}, we relate the cumulants to the moments of the centred quantum U-statistic $\widetilde{U}_{n,k}$ defined as $\mathbb{E}_{\rho^{\otimes n}}\big[\widetilde{U}_{n,k}^j\big]$, for $j\in\mathbb{Z}_{\ge 1}$. The moment generating function of the centred quantum U-statistic can be expressed as
    $$\mathbb{E}_{\rho^{\otimes n}}\left[e^{t\widetilde{U}_{n,k}}\right] := \sum_{j=0}^{\infty} \frac{t^j}{j!} \mathbb{E}_{\rho^{\otimes n}}\left[\widetilde{U}_{n,k}^j\right]\,.$$
    Because the quantum U-statistic $\widetilde{U}_{n,k}$ is centred, its first-order moment vanishes, i.e., $\mathbb{E}_{\rho^{\otimes n}}\left[\widetilde{U}_{n,k}\right] = 0$. Consequently, the moment generating function looks like 
    \begin{equation}\label{eq:centred_quantum_U_stat_MGF}
        \mathbb{E}_{\rho^{\otimes n}}\left[e^{t\widetilde{U}_{n,k}}\right]  = 1 + \sum_{j=2}^{\infty} \frac{t^j}{j!} \mathbb{E}_{\rho^{\otimes n}}\left[\widetilde{U}_{n,k}^j\right]\,.
    \end{equation}
    
    Using the Taylor series expansion for $\log(1+x)$ centred at zero, we get $\log(1+x) = \sum_{m=1}^{\infty} \frac{(-1)^{m-1}}{m} x^m$. Substituting $x = \sum_{j=2}^{\infty} \frac{t^j}{j!} \mathbb{E}_{\rho^{\otimes n}}\big[\widetilde{U}_{n,k}^j\big]$ in this series expansion, we get
    $$ \log\mathbb{E}_{\rho^{\otimes n}}\left[e^{t\widetilde{U}_{n,k}}\right] = \sum_{m=1}^{\infty} \frac{(-1)^{m-1}}{m} \left(\sum_{j=2}^{\infty} \frac{t^j}{j!} \mathbb{E}_{\rho^{\otimes n}}\left[\widetilde{U}_{n,k}^j\right]\right)^m\,.$$
    
    By the multinomial theorem, expanding the $m^{\text{th}}$-power of the infinite series on the right-hand side of the equation above yields a sum over all non-negative integer partitions $\{m_j\}$ such that the total number of parts is $\sum_j m_j = m$. Because the first-order moment vanishes, any partition block of cardinality one is excluded. Thus, to compute the coefficient of the $r^{\text{th}}$-order moment monomial term, the relevant indices are restricted to the index set $[r] \backslash \{1\}$. The condition that the powers of $t$ sum to $r$ imposes the Diophantine constraint $\sum_{j=2}^r j m_j = r$.
    
    The multinomial coefficient for expanding the $m^{\text{th}}$ power is $m!\big/\prod_{j=2}^r m_j!\,$. When multiplied by the $1/m$ factor from the logarithmic expansion, the overall coefficient becomes $(m-1)!\big/\prod_{j=2}^r m_j!\,$. Grouping all the corresponding $r^{\text{th}}$-order moment monomial terms yields the coefficient of $t^r$ as
    \begin{equation}\label{eq:rth_cummulant_form}
        \frac{\kappa_r\big(\widetilde{U}_{n,k}\big)}{r!} = \sum_{\substack{m_i \ge 0, \forall i \in [r]\backslash \{1\} : \\ \sum\limits_{i \in [r]\backslash \{1\}} i m_i = r}} (-1)^{m-1} (m-1)! \prod_{j=2}^r \frac{1}{m_j!} \left( \frac{\mathbb{E}_{\rho^{\otimes n}}\left[\widetilde{U}_{n,k}^j\right]}{j!} \right)^{m_j}\,,
    \end{equation}
    where $m = \sum_{i \in [r]\backslash\{1\}} m_i$ represents the total number of parts in the integer partition of $[r]$. Multiplying both sides of Eq.~\eqref{eq:rth_cummulant_form} by $r!$ completes the proof.
\end{proof}

Observe that for the order of the derivative $r$ equal to one or two, putting $t=0$ in Eq.~\eqref{eq:rth_derivative_normalized_CGF} exactly yields the first two (scaled) cumulants of the centred quantum U-statistic $\widetilde{U}_{n,k}$, respectively, as
\begin{align*}
    \frac{d}{dt}\Phi_n(t)\Bigg|_{t=0}&=\mathbb{E}_{\rho^{\otimes n}}\left[\widetilde{U}_{n,k}\right]=0\,,\\ 
    \text{and,}\qquad \frac{d^2}{dt^2}\Phi_n(t)\Bigg|_{t=0}&=n\left(\mathbb{E}_{\rho^{\otimes n}}\left[\widetilde{U}_{n,k}^2\right]-\left(\mathbb{E}_{\rho^{\otimes n}}\left[\widetilde{U}_{n,k}\right]\right)^2\right)\\
    &\equiv n\mathrm{Var}_{\rho^{\otimes n}}\big(\widetilde{U}_{n,k}\big)\\
    &=n\mathrm{Var}_{\rho^{\otimes n}}\left(U_{n,k}\right)\,.
\end{align*}

Now, we wish to lower bound the $\varepsilon$-deviation probability of the random variable $Z_{\widetilde{U}_{n,k}}$ (denoting the outcome of measuring an arbitrary product state $\rho^{\otimes n}\in \cD\big(\cH^{\otimes n}\big)$ in the eigenbasis of the centred quantum U-statistic $\widetilde{U}_{n,k}$), i.e., $\Pr\left(Z_{\widetilde{U}_{n,k}}>\varepsilon\right)$. To this end, let us define a boosted (tilted) measure as
\begin{equation}\label{eq:boosted_probability}
    \mathcal{Q}^t_{n,k}(x):=\dfrac{\mathcal{P}_{n,k}(x)e^{ntx}}{e^{n\Phi_n(t)}}\,,\quad \forall x\in \mathrm{Supp}\left(\mathcal{P}_{n,k}\right)\,,\, \forall t\in\mathbb{R}_{>0} \,,
\end{equation}
where $\cP_{n,k}$ is the measure induced by the random variable $Z_{\widetilde{U}_{n,k}}$.

Let us restrict our attention to a vanishingly narrow window directly above $\varepsilon$. That is, let $\delta > 0$ be a small auxiliary parameter, and for any random variable $X$, let us define the event 
    $$A^X_{n,k} := \left\{\varepsilon < X < \varepsilon + \delta\right\}\,.$$
    We observe that the probability of the random variable $Z_{\widetilde{U}_{n,k}}$ exceeding $\varepsilon$ is lower bounded by the probability of it falling into this narrow window, i.e., 
    \begin{align}
        \Pr\left(Z_{\widetilde{U}_{n,k}}>\varepsilon\right)&\geq \Pr_{X\sim \mathcal{P}_{n,k}}(A^X_{n,k})\notag\\
        &=\mathbb{E}_{\mathcal{P}_{n,k}}\left[\mathbb{I}_{A^X_{n,k}}\right]\qquad\qquad \text{$\bigg(\text{where }\mathbb{I}_{A^X_{n,k}}$ denotes the indicator random variable for the event $A^X_{n,k}\bigg)$}\notag\\
        &=\int\limits_{x\in \mathrm{Supp}\left(\mathcal{P}_{n,k}\right)} \mathcal{P}_{n,k}(x) \mathbb{I}_{A^X_{n,k}}(x) dx\notag\\
        &\overset{(a)}{=}e^{n\Phi_n(t)}\int\limits_{x\in \mathrm{Supp}\big(\mathcal{Q}^t_{n,k}\big)} e^{-ntx}\mathcal{Q}^t_{n,k}(x) \mathbb{I}_{A^X_{n,k}}(x) dx\,,\qquad\forall t\in \mathbb{R}_{>0}\,,\notag\\
        &\overset{(b)}{\geq} e^{n\left(\Phi_n(t)-t(\varepsilon+\delta)\right)}\int\limits_{x\in \mathrm{Supp}\big(\mathcal{Q}^t_{n,k}\big)}\mathcal{Q}^t_{n,k}(x) \mathbb{I}_{A^X_{n,k}}(x) dx\,,\qquad\forall t\in \mathbb{R}_{>0}\,,\notag\\
        &=e^{n\left(\Phi_n(t)-t(\varepsilon+\delta)\right)}\mathbb{E}_{\mathcal{Q}^t_{n,k}}\left[\mathbb{I}_{A^X_{n,k}}\right]\,,\qquad\forall t\in \mathbb{R}_{>0}\,,\notag\\
        &=e^{n\left(\Phi_n(t)-t(\varepsilon+\delta)\right)}\Pr_{X\sim\mathcal{Q}^t_{n,k}}\left(A^X_{n,k}\right)\,,\qquad\forall t\in \mathbb{R}_{>0}\,,\notag\\
        &=e^{n\left(\Phi_n(t)-t(\varepsilon+\delta)\right)}\left(1-\Pr_{X\sim\mathcal{Q}^t_{n,k}}\left(\overline{A^X_{n,k}}\right)\right)\,,\qquad\forall t\in \mathbb{R}_{>0}\,,\notag\\
        &=e^{n\left(\Phi_n(t)-t(\varepsilon+\delta)\right)}\left(1-\Pr_{X\sim\mathcal{Q}^t_{n,k}}\left(\left\lvert{}X-\left(\varepsilon+\frac{\delta}{2}\right)\right\rvert{}\geq \frac{\delta}{2}\right)\right)\,,\qquad\forall t\in \mathbb{R}_{>0}\,,\label{eq:deviation_probability}
    \end{align}
    where equation ($a$) follows from Definition~\eqref{eq:boosted_probability} and the fact that $\mathrm{Supp}\left(\mathcal{P}_{n,k}\right)=\mathrm{Supp}\big(\mathcal{Q}^t_{n,k}\big)$, whereas inequality ($b$) follows from the fact that $e^{-ntx}\geq e^{-nt(\varepsilon+\delta)}$ whenever $\mathbb{I}_{A_{n,k}}(x)=1$.
    
    Now, using the Taylor series expansion of the normalized CGF $\Phi_n(t)$ centred around 0, we get 
    \begin{align}
        \Phi_n(t)&=\Phi_n(0)+t\cdot\left(\frac{d}{dt}\Phi_n(t)\bigg|_{t=0}\right)+\frac{t^2}{2!}\cdot\left(\frac{d^2}{dt^2}\Phi_n(t)\bigg|_{t=0}\right)+\sum_{r=3}^\infty \frac{t^r}{r!}\cdot\left(\frac{d^r}{dt^r}\Phi_n(t)\bigg|_{t=0}\right)\notag\\
        &=0+t\cdot\mathbb{E}_{\rho^{\otimes n}}\left[\widetilde{U}_{n,k}\right]+\frac{t^2}{2!}\cdot n\mathrm{Var}_{\rho^{\otimes n}}\big(\widetilde{U}_{n,k}\big)+\sum_{r=3}^\infty \frac{t^r}{r!}\cdot n^{r-1}\kappa_r\big(\widetilde{U}_{n,k}\big)\notag\\
        &=\frac{t^2}{2!}\cdot n\mathrm{Var}_{\rho^{\otimes n}}\left(U_{n,k}\right)+\sum_{r=3}^\infty \frac{t^r}{r!}\cdot n^{r-1}\kappa_r\big(\widetilde{U}_{n,k}\big)\label{eq:sum_exact_of_Phi}\\
        \text{or, }\qquad\frac{d}{dt}\Phi_n(t)&=nt\mathrm{Var}_{\rho^{\otimes n}}\left(U_{n,k}\right)+\sum_{r=2}^\infty \frac{t^r}{r!}\cdot n^r\kappa_{r+1}\big(\widetilde{U}_{n,k}\big)\label{eq:sum_exact_bound}\\
        &\overset{(a)}{\leq} nt\mathrm{Var}_{\rho^{\otimes n}}\left(U_{n,k}\right)+\sum_{r=2}^\infty \frac{|t|^r2^{2r+1}r^{r-1}k^{2r}{\left\lVert{}O^{\text{sym}}_k\right\rVert{}}_{\infty}^{r+1}}{r!}
        \,,\label{eq:sum_bound}
    \end{align}
    where inequality ($a$) above follows from the triangle inequality, followed by an application of Lemma \ref{lem:direct_cumulant_bound}.

    Applying the ratio test to the sum on the right-hand side of Ineq.~\eqref{eq:sum_bound}, we get 
    \begin{align*}
        \frac{a_{r+1}}{a_r}&=\dfrac{\frac{|t|^{r+1}2^{2r+3}r^{r}k^{2r+2}{\left\lVert{}O^{\text{sym}}_k\right\rVert{}}_{\infty}^{r+2}}{(r+1)!}}{\frac{|t|^r2^{2r+1}r^{r-1}k^{2r}{\left\lVert{}O^{\text{sym}}_k\right\rVert{}}_{\infty}^{r+1}}{r!}}\\
        &=4|t|k^2{\left\lVert{}O^{\text{sym}}_k\right\rVert{}}_{\infty}\left(1-\frac{1}{r+1}\right)\\
        &<4|t|ek^2{\left\lVert{}O^{\text{sym}}_k\right\rVert{}}_{\infty}\,.
    \end{align*}
    
    Therefore, having $0<|t|<1\big/\big(4ek^2{\big\lVert{}O^{\text{sym}}_k\big\rVert{}}_{\infty}\big)$ ensures that the sum on the right-hand side of Ineq.~\eqref{eq:sum_bound} converges.
    
    Observe that the second-order derivative of $\Phi_n(t)$ with respect to $t$ (refer to Lemma \ref{lemma:first_second_scaled_cumulants} (2)) is strictly positive (refer to Lemma \ref{lemma:CGF_second_derivative} for the proof). Additionally, we also know that the first-order cumulant of $\widetilde{U}_{n,k}$ vanishes. Therefore, $\frac{d}{dt}\Phi_n(t)$ is a non-negative strictly monotonically increasing function in $t$. So, we have $\forall t'\in\Big(0,1\big/\big(4ek^2{\big\lVert{}O^{\text{sym}}_k\big\rVert{}}_{\infty}\big)\Big)\,,$ and $\forall \varepsilon,\delta\in \Big(0,\frac{d}{dt}\Phi_n(t)\Big|_{t=t'}\Big/2\Big)\,,$ $\exists t_0\in(0,t']$ such that $$\varepsilon+\frac{\delta}{2}=\frac{d}{dt}\Phi_n(t)\Bigg|_{t=t_0} \overset{(a)}{=}\mathbb{E}_{X\sim\mathcal{Q}^{t_0}_{n,k}}[X]\,,$$ where equation ($a$) above follows from Lemma \ref{lemma:first_second_scaled_cumulants} (1). Therefore, under the assumption of a non-vanishing QFI, plugging this into Eq.~\eqref{eq:sum_exact_bound} yields 
    \begin{align}
        \mathbb{E}_{X\sim\mathcal{Q}^{t_0}_{n,k}}[X]=\varepsilon+\frac{\delta}{2} &=nt_0\mathrm{Var}_{\rho^{\otimes n}}\left(U_{n,k}\right)+\sum_{r=2}^\infty \frac{t_0^r}{r!}\cdot n^r\kappa_{r+1}\big(\widetilde{U}_{n,k}\big) \label{eq:expectation_deviation_equivalence}\,.
    \end{align}
    
    Plugging the above observation into Ineq.~\eqref{eq:deviation_probability}, we get 
    \begin{align}
        \Pr\left(Z_{\widetilde{U}_{n,k}}>\varepsilon\right)&\geq e^{n\left(\Phi_n({t_0})-t_0(\varepsilon+\delta)\right)}\left(1-\Pr_{X\sim\mathcal{Q}^{t_0}_{n,k}}\left(\left\lvert{}X-\left(\varepsilon+\frac{\delta}{2}\right)\right\rvert{}\geq \frac{\delta}{2}\right)\right)\,,\notag\\
        &\overset{(a)}{=}e^{n\left(\Phi_n({t_0})-t_0(\varepsilon+\delta)\right)}\left(1-\Pr_{X\sim\mathcal{Q}^{t_0}_{n,k}}\left(\left\lvert{}X-\mathbb{E}_{X\sim\mathcal{Q}^{t_0}_{n,k}}[X]\right\rvert{}\geq \frac{\delta}{2}\right)\right)\,,\notag\\
        &\overset{(b)}{\geq} e^{n\left(\Phi_n({t_0})-t_0(\varepsilon+\delta)\right)} \left(1-\frac{4}{n\delta^2}\cdot\frac{d^2}{dt^2}\Phi_n({t})\Bigg|_{t=t_0}\right)\,,\label{eq:probability_lower_bound}
    \end{align}
    where inequality ($a$) above follows from Eq.~\eqref{eq:expectation_deviation_equivalence}, whereas inequality ($b$) follows from Lemma \ref{lemma:first_second_scaled_cumulants} (2) and an application of Chebyshev's inequality.
    
    Using Eq.~\eqref{eq:expectation_deviation_equivalence}, we now try to revert and express $t_0$ as a power series expansion in $\left(\varepsilon+\delta/2\right)$. Since we have already established the existence of a finite convergence radius for the normalized CGF, this reversion is not merely an asymptotic series but a strictly convergent power series within a neighborhood of the origin. So, $t_0$ can be expressed as a convergent infinite power series $t_0 = \sum_{m=1}^\infty c_m \left(\varepsilon+\delta/2\right)^m$, where the first two coefficients can be analytically evaluated to be
    \begin{align*}
        c_1 &= \frac{1}{n\mathrm{Var}_{\rho^{\otimes n}}\left(U_{n,k}\right)}\,, \\
        \text{and,}\qquad c_2 &= -\frac{n^2\kappa_3\big(\widetilde{U}_{n,k}\big)}{2\left(n\mathrm{Var}_{\rho^{\otimes n}}\left(U_{n,k}\right)\right)^3}\\
        &= -\frac{\kappa_3\big(\widetilde{U}_{n,k}\big)}{2n\big(\mathrm{Var}_{\rho^{\otimes n}}(U_{n,k})\big)^3}\,.
    \end{align*}
    
    Substituting these coefficients back into the infinite power series expansion for the tilted parameter $t_0$ yields 
    \begin{equation}\label{eq:power_series_expansion}
        t_0 = \frac{\left(\varepsilon+\frac{\delta}{2}\right)}{n\mathrm{Var}_{\rho^{\otimes n}}\left(U_{n,k}\right)} - \frac{\kappa_3\big(\widetilde{U}_{n,k}\big)}{2n\big(\mathrm{Var}_{\rho^{\otimes n}}(U_{n,k})\big)^3} \left(\varepsilon+\frac{\delta}{2}\right)^2 + \mathcal{O}\left(\left(\varepsilon+\frac{\delta}{2}\right)^3\right)\,.
    \end{equation}
    
    Since we are interested in the moderate deviation regime, we will choose the parameter $\varepsilon=\varepsilon_n$ such that $\varepsilon_n\to 0$ whereas $n\varepsilon_n^2\to\infty$ as $n\to\infty$. Meanwhile, to ensure that the narrow window defined by the event $A_{n,k}^X$ (where $X$ obeys the distribution $Q_{n,k}^{t_0}$) becomes negligible relative to the deviation $\varepsilon$, while simultaneously forcing the Chebyshev error to zero, we choose the scaling of the auxiliary window parameter $\delta$ such that 
    \begin{equation*}
        \delta=\delta_n = \varepsilon_n (n\varepsilon_n^2)^{-1/4}\,.
    \end{equation*}
    Because $n\varepsilon_n^2 \to \infty$, the above choice of $\delta$ guarantees that $n\delta_n^2 = (n\varepsilon_n^2)^{1/2} \to \infty\,,$ which ensures that the Chebyshev error term in \eqref{eq:probability_lower_bound} vanishes as $n\to\infty\,,$ i.e.,
    \begin{align*}
        &\frac{4}{n\delta_n^2} \frac{d^2}{dt^2}\Phi_n(t)\bigg|_{t=t_0} \to 0\,,\\
        \text{or, }\qquad& \Pr_{X\sim\mathcal{Q}^{t_0}_{n,k}}\left(A^X_{n,k}\right) \to 1\,.
    \end{align*}
    
    Therefore, taking $\log$ and multiplying both sides of Ineq.~\eqref{eq:probability_lower_bound} by $-1/n$, followed by taking the limit superior, we get
    \begin{equation}\label{eq:limsup_upper_bound}
        \limsup_{n\to\infty} \frac{-1}{n\varepsilon_n^2} \log \Pr\left(Z_{\widetilde{U}_{n,k}}>\varepsilon_n\right) \le \limsup_{n\to\infty} \frac{1}{\varepsilon_n^2} \left[ t_0(\varepsilon_n + \delta_n) - \Phi_n(t_0) \right]\,.
    \end{equation}
    
    Because $\varepsilon_n \to 0$ and $\delta_n \ll \varepsilon_n$, therefore the higher-order terms in the power series expansion of $t_0$ (i.e., Eq.~\eqref{eq:power_series_expansion}) vanish asymptotically. So, the tilting parameter simplifies to leading order as 
    $$t_0 = \frac{\varepsilon_n}{n\mathrm{Var}_{\rho^{\otimes n}}(U_{n,k})} + \mathcal{O}(\varepsilon_n^2)\,.$$ 
    Substituting the above value of $t_0$, alongside the Taylor series expansion of $\Phi_n(t_0)$ around 0 of Eq.~\eqref{eq:sum_exact_of_Phi} (i.e., $\Phi_n(t_0) = (t_0^2/2) n\mathrm{Var}_{\rho^{\otimes n}}(U_{n,k})  + \mathcal{O}\big(t_0^3\big)\,,$ since $t_0$ is arbitrarily small) into the exponent on the right-hand side of Ineq.~\eqref{eq:probability_lower_bound}, we get
    \begin{align*}
        t_0(\varepsilon_n + \delta_n) - \Phi_n(t_0) &= t_0\varepsilon_n - \frac{t_0^2}{2} n\mathrm{Var}_{\rho^{\otimes n}}(U_{n,k}) + t_0\delta_n - \mathcal{O}(t_0^3) \\
        &= \frac{\varepsilon_n^2}{n\mathrm{Var}_{\rho^{\otimes n}}(U_{n,k})} - \frac{\varepsilon_n^2}{2n\mathrm{Var}_{\rho^{\otimes n}}(U_{n,k})} + \mathcal{O}(\varepsilon_n\delta_n) + \mathcal{O}(\varepsilon_n^3) \\
        &= \frac{\varepsilon_n^2}{2n\mathrm{Var}_{\rho^{\otimes n}}(U_{n,k})} + o(\varepsilon_n^2)\,.
    \end{align*}
    
    The remainder terms $o(\varepsilon_n^2)$ vanish when divided by the rate scale $\varepsilon_n^2$. Finally, plugging the above into inequality~\eqref{eq:limsup_upper_bound}, and invoking the variance expansion identity of quantum U-statistic observables in terms of the leading-order intrinsic variance factor $\mathrm{Var}_{\rho}\left(\nabla f(\rho)\right)$ as given in Eq.~\eqref{eq:grad_marg_connection}, we get 
    \begin{align}
        \limsup_{n\to\infty} \frac{-1}{n\varepsilon_n^2} \log \Pr\left(Z_{\widetilde{U}_{n,k}}>\varepsilon_n\right) \le \frac{1}{2\mathrm{Var}_{\rho}\left(\nabla f(\rho)\right)}\,.
    \end{align}
    
    Again, noting the symmetricity of the local $k$-copy operators $\big(O^{\text{sym}}_k-f(\rho)\mathbb{I}^{\otimes k}\big)$ with respect to sign inversions (i.e., replacing the operator $\big(O^{\text{sym}}_k-f(\rho)\mathbb{I}^{\otimes k}\big)$ with $-\big(O^{\text{sym}}_k-f(\rho)\mathbb{I}^{\otimes k}\big)$ perfectly preserves the spectral bound ${\big\lVert{}O^{\text{sym}}_k-f(\rho)\mathbb{I}^{\otimes k}\big\rVert{}}_\infty\leq 2{\big\lVert{}O^{\text{sym}}_k\big\rVert{}}_\infty$, and identically maintains all underlying pair covariances), we have therefore proved the following two-sided concentration convergence theorem:

\begin{theorem}[Quantum U-Statistic Concentration Convergence Rate Asymptotic Upper Bound in the Large-$n$ Moderate Deviation Regime]\label{theo:mod_dev_conc_upper}
    Given an arbitrary state $\rho\in\mathcal{D}(\mathcal{H})$, let $f$ be a quantum polynomial functional with a non-degenerate estimation geometry (i.e., $\mathrm{Var}_{\rho}\left(\nabla f(\rho)\right) > 0$). Let $U_{n,k}$ be the corresponding unbiased quantum U-statistic defined over a permutation-invariant $k$-copy local kernel $O^{\text{sym}}_k$ for $f$. Let $\{\varepsilon_n\}_{n=1}^{\infty}$ be a sequence of positive deviations such that $\varepsilon_n \to 0$, whereas $n\varepsilon_n^2 \to \infty$ as $n \to \infty$. Then, the moderate-deviation quantum U-statistic fluctuation probability convergence rate in the large-$n$ limit is upper-bounded as 
    \begin{equation*}
        \limsup_{n\to\infty} -\frac{1}{n\varepsilon_n^2} \log \Pr\big(\big\lvert{}Z_{U_{n,k}}-f(\rho)\big\rvert{}>\varepsilon_n\big) \leq \frac{1}{2\mathrm{Var}_{\rho}\left(\nabla f(\rho)\right)}\,,
    \end{equation*}
    where $Z_{U_{n,k}}$ denotes the outcome of measuring the product state $\rho^{\otimes n}$ in the eigenbasis of the quantum U-statistic $U_{n,k}\,.$
\end{theorem}

Consequently, Theorems \ref{theo:mod_dev_conc} and \ref{theo:mod_dev_conc_upper} collectively demonstrate that the QFI exactly characterizes the concentration convergence rate of quantum U-statistic observables in the large-$n$ intermediate scaling regime.

\section{Conclusion and Acknowledgements}
In this manuscript, we establish a (first-order) marginal--gradient equivalence for permutation-invariant finite-copy kernels of polynomial (trace) functionals, and prove that the quantum U-statistic is the unique unbiased permutation-invariant extension to an arbitrary number of copies. Combining these results with the variance analysis of quantum U-statistics, and the multiparameter quantum Cram\'er--Rao framework, we show that their leading-order variance attains the quantum Cram\'er--Rao limit, establishing asymptotic efficiency without preliminary tomography or adaptive measurements. Furthermore, we derive a variance-sensitive matrix Bernstein-type concentration inequality to bound finite-sample estimation errors, which readily yields a finite-sample complexity lower bound on $n$. We then establish the Moderate Deviation Principle (MDP) to characterize purely Gaussian tail behaviour in the intermediate scaling regime, and also demonstrate that in this scaling regime, the probability concentration convergence rate is exactly characterised by the quantum Fisher information (QFI) in the large-$n$ limit, thereby imparting this fundamental quantum information-theoretic quantity a new operational meaning.

We proceed to apply these results to several other quantum information-theoretic quantities, including state purity, (squared) Hilbert--Schmidt distance, and the Bures $\chi^2$-divergence. For the latter, we derive an exact continuous-time integral representation and show that the spectral lower-bound condition imposed in \cite{BOW19} is sufficient but not necessary for bounded-variance estimation. In particular, bounded variance can persist even as the minimum eigenvalue of the reference state approaches zero.

AD acknowledges the support provided by the TCS Research Scholar Program (Cycle 19) from Tata Consultancy Services. NAW acknowledges the funding support from the National Quantum Mission, an initiative of the Department of Science and Technology, Govt. of India and the support provided by the Foundation for QC Innovation (FQCI), DST-NQM T-Hub at IISc Bengaluru, in facilitating this project. NAW also acknowledges the support provided by the Grant  ANRF/ARG/2025/012066/MS from the Department of Science \& Technology, Govt of India. 

\bibliography{master}
\bibliographystyle{IEEEtran}

\vspace{1cm}

\appendices

\section*{Organization of the Appendices}

The appendices provide supplementary mathematical proofs and technical derivations supporting the main results of this paper. The material is organized as follows:

\begin{itemize}
    \item \textbf{Appendix \ref{app:A}} details the structural correspondence between quantum marginal kernels and classical conditional expectations;
    
    \item \textbf{Appendix \ref{app:operator_span}} provides the algebraic proof of Lemma \ref{lemma:operator_span}, establishing the spanning property of permutation-invariant operators;
    
    \item \textbf{Appendix \ref{app:proof_grad_marg}} presents the explicit variance analysis of quantum U-statistics via the Hoeffding decomposition (thereby justifying Eq.~\eqref{eq:hoeffding_decomp_var}), and rigorously extends this analysis to higher-order intrinsic variance limits for degenerate polynomial functionals;
    
    \item \textbf{Appendix \ref{app:C}} contains the formal proofs regarding the continuous-time integral representation of the Bures $\chi^2$-divergence and the derivation of the inverse Symmetric Logarithmic Derivative (SLD) super-operator;

    \item \textbf{Appendix \ref{app:cumulant_bounds}} provides a detailed discussion on the cumulants of a centred quantum U-statistic, along with an upper bound similar to the spanning-tree upper bound on the connected absolute moments derived in Lemma \ref{lemma:spanning_tree_connected_moment}. We further exhibit the strict convexity of the normalized CGF of a quantum U-statistic, and present a few easy-to-observe connections between the first-order and the second-order (scaled) cumulants of a quantum U-statistic and the statistics of certain boosted (tilted) measures.
\end{itemize}

\renewcommand{\thesubsection}{\thesection-\Roman{subsection}}

\section{Correspondence Between Quantum Marginal Kernels and Classical Conditional Expectation}\label{app:A}

To understand the correspondence between the definition of marginal kernels and classical conditional expectation, consider a density operator $\rho$ having the spectral decomposition
\[
    \rho=\sum_{x\in S} p_x |x\rangle\langle x|,
\]
where $\{|x\rangle\}_{x\in S}$ denotes a basis of $\mathcal{H}^{\otimes m}$, and $\{p_x\}_{x\in S}$ denotes a measure. We can identify
the basis labels $x$ with the outcomes of a classical random
variable $X$ satisfying $\Pr(X=x)=p_x$. Now, if the observable $m$-copy permutation-invariant kernel
$O_m^{\mathrm{sym}}$ commutes with the product state $\rho^{\otimes m}$, then we know that they are simultaneously diagonalizable in the same basis. Let $O_m^{\mathrm{sym}}$ diagonalize in the product eigenbasis of $\rho^{\otimes m}$ as
\[
    O_m^{\mathrm{sym}} = \sum_{x_1,\ldots,x_m\in S} g(x_1,\ldots,x_m) |x_1,\ldots,x_m\rangle \langle x_1,\ldots,x_m|,
\]
where $g(x_1,\ldots,x_m)$ represents the diagonal entries of $O^{\text{sym}}_m$, i.e.,
\[g(x_1,\ldots,x_m):=\langle x_1,\ldots,x_m| O_m^{\mathrm{sym}} |x_1,\ldots,x_m\rangle\,,\forall x_i\in S, \forall i\in [m].\] Thus, in this basis, the quantum kernel is precisely the operator representation of some classical function $g:\mathbb{C}^m\to \mathbb{C}$. Moreover, the permutation-invariance of $O_m^{\mathrm{sym}}$ implies that $g$ is a symmetric
function, i.e., for every $\pi \in S_m$, we have
\begin{align*}
    g(x_{\pi(1)},\ldots,x_{\pi(m)})
    &= \langle x_{\pi(1)},\ldots,x_{\pi(m)}| O_m^{\mathrm{sym}} |x_{\pi(1)},\ldots,x_{\pi(m)}\rangle\\
    &= \langle x_1,\ldots,x_m| P^{\dagger}_{\pi}O_m^{\mathrm{sym}} p_{\pi} |x_1,\ldots,x_m\rangle\\
    &= \langle x_1,\ldots,x_m| O_m^{\mathrm{sym}} |x_1,\ldots,x_m\rangle\\
    &=g(x_1,\ldots,x_m).
\end{align*}

Let us now consider the $r^{\text{th}}$-order marginal kernel $O_{m,r}^{\mathrm{sym}}$, which is defined by tracing out the last $(m-r)$ subsystems. By evaluating the partial trace over the product eigenbasis, we obtain
\begin{align}
    O_{m,r}^{\mathrm{sym}} &= \mathrm{Tr}_{r+1 \dots m}\Big[O_m^{\mathrm{sym}} \big(\mathbb{I}^{\otimes r} \otimes \rho^{\otimes(m-r)}\big)\Big] \nonumber \\
    &\overset{\tiny (a)}{=} \mathrm{Tr}_{r+1 \dots m}\left[ \left( \sum_{x_1,\ldots,x_m \in S} g(x_1,\ldots,x_m) |x_1,\ldots,x_m\rangle \langle x_1,\ldots,x_m| \right) \left( \sum_{y_1,\ldots,y_m \in S} \left(\prod_{j=r+1}^{m}p_{y_j}\right) |y_1,\ldots,y_m\rangle \langle y_1,\ldots,y_m| \right) \right] \nonumber \\
    &{=} \mathrm{Tr}_{r+1 \dots m}\left[ \sum_{x_1,\dots,x_m \in S} g(x_1,\ldots,x_m) \left(\prod_{j=r+1}^{m}p_{x_j}\right) |x_1,\ldots,x_m\rangle \langle x_1,\ldots,x_m| \right] \nonumber \\
    &{=} \sum_{x_1,\dots,x_m \in S} g(x_1,\ldots,x_m) \left(\prod_{j=r+1}^{m}p_{x_j}\right) \mathrm{Tr}_{r+1 \dots m}\Big[ |x_1,\ldots,x_m\rangle \langle x_1,\ldots,x_m| \Big] \nonumber \\
    &\overset{\tiny (b)}{=} \sum_{x_1,\ldots,x_r\in S} \left[ \sum_{x_{r+1},\ldots,x_m\in S} g(x_1,\ldots,x_m) \prod_{j=r+1}^{m}p_{x_j} \right] |x_1,\ldots,x_r\rangle \langle x_1,\ldots,x_r|\,, \label{eq:marginal_classical_expansion}
\end{align}
where equality ($a$) follows from substituting the spectral expansion of the kernel $O_m^{\mathrm{sym}}$ and the fully expanded identity-padded state $\mathbb{I}^{\otimes r} \otimes \rho^{\otimes(m-r)}$ over the $m$-partite Hilbert space, and equality ($b$) follows by evaluating the partial trace over the subsystems $(r+1)$ through $m$ (collapsing the basis kets via $\mathrm{Tr}[|x\rangle\langle x|]=1$), and strategically grouping the sum to isolate the terms for the untraced first $r$ subsystems.

For independent random variables $X_1,\ldots,X_m \sim p$, the coefficient in square brackets on the right-hand side of equality ($b$) of Eq.~\eqref{eq:marginal_classical_expansion} is precisely the conditional expectation of $g$ given the first $r$ random variables, i.e.,
\begin{equation*}\label{eq:classical_conditional_expectation}
    \mathbb E\left[g(X_1,\ldots,X_m) \,\middle|\, X_1=x_1,\ldots,X_r=x_r\right] = \sum_{x_{r+1},\ldots,x_m\in S}
    g(x_1,\ldots,x_m) \prod_{j=r+1}^{m}p_{x_j}\,.
\end{equation*}
Now, defining the $r^{\text{th}}$-order marginal of the classical function $g$ as
\[g_r(x_1,\ldots,x_r):=\mathbb E\left[ g(X_1,\ldots,X_m) \,\middle|\, X_1=x_1,\ldots,X_r=x_r\right]\,,\]
Eq.~\eqref{eq:marginal_classical_expansion} becomes exactly equal to
\[O_{m,r}^{\mathrm{sym}}=\sum_{x_1,\ldots,x_r} g_r(x_1,\ldots,x_r)|x_1,\ldots,x_r\rangle \langle x_1,\ldots,x_r|\,.\]
This shows that the quantum marginal kernel is exactly the operator representation of the classical conditional expectation. Therefore, when the kernel $O^{\text{sym}}_m$ is diagonal in the product eigenbasis of $\rho^{\otimes m}$, the quantum operation of multiplying by $\rho^{\otimes(m-r)}$ and tracing out the corresponding $(m-r)$ subsystems reduces exactly to computing the expectation of the classical function $g$ conditioned on the first $r$ random variables. In this sense, the definition of the quantum marginal kernel is a natural generalization of classical conditional expectation.

\section{Proof of Lemma \ref{lemma:operator_span}}\label{app:operator_span}

Observe that the direction $\mathrm{span}_{\mathbb{C}}\{A^{\otimes n}\mid A\in M_d\}\subseteq \mathcal{K}(M_d)$ is obviously true, since any linear combination $\sum_{j=1}^m c_j X_j^{\otimes (n-1)}$ of matrices $X_j \in M_d$, and constants $c_j \in \mathbb{C}$, is a permutation-invariant matrix lying in the general tensor space $M_d^{\otimes n}$. Therefore, the only non-trivial direction to prove is that $\mathcal{K}(M_d)\subseteq \mathrm{span}_{\mathbb{C}}\{A^{\otimes n}\mid A\in M_d\}\,.$
    
    Let $\mathcal{W}_n = \mathrm{span}_{\mathbb{C}} \big\{ A^{\otimes n} \mid A \in M_d \big\}$. By definition, the subspace $\mathcal{K}(M_d)$ of $M_d^{\otimes n}$ is linearly spanned by permutation-invariant elementary tensors of the form
    \begin{equation*}
        S_n(A_1, \dots, A_n) = \sum_{\pi \in S_n} A_{\pi(1)} \otimes \dots \otimes A_{\pi(n)}
    \end{equation*}
    for arbitrary $A_1, \dots, A_n \in M_d$. We proceed by induction on $n$ to show that $S_n(A_1, \dots, A_n) \in \mathcal{W}_n$.
    
    \textbf{Base case ($n=2$):} For any $A, B \in M_d$, we observe the identity
    \begin{equation*}
        A \otimes B + B \otimes A = (A+B)^{\otimes 2} - A^{\otimes 2} - B^{\otimes 2}.
    \end{equation*}
    The right-hand side is clearly a linear combination of matrices of the form $X^{\otimes 2}$ for $X \in M_d$, establishing that $S_2(A,B) \in \mathcal{W}_2$.
    
    \textbf{Inductive step:} Assume the claim holds for $n-1$. Let $A_1, \dots, A_n \in M_d$. By the inductive hypothesis, the symmetrization over the first $n-1$ elements lies in $\mathcal{W}_{n-1}$, meaning that
    \begin{equation*}
        S_{n-1}(A_1, \dots, A_{n-1}) = \sum_{j=1}^m c_j X_j^{\otimes (n-1)},
    \end{equation*}
    for some constants $c_j \in \mathbb{C}$ and matrices $X_j \in M_d$. 
    
    To construct $S_n(A_1, \dots, A_n)$, we can express the symmetric group $S_n$ through its left coset decomposition with respect to $S_{n-1}$ (the subgroup of permutations fixing the $n^{\text{th}}$ index). Any permutation $\pi \in S_n$ can be uniquely decomposed as a permutation $\sigma \in S_{n-1}$ acting on the first $(n-1)$ elements, followed by a transposition that places the $n^{\text{th}}$ element into the $k^{\text{th}}$ tensor position, where $k \in \{1, \dots, n\}$. This decomposition allows us to factor the sum over $S_n$ into an outer sum over the insertion position $k$, and an inner sum over the permutations $\sigma \in S_{n-1}$ as
    \begin{equation*}
        S_n(A_1, \dots, A_n) = \sum_{k=1}^n \Bigg( \sum_{\sigma \in S_{n-1}} A_{\sigma(1)} \otimes \dots \otimes A_{\sigma(k-1)} \otimes A_n \otimes A_{\sigma(k)} \otimes \dots \otimes A_{\sigma(n-1)} \Bigg).
    \end{equation*}
    The inner summation over $\sigma$ is exactly the symmetrization operation $S_{n-1}$ acting on the subset $\{A_1, \dots, A_{n-1}\}$, but split around the $k^{\text{th}}$ tensor factor. Substituting the inductive hypothesis $S_{n-1}(A_1, \dots, A_{n-1}) = \sum_{j=1}^m c_j X_j^{\otimes (n-1)}$ directly into this structure linearly distributes $A_n$ into the $k^{\text{th}}$ position of each basis tensor $X_j$, yielding
    \begin{equation*}
        S_n(A_1, \dots, A_n) = \sum_{j=1}^m c_j \sum_{k=1}^n \left( X_j^{\otimes (k-1)} \otimes A_n \otimes X_j^{\otimes (n-k)} \right).
    \end{equation*}
    By linearity, it is sufficient to prove that for any $X, Y \in M_d$, the term 
    \begin{equation*}
        T(X,Y) = \left(\sum_{k=1}^n X^{\otimes (k-1)} \otimes Y \otimes X^{\otimes (n-k)}\right) \in \mathcal{W}_n\,.
    \end{equation*} 
    Consider the tensor polynomial in $t \in \mathbb{C}$ defined by
    \begin{equation*}
        (X + tY)^{\otimes n} = X^{\otimes n} + t T(X,Y) + \mathcal{O}(t^2).
    \end{equation*}
    Since $M_d$ is closed under linear combinations, we have $(X+tY) \in M_d$, which guarantees $(X+tY)^{\otimes n} \in \mathcal{W}_n$ for all $t$. Since $\mathcal{W}_n$ is a vector space, the coefficients of this polynomial (which can be extracted via interpolation at distinct values of $t$) must also reside in $\mathcal{W}_n$. Specifically, the linear coefficient $T(X,Y) \in \mathcal{W}_n$. Consequently, $S_n(A_1, \dots, A_n) \in \mathcal{W}_n$, completing the induction. \hfill\QED

\section{Variance Analysis of Quantum U-Statistics (Eq.~\eqref{eq:hoeffding_decomp_var}) and Higher-Order Extensions}\label{app:proof_grad_marg}

We know by Lemma \ref{lemma:u_stat_uniqueness} that any unbiased, permutation-invariant, global $n$-copy estimator $O_n^{\text{sym}}$ for a quantum polynomial functional $f(\rho)$ is uniquely determined to be a quantum U-statistic observable estimating $f(\rho)$. Therefore, applying quantum Hoeffding decomposition \cite{GB2010}, we can expand the variance of the observable $O_n^{\text{sym}}$ over the product state $\rho^{\otimes n}$ as a sum of the variances of its mutually orthogonal, centered Hoeffding projections, given by
\begin{equation*}
    \mathrm{Var}_{\rho^{\otimes n}}(O_n^{\text{sym}}) = \sum_{i=1}^k \binom{k}{i}^2\binom{n}{i}^{-1} \mathrm{Var}_{\rho}(H_{k,i}^{\text{sym}})\,,
\end{equation*}
where $H_{k,i}^{\text{sym}}$ denotes the $i^{\text{th}}$-order centered Hoeffding projection of a permutation-invariant, local $k$-copy kernel $O_k^{\text{sym}}$ of $f(\rho)$. For a specific subset of subsystems $S$ of size $i$, the centered projection $H_{k,i}^{\text{sym}}$ is explicitly defined in terms of the lower-order marginal kernels via the alternating sum
\begin{equation*}
    H_{k,i}^{\text{sym}} = \sum_{A \subseteq S} (-1)^{|S|-|A|} O_{k,|A|}^{\text{sym}}\,,
\end{equation*}
where $O_{k,|A|}^{\text{sym}}$ represents the $|A|^{\text{th}}$-order marginal kernel evaluated on the subset $A$.
 
To analyze the asymptotic behavior of the variance in the large-$n$ limit, we isolate the first-order term ($i=1$) from the remaining higher-order terms ($i \ge 2$) of the summation. Noting that $\binom{n}{1}^{-1} = \frac{1}{n}$, and that the combinatorial weights force all higher-order variances ($i \ge 2$) to scale as $\mathcal{O}(1/n^2)$, we obtain
\begin{align*} 
   \mathrm{Var}_{\rho^{\otimes n}}(O_n^{\text{sym}}) &= \frac{k^2}{n} \mathrm{Var}_{\rho}(H_{k,1}^{\text{sym}}) + \sum_{i=2}^k \binom{k}{i}^2\binom{n}{i}^{-1} \mathrm{Var}_{\rho}(H_{k,i}^{\text{sym}}) \\ 
   &= \frac{1}{n} \mathrm{Var}_{\rho}(k H_{k,1}^{\text{sym}}) + \mathcal{O}\left(\frac{1}{n^2}\right)\,. 
\end{align*}
 
By the definition of the Hoeffding decomposition, the first-order centered projection is exactly the first-order marginal kernel shifted by the functional's expected value, $H_{k,1}^{\text{sym}} = O_{k,1}^{\text{sym}} - f(\rho)\mathbb{I}$. Since the quantum variance of any observable is invariant under translations by a scalar multiple of the identity, we have $\mathrm{Var}_{\rho}(H_{k,1}^{\text{sym}}) = \mathrm{Var}_{\rho}(O_{k,1}^{\text{sym}})$. 
 
We can further shift this marginal kernel by an additional state-dependent scalar $C(\rho) = r c_0 + \sum_{j=1}^r (r-j) P_j(\rho)$ without altering its variance to get
\begin{equation*}
    \mathrm{Var}_{\rho^{\otimes n}}(O_n^{\text{sym}}) = \frac{1}{n} \mathrm{Var}_{\rho}\big(k O_{k,1}^{\text{sym}} - C(\rho)\mathbb{I}\big) + \mathcal{O}\left(\frac{1}{n^2}\right)\,.
\end{equation*} 
 
By Lemma \ref{lemma:generic_poly_algebraic}, this shifted first-order marginal kernel exactly recovers the gradient of the functional such that $\nabla f(\rho) = k O_{k,1}^{\text{sym}} - C(\rho)\mathbb{I}$. Substituting this identity into our asymptotic expansion directly yields the desired result
\begin{equation*}
    \mathrm{Var}_{\rho^{\otimes n}}(O_n^{\text{sym}}) = \frac{1}{n} \mathrm{Var}_{\rho}\big(\nabla f(\rho)\big) + \mathcal{O}\left(\frac{1}{n^2}\right)\,.
\end{equation*} 

It is important to note that when the functional gradient evaluates to a scalar multiple of the identity, the leading $\mathcal{O}(1/n)$ variance term vanishes. In such degenerate cases, the asymptotic variance limit shifts to a second-order decay governed by the functional's Hessian operator. For the sake of completeness, we formally state and prove this claim as a corollary.

\begin{corollary}[Second-Order Intrinsic Quantum Variance Limit for Degenerate Polynomial Functionals]\label{cor:degenerate_variance}
    Let $f(\rho)$ be a polynomial functional evaluated using an unbiased global permutation-invariant estimator $O_n^{\text{sym}}$. If the functional's gradient is a scalar multiple of identity at the true state, i.e., $\nabla f(\rho) = c_1\mathbb{I}$ (rendering the first-order estimation degenerate), then provided that the functional's Hessian operator $\nabla^2 f(\rho)$ is not a scalar multiple of identity (i.e., $\nabla^2 f(\rho)\neq c_2\mathbb{I}^{\otimes 2}$), the leading term of the intrinsic quantum variance scales as $\mathcal{O}(1/n^2)$ and is governed by $\nabla^2 f(\rho)$. That is,
    \begin{equation*}
        \mathrm{Var}_{\rho^{\otimes n}}(O_n^{\text{sym}}) = \frac{1}{2n^2} \mathrm{Var}_{\rho^{\otimes 2}}\big(\nabla^2 f(\rho)\big) + \mathcal{O}\left(\frac{1}{n^3}\right)\,.
    \end{equation*}
\end{corollary}

\begin{proof}
    By the quantum Hoeffding decomposition \cite{GB2010}, the total variance of the global estimator expands as a sum of mutually orthogonal, centered Hoeffding projections, given by
    \begin{equation*}
        \mathrm{Var}_{\rho^{\otimes n}}(O_n^{\text{sym}}) = \sum_{i=1}^k \binom{k}{i}^2\binom{n}{i}^{-1} \mathrm{Var}_{\rho^{\otimes i}}(H_{k,i}^{\text{sym}})\,.
    \end{equation*}
    
    By Lemma \ref{lemma:generic_poly_algebraic}, the first-order marginal kernel $O_{k,1}^{\text{sym}}$ is equivalent (upto an additive factor of some scalar multiple of the Identity operator) to the functional gradient. If $\nabla f(\rho) = c\mathbb{I}$, then $O_{k,1}^{\text{sym}}$ is also a scalar multiple of the identity matrix. By definition, the first-order centered Hoeffding projection is $H_{k,1}^{\text{sym}} = O_{k,1}^{\text{sym}} - f(\rho)\mathbb{I}$. Since shifting a scalar identity matrix by another scalar identity matrix yields a scalar identity matrix, the variance of $H_{k,1}^{\text{sym}}$ evaluates exactly to zero.
    
    Consequently, the $i=1$ term in the Hoeffding sum vanishes. The dominant leading-order variance is therefore forced to the $i=2$ term. Expanding the combinatorial weights for $i=2$ yields
    \begin{align*}
        \mathrm{Var}_{\rho^{\otimes n}}(O_n^{\text{sym}}) &= \binom{k}{2}^2\binom{n}{2}^{-1} \mathrm{Var}_{\rho^{\otimes 2}}(H_{k,2}^{\text{sym}}) + \mathcal{O}\left(\frac{1}{n^3}\right) \\
        &= \frac{k^2(k-1)^2}{4} \left( \frac{2}{n(n-1)} \right) \mathrm{Var}_{\rho^{\otimes 2}}(H_{k,2}^{\text{sym}}) + \mathcal{O}\left(\frac{1}{n^3}\right) \\
        &= \frac{k^2(k-1)^2}{2n(n-1)} \mathrm{Var}_{\rho^{\otimes 2}}(H_{k,2}^{\text{sym}}) + \mathcal{O}\left(\frac{1}{n^3}\right) \\
        &= \frac{k^2(k-1)^2}{2n^2} \mathrm{Var}_{\rho^{\otimes 2}}(H_{k,2}^{\text{sym}}) + \mathcal{O}\left(\frac{1}{n^3}\right)\,,
    \end{align*}
    where the final equality follows from the asymptotic expansion $\frac{1}{n(n-1)} = \frac{1}{n^2}(1 - 1/n)^{-1} = \frac{1}{n^2} + \mathcal{O}(1/n^3)\,.$ To connect the remaining variance to the functional's geometry, we evaluate the second-order directional derivative of $f(\rho)$ along the path $\rho_t = \rho + tX$. By standard Taylor expansion, the second-order derivative defines the Hessian operator $\nabla^2 f(\rho) \in \mathcal{B}(\mathcal{H}^{\otimes 2})$ acting on two copies of the perturbation as
    \begin{equation*}
        \frac{d^2}{dt^2} f(\rho_t) \bigg|_{t=0} = \mathrm{Tr}\big[ \nabla^2 f(\rho) X^{\otimes 2} \big]\,.
    \end{equation*}
    
    We simultaneously evaluate this second-order derivative structurally using the $k$-copy permutation-invariant kernel expectation, i.e., $f(\rho_t) = \mathrm{Tr}\big[O_k^{\text{sym}} (\rho + tX)^{\otimes k}\big]$. Taking the second-order derivative with respect to $t$ isolates the terms in the tensor expansion that contain exactly two copies of $X$ and $k-2$ copies of $\rho$. There are $k(k-1)$ such ordered pairs in the expansion. Since $O_k^{\text{sym}}$ is invariant under any permutation of its subsystems, every pair yields an identical trace evaluation. Combining these terms yields
    \begin{align*}
        \frac{d^2}{dt^2} \mathrm{Tr}\big[O_k^{\text{sym}} (\rho + tX)^{\otimes k}\big] \bigg|_{t=0} &= k(k-1) \mathrm{Tr}\big[ O_k^{\text{sym}} (X \otimes X \otimes \rho^{\otimes (k-2)}) \big] \\
        &= k(k-1) \mathrm{Tr}\big[ \big( \mathrm{Tr}_{3 \dots k}[O_k^{\text{sym}} (\mathbb{I}^{\otimes 2} \otimes \rho^{\otimes (k-2)})] \big) X^{\otimes 2} \big] \\
        &= \mathrm{Tr}\big[ k(k-1) O_{k,2}^{\text{sym}} X^{\otimes 2} \big]\,,
    \end{align*}
    where $O_{k,2}^{\text{sym}}$ is exactly the second-order marginal kernel. Equating the algebraic and structural evaluations for arbitrary traceless $X$ dictates that the Hessian operator is proportional to the second-order marginal kernel as
    \begin{equation*}
        \nabla^2 f(\rho) = k(k-1) O_{k,2}^{\text{sym}} - C_2(\rho)\mathbb{I}^{\otimes 2}\,,
    \end{equation*}
    where $C_2(\rho)$ is a scalar shift resulting from identity-padding.
    
    The second-order centered Hoeffding projection is explicitly defined as $H_{k,2}^{\text{sym}} = O_{k,2}^{\text{sym}} - O_{k,1}^{\text{sym}} \otimes \mathbb{I} - \mathbb{I} \otimes O_{k,1}^{\text{sym}} + f(\rho)\mathbb{I}^{\otimes 2}$. Since $O_{k,1}^{\text{sym}} = c\mathbb{I}$ by our degeneracy assumption, all lower-order subtraction terms are scalar multiples of the identity matrix $\mathbb{I}^{\otimes 2}$. 
    
    Since the quantum variance is perfectly invariant under scalar identity shifts, the variance of the centered projection equals the variance of the uncentered marginal, i.e., $\mathrm{Var}_{\rho^{\otimes 2}}(H_{k,2}^{\text{sym}}) = \mathrm{Var}_{\rho^{\otimes 2}}(O_{k,2}^{\text{sym}})$. 
    
    Substituting the relation $\nabla^2 f(\rho) = k(k-1) O_{k,2}^{\text{sym}} - C_2(\rho)\mathbb{I}^{\otimes 2}$ into the isolated $i=2$ variance term yields
    \begin{align*}
        \mathrm{Var}_{\rho^{\otimes n}}(O_n^{\text{sym}}) &= \frac{k^2(k-1)^2}{2n^2} \mathrm{Var}_{\rho^{\otimes 2}}(O_{k,2}^{\text{sym}}) + \mathcal{O}\left(\frac{1}{n^3}\right) \\
        &= \frac{1}{2n^2} \mathrm{Var}_{\rho^{\otimes 2}}\big( k(k-1) O_{k,2}^{\text{sym}} \big) + \mathcal{O}\left(\frac{1}{n^3}\right) \\
        &= \frac{1}{2n^2} \mathrm{Var}_{\rho^{\otimes 2}}\big( \nabla^2 f(\rho) \big) + \mathcal{O}\left(\frac{1}{n^3}\right)\,.
    \end{align*}
\end{proof}

More generally, this geometric variance hierarchy extends to any arbitrary order $r \le k$, where $k$ is the degree of the input polynomial to $f(\rho)$. If the estimation task is completely degenerate up to order $(r-1)$, meaning the first $(r-1)$ directional derivatives of the functional behave as scalar identity shifts at the true state, then the leading non-zero term in the Hoeffding decomposition shifts exactly to the $r^{\text{th}}$-order component. In such highly degenerate regimes, the asymptotic variance exhibits a suppressed polynomial decay scaling as $\mathcal{O}(1/n^r)$, and its exact leading coefficient is governed by the intrinsic variance of the $r^{\text{th}}$-order tensor gradient $\nabla^r f(\rho)$.

\section{Proofs Regarding the Bures $\chi^2$-Divergence and Inverse SLD}\label{app:C}

\subsection{Proof of Lemma \ref{lem:integral_measured_chi2}}\label{subsec:integral_measured_chi2}

To find the specific operator $Y^*$ that satisfies the symmetric logarithmic derivative equation $\rho = \frac{1}{2}(Y^*\sigma + \sigma Y^*)$, we map our quantum variables to the continuous-time Lyapunov form mentioned in Proposition \ref{prop:lyapunov} by setting
    \begin{align*}
        Q &= 2\rho, \\
        A &= -\sigma, \\
        A^\dagger &= -\sigma^\dagger = -\sigma, \\
        X &= Y^*.
    \end{align*}
    Substituting these into the standard Lyapunov equation yields
    \begin{equation}
        -\sigma Y^* - Y^*\sigma + 2\rho = 0 \iff \rho = \frac{1}{2}(Y^*\sigma + \sigma Y^*).\label{eq:Y_star}
    \end{equation}
    Applying the integral solution to these specific matrices provides the explicit continuous-time form for $Y^*$ as
    \begin{equation}
        Y^* = \int_0^\infty e^{-\tau\sigma} (2\rho) e^{-\tau\sigma} d\tau = 2 \int_0^\infty e^{-\tau\sigma} \rho e^{-\tau\sigma} d\tau.\nonumber
    \end{equation}
    Using Definition \ref{def:bures}, the Bures $\chi^2$-divergence evaluates to
    \begin{equation}
      \chi^2_{\mathrm{B}}(\rho \| \sigma) = \mathrm{Tr}[\rho Y^*] - 1,\nonumber
    \end{equation}
    where $Y^\star$ satisfies the condition mentioned in Eq.~\eqref{eq:Y_star}. Substituting the integral representation of $Y^*$ into the above trace formula gives
    \begin{align}
       \chi^2_{\mathrm{B}}(\rho \| \sigma) &= \mathrm{Tr}\left[\rho \left( 2 \int_0^\infty e^{-\tau\sigma} \rho e^{-\tau\sigma} d\tau \right)\right] - 1 \nonumber \\
        &= 2 \int_0^\infty \mathrm{Tr}[\rho e^{-\tau\sigma} \rho e^{-\tau\sigma}] d\tau - 1 \nonumber \\
        &= 2 \int_0^\infty \mathrm{Tr}[(\rho e^{-\tau\sigma})^2] d\tau - 1\,.\nonumber
    \end{align}\hfill\QED

\subsection{Proof of Corollary \ref{cor:spectral_omega}}\label{subsec:spectral_omega}

    From the continuous-time derivation in Lemma \ref{lem:integral_measured_chi2}, the super-operator applied to $\rho$ is given by the integral 
    \begin{equation*}
        \Omega_{\sigma}(\rho) = 2 \int_0^\infty e^{-\tau\sigma} \rho e^{-\tau\sigma} d\tau\,. 
    \end{equation*}
    
    Now, using the spectral decomposition of $\sigma$, the matrix exponential can be written as $e^{-\tau\sigma} = \sum_{i=1}^d e^{-\tau\lambda_i} |i\rangle \langle i|$. Substituting this into the integral yields
    \begin{align*}
        \Omega_{\sigma}(\rho) &= 2 \int_0^\infty \left( \sum_{j=1}^d e^{-\tau\lambda_j} |j\rangle \langle j| \right) \rho \left( \sum_{k=1}^d e^{-\tau\lambda_k} |k\rangle \langle k| \right) d\tau \\
        &= 2 \int_0^\infty \sum_{j,k=1}^d e^{-\tau(\lambda_j + \lambda_k)} |j\rangle \langle j| \rho |k\rangle \langle k| d\tau\,.
    \end{align*}
    Exploiting the linearity of the integral, we interchange the summation and the integration to obtain
    \begin{equation*}
        \Omega_{\sigma}(\rho) = \sum_{j,k=1}^d 2 \left( \int_0^\infty e^{-\tau(\lambda_j + \lambda_k)} d\tau \right) |j\rangle \langle j| \rho |k\rangle \langle k|\,.
    \end{equation*}
    Since $\sigma$ is full-rank, its eigenvalues are positive, ensuring $\lambda_j + \lambda_k > 0$. We evaluate the definite integral exactly as
    \begin{equation*}
        \int_0^\infty e^{-\tau(\lambda_j + \lambda_k)} d\tau = \left[ \frac{e^{-\tau(\lambda_j + \lambda_k)}}{-(\lambda_j + \lambda_k)} \right]_0^\infty = \frac{1}{\lambda_j + \lambda_k}\,.
    \end{equation*}
    Substituting this evaluated integral back into the summation yields the explicit spectral form as
    \begin{equation*}
        \Omega_{\sigma}(\rho) = \sum_{j,k=1}^d \frac{2}{\lambda_j + \lambda_k} |j\rangle \langle j| \rho |k\rangle \langle k|\,.
    \end{equation*}\hfill\QED

\section{Quantum Cumulants and (Normalized) CGF}
\label{app:cumulant_bounds}

In this appendix, we establish the finite-sample upper bounds on the cumulants of the \textit{centred} quantum U-statistic. We begin by defining the necessary notation, followed by an algebraic reduction of the multilinear cumulant expansion, and culminate in a purely discrete, tree-graph combinatorial bound that cleanly bypasses the diverging factorials typical of raw M\"obius inversion on partition lattices. 

Following this, we show that the second-order derivative of the normalized CGF of a centred quantum U-statistic (as defined in Definition~\eqref{eq:normalized_CGF}) is strictly positive, which in turn is a sufficient condition for its strict convexity. We also show that the first-order and second-order (scaled) cumulants of the centred quantum U-statistic exactly characterize the mean and variance, respectively, of certain boosted (tilted) measures defined with respect to the measure induced by measuring the product state $\rho^{\otimes n}$ in the eigenbasis of the centred quantum U-statistic.

\subsection{Notation and Setup}

Let $\mathcal{H}$ be a finite-dimensional Hilbert space and $\rho\in\mathcal{D}(\mathcal{H})$ be an arbitrary quantum state. Let $O_k^{\mathrm{sym}}\in\mathcal{B}(\mathcal{H}^{\otimes k})$ be a permutation-invariant observable with a bounded spectral radius, for some $k\in\mathbb{Z}_{>0}$. For any size-$k$ subset of indices $\Lambda \in \binom{[n]}{k}$, let $\big(O_k^{\mathrm{sym}} - f(\rho)\mathbb{I}^{\otimes k}\big)_{\Lambda}$ denote the operator acting trivially (i.e., the identity operator) on all tensor factors outside of $\Lambda$, and acting as the operator $\big(O_k^{\mathrm{sym}} - f(\rho)\mathbb{I}^{\otimes k}\big)$ on the subsystems in $\Lambda$. We define the centered local operator normalized by the combinatorial averaging factor as
\begin{equation}\label{eq:centred_local_operator}
    h_\Lambda := \frac{1}{\binom{n}{k}} \big( O_k^{\mathrm{sym}} - f(\rho)\mathbb{I}^{\otimes k} \big)_{\Lambda}\,.
\end{equation}
Then, the centered quantum U-statistic is simply the unweighted sum of these local terms over all possible $k$-subsets, i.e.,
\begin{align}
    \widetilde{U}_{n,k} &\quad=\quad \sum_{\Lambda \in \binom{[n]}{k}} h_\Lambda\notag\\
    &\quad=\quad \frac{1}{\binom{n}{k}} \sum_{\Lambda \in \binom{[n]}{k}} \big(O_k^{\mathrm{sym}} - f(\rho)\mathbb{I}^{\otimes k} \big)_{\Lambda}\,.\label{eq:quantum_U_stat_wrt_centred_local_operators}
\end{align}

For any family of self-adjoint operators with bounded spectral radius $\{A_1, \dots, A_m\}\subset \mathcal{B}(\mathcal{H}^{\otimes n})$, let their joint CGF under the product state $\rho^{\otimes n}$ be defined as
\begin{equation*}
    K_{A_1,\ldots,A_m}(t_1,\ldots,t_m) := \log \operatorname{Tr} \left[ \rho^{\otimes n} \exp\left( \sum_{j=1}^m t_j A_j \right) \right]\,.
\end{equation*}
Then, the joint cumulant is defined as an $r^{\text{th}}$-order partial derivative of the joint CGF evaluated at the origin, given by
\begin{equation*}
    \kappa(A_1,\ldots,A_r) := \left. \frac{\partial^r}{\partial t_1\cdots\partial t_r} K_{A_1,\ldots,A_r}(t_1,\ldots,t_r) \right|_{t_i=0, \forall i\in[r]}\,.
\end{equation*}
In particular, the $r^{\text{th}}$-order cumulant of $\widetilde{U}_{n,k}$ is the joint cumulant of $r$ identical copies of $\widetilde{U}_{n,k}$, i.e., $\kappa_r(\widetilde{U}_{n,k}) := \kappa(\widetilde{U}_{n,k},\ldots,\widetilde{U}_{n,k})\,.$ Similarly, the normalized CGF of $\widetilde{U}_{n,k}$ as given by Definition~\eqref{eq:normalized_CGF} is nothing but the univariate, unioperator joint CGF (under the product state $\rho^{\otimes n}$) of $\widetilde{U}_{n,k}$ scaled by a factor of $1/n$, given by 
$$\frac{1}{n}K_{\underbrace{\widetilde{U}_{n,k},\ldots,\widetilde{U}_{n,k}}_{n-\text{many}}}(\overbrace{t,\ldots,t}^{n-\text{many}})=\frac{1}{n}\log\mathrm{Tr}\left[\rho^{\otimes n}e^{nt\widetilde{U}_{n,k}}\right]\,.$$

\subsection{Univariate, Unioperator CGFs and Cumulants}

The following lemma connects the derivatives of univariate, unioperator CGFs to joint cumulants.

\begin{lemma}[Univariate, Unioperator CGF Identity]
\label{lem:univariate_cumulant_identity}
    Let $\rho^{\otimes n}\in\mathcal{D}(\mathcal{H}^{\otimes n})$ be an arbitrary product state, and $A\in\mathcal{B}(\mathcal{H}^{\otimes n})$ be a self-adjoint operator with bounded spectral radius. Let $K_A(s) := \log \mathrm{Tr}[\rho^{\otimes n} e^{sA}]$ denote the univariate, unioperator CGF of the operator $A$ under the product state $\rho^{\otimes n}$. Then, the $r^{\text{th}}$-order derivative of $K_A(s)$ evaluated at the origin $s=0$ exactly yields the $r^{\text{th}}$-order cumulant of $A$, i.e.,
    \begin{equation*}
        \frac{d^r}{ds^r} K_A(s) \bigg|_{s=0} = \kappa_r(A)\,,
    \end{equation*}
    where $\kappa_r(A) := \kappa(A, \dots, A)$ is the  joint cumulant of $r$ identical copies of $A$.
\end{lemma}

\begin{proof}
    Let us define the multivariate, unioperator CGF of the operator $A$ under the product state $\rho^{\otimes n}$ as \begin{equation*}
        K_A(t_1, \dots, t_r) := \log \mathrm{Tr}\left[\rho^{\otimes n} \exp\left( \sum_{i=1}^r t_i A \right) \right]\,.
    \end{equation*}
    
    Then, by definition, the joint cumulant of $r$ identical copies of the operator $A$ is generated by evaluating an $r^{\text{th}}$-order partial derivative of the multivariate, unioperator CGF of $A$ (under the product state $\rho^{\otimes n}$) at the origin, given by
    \begin{equation}
        \kappa_r(A) = \frac{\partial^r}{\partial t_1 \cdots \partial t_r} K_A(t_1, \dots, t_r) \bigg|_{t_i=0, \forall i\in[r]}\,.
        \label{eq:joint_cumulant_def_repeated}
    \end{equation}
    
    Observe that the operator argument in the exponential can be factored as $\sum_{i=1}^r t_i A = (t_1 + \dots + t_r)A$. Let us define a sum map acting on the parameters as $u(t_1, \dots, t_r) := t_1 + \dots + t_r$. We can then express the multivariate, unioperator CGF of $A$ as a composition of the univariate, unioperator CGF $K_A$ of $A$ and the sum map $u$ as
    \begin{equation*}
        K_A(t_1, \dots, t_r) = K_A\big( u(t_1, \dots, t_r) \big) = K_A(t_1 + \dots + t_r)\,.
    \end{equation*}
    
    We differentiate this composed function using the chain rule. Taking the partial derivative with respect to $t_1$ yields
    \begin{align*}
        \frac{\partial}{\partial t_1} K_A(t_1, \dots, t_r) &= \frac{d}{ds} K_A(s) \cdot \frac{\partial\big(u(t_1, \dots, t_r)\big)}{\partial t_1} \qquad \big(\text{where } s=u(t_1, \dots, t_r)\big)\\
        &= \frac{d}{ds} K_A(s)\,.
    \end{align*}
    Since the derivative of the sum map $u(t_1, \dots, t_r)$ with respect to any individual parameter $t_i$ is exactly $1$, each subsequent partial derivative simply advances the derivative of the univariate, unioperator CGF $K_A$. Applying the partial derivatives sequentially for $t_2, t_3, \dots, t_r$ yields
    \begin{equation*}
        \frac{\partial^r}{\partial t_1 \cdots \partial t_r} K_A(t_1, \dots, t_r) = \frac{d^r}{ds^r} K_A(s)\,.
    \end{equation*}
    
    Finally, evaluating this expression at the origin $t_i=0, \forall i\in[r]$, yields
    \begin{equation*}
        \frac{\partial^r}{\partial t_1 \cdots \partial t_r} K_A(t_1, \dots, t_r) \bigg|_{t_i=0, \forall i\in[r]} = \frac{d^r}{ds^r} K_A(s) \bigg|_{s=0}\,.
    \end{equation*}
    Equating this result with Eq.~\eqref{eq:joint_cumulant_def_repeated} completes the proof.
\end{proof}

The above lemma naturally yields the following corollary connecting the derivatives of the normalized CGF (of a centred quantum U-statistic) evaluated at the origin to its (scaled) cumulants.

\begin{corollary}\label{col:normalized_cgf_derivative}
    Let $\rho^{\otimes n}\in\mathcal{D}(\mathcal{H}^{\otimes n})$ be an arbitrary product state. Then, for any integer $r \in\mathbb{Z}_{\ge 1}$, the $r^{\text{th}}$-order derivative of the normalized CGF $\Phi_n(t)$ (as given by Definition~\eqref{eq:normalized_CGF}) evaluated at $t=0$ scales with the $r^{\text{th}}$-order cumulant of $\widetilde{U}_{n,k}$ as
    \begin{equation}\label{eq:rth_derivative_normalized_CGF}
        \frac{d^r}{dt^r} \Phi_n(t)\bigg|_{t=0} = n^{r-1}\kappa_r\big(\widetilde{U}_{n,k}\big)\,.
    \end{equation}
\end{corollary}

\begin{proof}
    Let $K_{\widetilde{U}_{n,k}}(s):= \log \mathrm{Tr}\Big[\rho^{\otimes n} e^{s \widetilde{U}_{n,k}}\Big]$ denote the univariate, unioperator CGF of the centered quantum U-statistic $\widetilde{U}_{n,k}$ under the product state $\rho^{\otimes n}$. By Lemma \ref{lem:univariate_cumulant_identity}, its $r^{\text{th}}$-order derivative evaluated at the origin yields the $r^{\text{th}}$-order cumulant of the centered quantum U-statistic, i.e., $$\frac{d^r}{du^r}K_{\widetilde{U}_{n,k}}(u)\bigg|_{u=0} = \kappa_r(\widetilde{U}_{n,k})\,.$$ Since, by definition, $\Phi_n(t) = \frac{1}{n} K_{\widetilde{U}_{n,k}}(nt)$, applying the chain rule $r$ times with respect to $t$ isolates the scaling factor as
    \begin{align*}
        \frac{d^r}{dt^r} \Phi_n(t)\bigg|_{t=0} &= \frac{1}{n} \cdot n^r\cdot \frac{d^r}{dv^r} K_{\widetilde{U}_{n,k}}(v)\bigg|_{v=0}\qquad \big(\text{where }v=nt\big)\\
        &= n^{r-1}\kappa_r\big(\widetilde{U}_{n,k}\big)\,.
    \end{align*}
\end{proof}

\subsection{Spanning-Tree Upper Bound for Cumulants of Centred Quantum U-Statistic}

Nevertheless, the exact combinatorial form of the $r^{\text{th}}$-order cumulant, as presented by Lemma \ref{lem:scaled_cgf_derivative}, is rather too complicated, and Theorem \ref{theo:mod_dev_conc_upper} can be proved using only a suitable upper bound on the absolute value of the cumulant. The following lemma helps achieve this by defining joint cumulants over connected sums of length $r$.

\begin{lemma}[Connected Multilinear Expansion of the Cumulant]
\label{lem:connected_cumulant_expansion}
    For every $r\in\mathbb{Z}_{\geq 2}$, the $r^{\text{th}}$-order cumulant of the centred quantum U-statistic $\widetilde{U}_{n,k}$ (under the product state $\rho^{\otimes n}$) admits the connected multilinear expansion given by
    \begin{equation*}
        \kappa_r(\widetilde{U}_{n,k}) = \sum_{\substack{\Lambda_1,\ldots,\Lambda_r\in\binom{[n]}{k}:\\ G(\boldsymbol{\Lambda}) \text { is connected}}} \kappa \left( h_{\Lambda_1},\ldots,h_{\Lambda_r} \right)\,,
    \end{equation*}
    where $G(\boldsymbol{\Lambda})$ is the intersection graph of the subset-tuple $\boldsymbol{\Lambda} = (\Lambda_1,\dots,\Lambda_r)$, i.e., the edge $\{i,j\}\in E\big(G(\mathbf{\Lambda})\big)$ if and only if $\Lambda_i \cap \Lambda_j \neq \emptyset,\ \forall i,j\in\binom{[r]}{2}$.
\end{lemma}

\begin{proof}
    We divide the proof into two steps: establishing the unconstrained multilinear sum, and proving the vanishing of disconnected spatial contributions.
    
    \begin{enumerate}
        \item \textit{Multilinear Expansion via the Chain Rule:} Let $\mathcal{S} = \binom{[n]}{k}$ denote the collection of all index subsets of size-$k$. Let us introduce a multivariate, multioperator CGF over the independent parameters $\{s_\Lambda\}_{\Lambda \in \mathcal{S}}\,,$ and the local operators $\{h_{\Lambda}\}_{\Lambda \in \mathcal{S}}$ given by Definition~\eqref{eq:centred_local_operator}, as 
        \begin{equation*}
            K_{\{h_\Lambda\}_{\Lambda \in \mathcal{S}}}(\{s_\Lambda\}_{\Lambda \in \mathcal{S}}) := \log \mathrm{Tr}\left[ \rho^{\otimes n} \exp\left( \sum_{\Lambda \in \mathcal{S}} s_\Lambda h_\Lambda \right) \right]\,.
        \end{equation*}
        Since $\widetilde{U}_{n,k} = \sum_{\Lambda \in \mathcal{S}} h_\Lambda$, the univariate, unioperator CGF $K_{\widetilde{U}_{n,k}}(s)$ is the restriction of the multivariate, multioperator CGF $K_{\{h_\Lambda\}_{\Lambda \in \mathcal{S}}}$ to the uniform diagonal mapping $s_\Lambda(s) = s,\forall \Lambda \in \mathcal{S}$, for some formal real parameter $s$. Thus, we can compose the multivariate, multioperator CGF and the diagonal map to get $$K_{\widetilde{U}_{n,k}}(s) = K_{\{h_\Lambda\}_{\Lambda \in \mathcal{S}}}\big( \{ s_\Lambda(s) \}_{\Lambda \in \mathcal{S}} \big)\,.$$
        
        We differentiate this composed function using the multivariable chain rule. The first-order derivative of the univariate, unioperator CGF $K_{\widetilde{U}_{n,k}}(s)$ with respect to $s$ evaluates to
        \begin{align*}
            \frac{d}{ds} K_{\widetilde{U}_{n,k}}(s) &= \sum_{\Lambda \in \mathcal{S}} \frac{\partial }{\partial s'_\Lambda} K_{\{h_\Lambda\}_{\Lambda \in \mathcal{S}}}\big( \{ s'_\Lambda \}_{\Lambda \in \mathcal{S}} \big) \cdot \frac{d s_\Lambda(s)}{ds}\qquad\big(\text{where }s'_\Lambda=s_\Lambda(s)\big)\\
            &= \sum_{\Lambda \in \mathcal{S}} \frac{\partial}{\partial s'_\Lambda}K_{\{h_\Lambda\}_{\Lambda \in \mathcal{S}}}\big( \{ s'_\Lambda \}_{\Lambda \in \mathcal{S}} \big)\,.
        \end{align*}
        Iterating this multivariable chain rule $r$ times yields an unconstrained sum over all $r$-tuples of $k$-subsets as
        \begin{equation*}
            \frac{d^r}{ds^r} K_{\widetilde{U}_{n,k}}(s) = \sum_{\Lambda_1, \ldots, \Lambda_r \in \mathcal{S}} \frac{\partial^r}{\partial s'_{\Lambda_1} \cdots \partial s'_{\Lambda_r}}K_{\{h_\Lambda\}_{\Lambda \in \mathcal{S}}}\big(\{ s'_\Lambda \}_{\Lambda \in \mathcal{S}}\big)\,,\qquad\text{where }s'_\Lambda=s_\Lambda(s)\,.
        \end{equation*}
        Evaluating this expression exactly at $s=0$ forces all auxiliary parameters $s'_{\Lambda}$ to zero. By Lemma \ref{lem:univariate_cumulant_identity}, the left side is $\kappa_r(\widetilde{U}_{n,k})$. The $r^{\text{th}}$-order partial derivatives on the right, evaluated at the origin, are by definition the joint cumulants of the local operators $h_{\Lambda_i}\,,\, i\in[r]$. Thus,
        \begin{equation}
            \kappa_r(\widetilde{U}_{n,k}) = \sum_{\Lambda_1, \ldots, \Lambda_r \in \mathcal{S}} \kappa\left( h_{\Lambda_1}, \ldots, h_{\Lambda_r} \right)\,.
            \label{eq:unconstrained_cumulant_sum}
        \end{equation}
    
        \item \textit{Vanishing of Disconnected Contributions:} Fix an ordered tuple $\boldsymbol{\Lambda} = (\Lambda_1,\ldots,\Lambda_r)$ and suppose that its intersection graph $G(\boldsymbol{\Lambda})$ is disconnected. Let $C_1,\ldots,C_q$ be its connected components, where $q\geq2$. Thus, $[r] = C_1\sqcup\cdots\sqcup C_q$. By the definition of the intersection graph, if $i\in C_a$ and $j\in C_b$ with $a\neq b$, then $\Lambda_i\cap\Lambda_j=\varnothing\,.$
    
        For each connected component $C_a$, define its physical support by $R_a := \bigcup_{i\in C_a}\Lambda_i$. We claim that $R_a\cap R_b=\varnothing$ for all $a\neq b$. Indeed, if there were an element $\ell\in R_a\cap R_b$, there would exist $i\in C_a$ and $j\in C_b$ such that $\ell\in\Lambda_i\cap\Lambda_j$. This implies $i$ and $j$ are adjacent in $G(\boldsymbol{\Lambda})$, contradicting that they belong to distinct connected components. Hence the physical supports $R_1,\ldots,R_q$ are pairwise disjoint.
        
        Now, we define the multivariate, unioperator MGF for this specific tuple as
        \begin{equation*}
            M_{\boldsymbol{\Lambda}}(t_1,\ldots,t_r) := \operatorname{Tr} \left[ \rho^{\otimes n} \exp\left( \sum_{i=1}^r t_i h_{\Lambda_i} \right) \right]\,.
        \end{equation*}
        For each component $C_a$, define $H_a := \sum_{i\in C_a}t_i h_{\Lambda_i}$. Since every $h_{\Lambda_i}$ with $i\in C_a$ is supported on $\Lambda_i\subseteq R_a$, the operator $H_a$ is supported entirely on $R_a$. Since $R_1,\ldots,R_q$ are pairwise disjoint, $[H_a,H_b]=0$ for all $a\neq b$. Therefore, $$\exp\left( \sum_{i=1}^r t_i h_{\Lambda_i} \right) = \prod_{a=1}^q\exp(H_a)\,.$$
        
        Since the operators $\exp(H_a)$ have pairwise disjoint supports, and $\rho^{\otimes n}$ is a product state, the multivariate, unioperator MGF $M_{\boldsymbol{\Lambda}}(t_1,\ldots,t_r)$ factorizes as
        \begin{align*}
            &M_{\boldsymbol{\Lambda}}(t_1,\ldots,t_r) = \prod_{a=1}^q M_{C_a}(\boldsymbol{t}_{C_a})\,,\\
            \text{where, }\qquad&M_{C_a}(\boldsymbol{t}_{C_a}) := \operatorname{Tr}\left[ \rho^{\otimes n} \exp\left( \sum_{i\in C_a}t_i h_{\Lambda_i} \right) \right]\,,\quad\text{ and }\quad \boldsymbol{t}_{C_a} := (t_i)_{i\in C_a}\,.
        \end{align*}
        
        Taking the logarithm transforms this product into an additive sum as
        \begin{equation*}
            K_{\boldsymbol{\Lambda}}(t_1,\ldots,t_r) = \sum_{a=1}^q \log M_{C_a}(\boldsymbol{t}_{C_a})\,.
        \end{equation*}
        We now take an $r^{\text{th}}$-order partial derivative with respect to all the $r$ variables. Consider one term in the sum, say $\log M_{C_a}(\boldsymbol{t}_{C_a})$. Since $C_a$ is a proper subset of $[r]$ (because $q \ge 2$), there exists at least one index $j\in[r]\setminus C_a$. This function does not depend on $t_j$, meaning $\frac{\partial}{\partial t_j} \log M_{C_a}(\boldsymbol{t}_{C_a}) = 0$. Consequently, we get a vanishing $r^{\text{th}}$-order partial derivative over all terms as
        \begin{equation*}
            \frac{\partial^r}{\partial t_1\cdots\partial t_r} K_{\boldsymbol{\Lambda}}(t_1,\ldots,t_r) = 0\,.
        \end{equation*}
        Evaluating at the origin proves that $\kappa(h_{\Lambda_1},\ldots,h_{\Lambda_r}) = 0$ for any disconnected tuple. Therefore, removing these vanishing terms from Eq.~\eqref{eq:unconstrained_cumulant_sum} restricts the summation to only those subset-tuples whose intersection graph is connected.
    \end{enumerate}
\end{proof}

Now, to explicitly bound the connected cumulant sum of Lemma \ref{lem:connected_cumulant_expansion}, while retaining the radius of convergence required for obtaining exponential tail bounds, we employ the Penrose Tree-Graph identity \cite{penrose1967convergence, Bissacotetal2010}, which cleanly replaces an unstructured sum with a structured summation over spanning trees.

\begin{lemma}[Spanning-Tree Upper Bound on Cumulants]\label{lem:direct_cumulant_bound}
    For every $r \in\mathbb{Z}_{\geq 2}$, the absolute value of the $r^{\text{th}}$-order cumulant of the centered quantum U-statistic is upper bounded as
    \begin{equation*}
        \big| \kappa_r\big(\widetilde{U}_{n,k}\big) \big| \le  \frac{2^{2r-1} r^{r-2} k^{2r-2} {\big\|O_k^{\mathrm{sym}}\big\|}_\infty^r}{n^{r-1}}\,.
    \end{equation*}
\end{lemma}

\begin{proof}
    By the fundamental definition of the M\"obius function on the partition lattice, the joint cumulant can be written as a sum over all connected graphs $G$ on the vertex set $V = \{1, \dots, r\}$. We establish the bound in a sequence of structural steps mapping these graphs to spanning trees.
    
    \begin{enumerate}
        \item \textit{The Boolean Lattice and Graph-Cumulants:} Every graph $G \subseteq E$ (where $E$ is the set of edges of the complete graph $K_r$) induces a partition $\pi(G) \in \mathcal{P}([r])$ formed by its connected components. We define the graph-dependent moment $m(G)$ identically to the partition moment as
        \begin{equation*}
            m(G) := m_{\pi(G)} = \prod_{B \in \pi(G)} \operatorname{Tr}\left[ \rho^{\otimes n} \prod_{i \in B} h_{\Lambda_i} \right]\,.
        \end{equation*}
        
        Since many different subgraphs yield the same partition (e.g., a triangle and a simple 3-vertex path both form a single connected block), we cannot directly evaluate the physical cumulant on the graph lattice. Instead, we introduce a purely edge-dependent algebraic weight, which we term the \textit{graph-cumulant} $\kappa(G)$. We formally \textit{define} $\kappa(G)$ by treating the known physical moments $m(G)$ as constants and setting up a lower-triangular system of linear equations on the Boolean lattice. Specifically, $\kappa(G)$ is implicitly defined as the unique sequence of unknowns that satisfies the exact inclusion sum given by
        \begin{equation}
            m(G) = \sum_{G' \subseteq G} \kappa(G')\,.
            \label{eq:boolean_moment_cumulant}
        \end{equation}
        Because the subgraph inclusion relation imposes a partial order graded by the number of edges, Eq.~\eqref{eq:boolean_moment_cumulant} establishes a lower-triangular linear system with a unit diagonal. This renders the system non-singular. Consequently, the sequence $\kappa(G)$ can be computed recursively via forward substitution using only the moments of strictly smaller subgraphs, guaranteeing that exactly one unique graph-cumulant exists for every possible graph.
        
        \item \textit{Boolean M\"obius Inversion and the Cut Operator:} Eq.~\eqref{eq:boolean_moment_cumulant} defines $m(G)$ as a sum over the Boolean lattice (the partially ordered set of subsets under inclusion). The general M\"obius inversion formula on a poset states that if $f(x) = \sum_{y \le x} g(y)$, then $g(x) = \sum_{y \le x} \mu(y, x) f(y)$. Thus, we can isolate $\kappa(G)$ as
        \begin{equation*}
            \kappa(G) = \sum_{G' \subseteq G} \mu(G', G) m(G')\,.
        \end{equation*}
        
        For the Boolean lattice, the interval $[G', G]$ is structurally isomorphic to the power set of the difference $G \setminus G'$. The M\"obius function $\mu(G', G)$ depends only on the size of this difference, $k = |G \setminus G'|$. By the fundamental zero-sum constraint of the M\"obius function, for any $k > 0$, the sum of $\mu$ over all intermediate subgraphs $H$ (where $G' \subseteq H \subseteq G$) must be zero. Grouping the subgraphs by their size relative to $G'$, this constraint reduces to the standard binomial identity $\sum_{j=0}^{k} \binom{k}{j} \mu_j = 0,$ where $\mu_0 = 1$. The unique algebraic solution to this recurrence is the alternating sequence $\mu_k = (-1)^k$. Substituting $k = |G \setminus G'|$, we obtain the explicit Boolean M\"obius function $\mu(G', G) = (-1)^{|G \setminus G'|}$. This yields the explicit alternating sum for the graph-cumulant as
        \begin{equation}
            \kappa(G) = \sum_{G' \subseteq G} (-1)^{|G \setminus G'|} m(G')\,.
            \label{eq:boolean_alternating_sum}
        \end{equation}
        
        We can express this alternating sum combinatorially using discrete edge-cut operators. For any edge $e$, define the forward difference operator
        \begin{equation*}
            \Delta_e f(G) := f(G) - f(G \setminus \{e\})\,.
        \end{equation*}
        Notice that this operator can be written as $\Delta_e = I - S_e$, where $I$ is the identity operator ($I f(G) = f(G)$) and $S_e$ is the edge-removal operator ($S_e f(G) = f(G \setminus \{e\})$). Since removing edge $e_1$ and then edge $e_2$ is identical to removing $e_2$ and then $e_1$, these operators commute. This allows us to expand the compound product as
        \begin{equation*}
            \left( \prod_{e \in G} \Delta_e \right) m(G) = \left( \prod_{e \in G} (I - S_e) \right) m(G)\,.
        \end{equation*}
        Expanding this product is exactly analogous to expanding a binomial product $\prod (1 - x_i)$. The expansion generates a sum over all possible subsets of edges $D \subseteq G$ to be removed. Choosing the operator $S_e$ for the edges in $D$ (and the identity $I$ for the rest) contributes a sign of $(-1)^{|D|}$ and results in the subgraph $G \setminus D$. Thus, the expansion yields
        \begin{equation}
            \left( \prod_{e \in G} \Delta_e \right) m(G) = \sum_{D \subseteq G} (-1)^{|D|} m(G \setminus D)\,.\label{eq:cut_moment}
        \end{equation}
        Setting $G' = G \setminus D$, we substitute $|D| = |G \setminus G'|$, rendering this exactly equal to the Boolean inversion in Eq.~\eqref{eq:boolean_alternating_sum}. Thus, the graph-cumulant is precisely the fully-cut moment
        \begin{equation}
            \kappa(G) = \left( \prod_{e \in G} \Delta_e \right) m(G)\,.
            \label{eq:graph_cumulant_operator}
        \end{equation}
    
        \item \textit{Projecting Graph-Cumulants to Partition-Cumulants:} Evaluating Eq.~\eqref{eq:boolean_moment_cumulant} for the single connected block $\hat{1} = [r]$ corresponding to the complete graph $K_r$ gives $m_{\hat{1}}  = m(K_r) = \sum_{G \subseteq K_r} \kappa(G)$. Grouping this sum by the connected components each graph forms yields
        \begin{equation*}
            m_{\hat{1}} = \sum_{\pi \in \mathcal{P}([r])} \left( \sum_{G : \pi(G) = \pi} \kappa(G) \right)\,.
        \end{equation*}
        By structural uniqueness on the partition poset, the physical partition cumulant $\kappa_\pi$ matches the inner sum identically. For the fully connected physical cumulant $\kappa_{\hat{1}}$, the sum restricts strictly to the set of connected graphs $\mathcal{C}_r$, i.e.,
        \begin{equation*}
            \kappa(h_{\Lambda_1}, \dots, h_{\Lambda_r}) \equiv \kappa_{\hat{1}} = \sum_{G \in \mathcal{C}_r} \kappa(G)\,.
        \end{equation*}
    
        \item \textit{The Penrose Tree-Graph Collapse:} To reduce the sum from all connected graphs $\mathcal{C}_r$ to just the spanning trees $\mathcal{T}_r$, we use the Penrose tree-routing map \cite{penrose1967convergence, Bissacotetal2010}. Fix a global lexicographic ordering on the edges of $K_r$. For any connected graph $G \in \mathcal{C}_r$, let $T(G) \in \mathcal{T}_r$ be its unique lexicographically first minimal spanning tree. 
    
        O. Penrose proved that this mapping partitions the set of all connected graphs into disjoint Boolean intervals \cite{penrose1967convergence, Bissacotetal2010}. The preimage of any spanning tree $T$ is the interval $[T, T \cup E_T]$, where $E_T \subseteq E \setminus T$ is a specific set of non-tree edges that do not violate the minimality of $T$. We rewrite the sum over these intervals and substitute Eq.~\eqref{eq:graph_cumulant_operator} to get
        \begin{align}
            \kappa_{\hat{1}} &= \sum_{T \in \mathcal{T}_r} \sum_{H \subseteq E_T} \kappa(T \cup H) \nonumber \\
            &= \sum_{T \in \mathcal{T}_r} \sum_{H \subseteq E_T} \left( \prod_{e \in T \cup H} \Delta_e \right) m(T \cup H) \nonumber \\
            &= \sum_{T \in \mathcal{T}_r} \left( \prod_{e \in T} \Delta_e \right) \left[ \sum_{H \subseteq E_T} \left( \prod_{h \in H} \Delta_h \right) m(T \cup H) \right]\,.\label{eq:conn_phys_cum}
        \end{align}
    
        \item \textit{Boolean Telescoping:} The term inside the squared parentheses in Eq.~\eqref{eq:conn_phys_cum} is a classic Boolean telescoping sum. Let $f(H) := m(T \cup H)$ and $g(H) := ( \prod_{h \in H} \Delta_h ) f(H)$. As shown in Step 2, combining Eq.~\eqref{eq:boolean_moment_cumulant} and \eqref{eq:graph_cumulant_operator} yields
        \begin{align}
            m(T \cup E_T) &= \sum_{X \subseteq T \cup E_T} \left( \prod_{x \in X} \Delta_h \right) m(X)\notag\\
            &\overset{(a)}{=} \sum_{\substack{H \subseteq E_T \\ J \subseteq T}} \left( \prod_{h \in H} \Delta_h \cdot  \prod_{j \in J} \Delta_j \right) m(H\cup J)\notag\\
            &= \sum_{H \subseteq E_T} \left( \prod_{h \in H} \Delta_h \right)\sum_{ J \subseteq T}\left(  \prod_{j \in J} \Delta_j \right) m(H\cup J)\notag\\
            &= \sum_{H \subseteq E_T} \left( \prod_{h \in H} \Delta_h \right) \sum_{ J \subseteq T}\left(  \prod_{j \in J} \Delta_j \right) \sum_{ B \subseteq H \cup J}\left(  \prod_{b \in B} \Delta_b \right) m(B)\notag\\
             &= \sum_{H \subseteq E_T} \left( \prod_{h \in H} \Delta_h \right) \sum_{ C \subseteq T \cup H}\left(  \prod_{c\in C} \Delta_j \right)m(C)\notag\\
             &= \sum_{H \subseteq E_T} \left( \prod_{h \in H} \Delta_h \right) m(T \cup H)\notag\,,
        \end{align}
        where equation $(a)$ above follows from commutativity of the operators $\{\Delta_e\}_{e\in E}\,.$
        Since $T \cup E_T$ contains the spanning tree $T$, it forms the single connected block $\hat{1} = [r]\,.$ Thus, its moment $m(T\cup E_T)$ is the true physical joint expectation $m_{\hat{1}} = \operatorname{Tr}[ \rho^{\otimes n} \prod_{i=1}^r h_{\Lambda_i} ]\,.$ The algebraic collapse replaces the massive graph sum with an exact spanning-tree identity as
        \begin{equation*}
            \kappa(h_{\Lambda_1}, \dots, h_{\Lambda_r}) = \sum_{T \in \mathcal{T}_r} \left( \prod_{e \in T} \Delta_e \right) \operatorname{Tr}\left[ \rho^{\otimes n} \prod_{i=1}^r h_{\Lambda_i} \right]\,.
        \end{equation*}
    
        \item \textit{The Global Spanning Tree Bound:} To bound the compound operator $\prod_{e \in T} \Delta_e$ acting on the fully connected tree moment $m(T)$, we use Eq.~\eqref{eq:cut_moment} to yield
        \begin{equation*}
            \left( \prod_{e \in T} \Delta_e \right) m(T) = \sum_{D \subseteq T} (-1)^{|D|} m(T \setminus D)\,.
        \end{equation*}
        By the triangle inequality, the magnitude is bounded by $\sum_{D \subseteq T} | m(T \setminus D) |$. Removing edges splits the tree into a forest with partition $\pi(T \setminus D)$. The moment factorizes as $m(T \setminus D) = \prod_{B \in \pi(T \setminus D)} \operatorname{Tr}\big[ \rho^{\otimes n} \prod_{i \in B} h_{\Lambda_i} \big]\,.$ 
        
        Since $\big|\mathrm{Tr}[\rho^{\otimes n} X]\big| \le {\|X\|}_\infty$ and the norm is submultiplicative (${\|XY\|}_\infty \le {\|X\|}_\infty {\|Y\|}_\infty$), each factorized block moment is bounded by the product of the individual operator norms within that block. Multiplying these bounds across all blocks in the partition perfectly recombines the individual norms as $| m(T \setminus D) | \le \prod_{i=1}^r {\| h_{\Lambda_i} \|}_\infty\,.$
        Crucially, this uniform upper bound holds for \textit{every} cut subset $D$. Since a spanning tree on $r$ vertices possesses $r-1$ edges, there are $2^{r-1}$ such subsets in the sum, thereby yielding the global bound
        \begin{equation}
            \left| \left( \prod_{e \in T} \Delta_e \right) m(T) \right| \le 2^{r-1} \prod_{i=1}^r {\| h_{\Lambda_i} \|}_\infty\,.\label{eq:uniform_tree_bound_1}
        \end{equation}
        Since $T$ is a completely connected graph, its moment $m(T) = m_{\hat{1}} = \operatorname{Tr}\big[ \rho^{\otimes n} \prod_{i=1}^r h_{\Lambda_i}\big]$. Therefore, combining this relation with Eq.~\eqref{eq:uniform_tree_bound_1} yields the following uniform bound for any tree $T \in \mathcal{T}_r$
        \begin{equation}
            \left| \left( \prod_{e \in T} \Delta_e \right) \operatorname{Tr}\left[ \rho^{\otimes n} \prod_{i=1}^r h_{\Lambda_i}\right]\right| \le 2^{r-1} \prod_{i=1}^r {\| h_{\Lambda_i} \|}_\infty\,.\label{eq:uniform_tree_bound_2}
        \end{equation}
    
        \item \textit{Spatial Overlap and the Intersection Graph Constraint:} Substituting this tree identity into the multilinear expansion of Lemma \ref{lem:connected_cumulant_expansion}, we get
        \begin{equation*}
            \kappa_r(\widetilde{U}_{n,k}) = \sum_{\substack{\Lambda_1,\ldots,\Lambda_r\in\binom{[n]}{k}\\ G(\boldsymbol{\Lambda})\ {\rm connected}}} \sum_{T \in \mathcal{T}_r} \left( \prod_{e \in T} \Delta_e \right) \operatorname{Tr}\left[ \rho^{\otimes n} \prod_{i=1}^r h_{\Lambda_i} \right]\,.
        \end{equation*}
        Consider an edge $e = (i,j)$ in a tree $T$. If $\Lambda_i$ and $\Lambda_j$ are disjoint ($\Lambda_i \cap \Lambda_j = \emptyset$), the operators act on distinct tensor factors. Since $\rho^{\otimes n}$ is a pure tensor product state, the expectation fully factorizes. When the cut operator $\Delta_{ij} = I - S_{ij}$ acts on this pre-factorized state, the forward difference evaluates identically to zero ($\Delta_{ij} \langle h_i h_j \rangle = \langle h_i h_j \rangle - \langle h_i \rangle \langle h_j \rangle = 0$). 
        
        If even one edge in $T$ lacks spatial overlap, the entire compound tree operator vanishes. Thus, for a term to survive, $T$ must be a strict subgraph of the intersection graph $G(\boldsymbol{\Lambda})$. Taking absolute values and applying the uniform tree bound obtained in Eq.~\eqref{eq:uniform_tree_bound_2} in Step 6 isolates the connected sum as
        \begin{align}
            \big| \kappa_r\big(\widetilde{U}_{n,k}\big) \big| &\le 2^{r-1} \sum_{T \in \mathcal{T}_r} \sum_{\substack{\Lambda_1, \dots, \Lambda_r \in\binom{[n]}{k}\\ T \subseteq G(\mathbf{\Lambda})}} \prod_{i=1}^r {\| h_{\Lambda_i} \|}_\infty\notag\\
            &= 2^{r-1} \sum_{T \in \mathcal{T}_r} \sum_{\substack{\Lambda_1, \dots, \Lambda_r \in\binom{[n]}{k}\\ \Lambda_{i}\cap\Lambda_{j}\ne\emptyset\,,\,\forall(i,j)\in E(T)}} \prod_{i=1}^r {\| h_{\Lambda_i} \|}_\infty\,.
            \label{eq:tree_sum_isolated}
        \end{align}
    
        \item \textit{Combinatorial Pruning and Final Algebraic Simplification:} Using the triangle inequality, ${\| h_{\Lambda_i} \|}_\infty \le 2 {\big\|O_k^{\mathrm{sym}}\big\|}_\infty \big/ \binom{n}{k}\,.$ Factoring this uniform constant out of the sums yields
        \begin{equation}
            \big| \kappa_r\big(\widetilde{U}_{n,k}\big) \big| \le 2^{r-1} \left( \frac{2{\big\|O_k^{\mathrm{sym}}\big\|}_\infty}{\binom{n}{k}} \right)^r \sum_{T \in \mathcal{T}_r} \sum_{\substack{\Lambda_1, \dots, \Lambda_r \\ \Lambda_{i}\cap\Lambda_{j}\ne\emptyset\,,\,\forall(i,j)\in E(T)}} 1\,.
            \label{eq:factored_cumulant_bound}
        \end{equation}
        
        We note that the double summation term in Eq.~\eqref{eq:factored_cumulant_bound} matches exactly with Eq.~\eqref{eq:double_sum_inequality} in the proof of Lemma \ref{lemma:spanning_tree_connected_moment}. Therefore, using the same argument in the proof of Lemma \ref{lemma:spanning_tree_connected_moment}, we have
        \begin{align}
            \big| \kappa_r\big(\widetilde{U}_{n,k}\big) \big| &\le 2^{2r-1} r^{r-2} {\big\|O_k^{\mathrm{sym}}\big\|}_\infty^r \left( \frac{k^2}{n} \right)^{r-1} \nonumber \\
            &= \frac{2^{2r-1} r^{r-2} k^{2r-2} {\big\|O_k^{\mathrm{sym}}\big\|}_\infty^r}{n^{r-1}}\,.\notag
        \end{align}
    \end{enumerate}
\end{proof}

\subsection{Strict Convexity of Normalized CGF $\Phi_n(t)$}

We now prove that for a non-degenerate centred quantum U-statistic, the second-order derivative of its normalized CGF (as given by Definition~\eqref{eq:normalized_CGF}) is strictly positive. This in turn implies that the normalized CGF is \textit{strictly} convex for all positive real values of $t$.

\begin{lemma}\label{lemma:CGF_second_derivative}
    The normalized CGF of the centred, non-degenerate quantum U-statistic $\widetilde{U}_{n,k}$, as given by Definition~\eqref{eq:normalized_CGF}, is strictly convex in the parameter $t$.
\end{lemma}

\begin{proof}
    Let us construct a specific sesquilinear form for arbitrary operators $X$ and $Y$ as 
    \begin{equation}\label{eq:psd_inner_product}
        \langle X, Y \rangle_{\rho, t}:= \mathrm{Tr}\left[\rho^{\otimes n} X e^{nt\widetilde{U}_{n,k}} Y^\dagger \right]\,,\quad \forall \rho\in\mathcal{D}(\mathcal{H}), \forall t\in\mathbb{R}_{>0}\,.
    \end{equation}
    Before using this form, let us first verify that this defines a valid positive semi-definite inner product operator on the space of matrices. Obviously, the form is linear in its first argument, conjugate-linear in its second argument, and is symmetric in its arguments up to conjugacy (i.e., swapping the arguments return the complex conjugate of the original form). To verify positive semi-definiteness, we rewrite the trace by factoring $\rho^{\otimes n}$ and $e^{nt\widetilde{U}_{n,k}}$ into their respective square roots (as both are positive semi-definite operators) to get
    $$\langle X, X \rangle_{\rho, t} = \mathrm{Tr}\left[\sqrt{\rho^{\otimes n}} \sqrt{\rho^{\otimes n}} X \sqrt{e^{nt\widetilde{U}_{n,k}}} \sqrt{e^{nt\widetilde{U}_{n,k}}} X^\dagger \right]\,.$$
    Using the cyclic property of the trace, this becomes
    $$\langle X, X \rangle_{\rho, t} = \mathrm{Tr}\left[ \left(\sqrt{\rho^{\otimes n}} X \sqrt{e^{nt\widetilde{U}_{n,k}}}\right) \left(\sqrt{\rho^{\otimes n}} X \sqrt{e^{nt\widetilde{U}_{n,k}}}\right)^\dagger \right]\geq 0\,.$$
    
    Now, since Definition~\eqref{eq:psd_inner_product} defines a valid PSD inner product operator on the space of matrices, it also obeys the Cauchy-Schwarz inequality, i.e., 
    \begin{equation}\label{eq:psd_inner_product_cauchy_schwarz}
        \left\vert{}\left\langle X, Y \right\rangle_{\rho, t}\right\vert{}^2 \le \langle X, X \rangle_{\rho, t} \langle Y, Y \rangle_{\rho, t}\,.
    \end{equation}
    Let $X = \widetilde{U}_{n,k}\,,$ and $Y = \mathbb{I}$. Because $\widetilde{U}_{n,k}$ is a Hermitian observable that commutes with $e^{nt\widetilde{U}_{n,k}}$, we get
    \begin{enumerate}
        \item $\left\langle \widetilde{U}_{n,k}, \widetilde{U}_{n,k} \right\rangle_{\rho, t} = \mathrm{Tr}\left[\rho^{\otimes n} \widetilde{U}_{n,k}^2 e^{nt\widetilde{U}_{n,k}}\right]\,,$ and
        \item $\left\langle \mathbb{I}, \mathbb{I} \right\rangle_{\rho, t} = \mathrm{Tr}\left[\rho^{\otimes n} e^{nt\widetilde{U}_{n,k}}\right]\,,$ and
        \item $\left\langle \widetilde{U}_{n,k}, \mathbb{I} \right\rangle_{\rho, t} = \mathrm{Tr}\left[\rho^{\otimes n} \widetilde{U}_{n,k} e^{nt\widetilde{U}_{n,k}}\right]\,.$
    \end{enumerate}
    Substituting these into Ineq.~\eqref{eq:psd_inner_product_cauchy_schwarz} gives
    $$\left( \mathrm{Tr}\left[\rho^{\otimes n} \widetilde{U}_{n,k} e^{nt\widetilde{U}_{n,k}}\right] \right)^2 \le \mathrm{Tr}\left[\rho^{\otimes n} \widetilde{U}_{n,k}^2 e^{nt\widetilde{U}_{n,k}}\right] \mathrm{Tr}\left[\rho^{\otimes n} e^{nt\widetilde{U}_{n,k}}\right]\,.$$
    Dividing both sides by $\left(\mathrm{Tr}\left[\rho^{\otimes n} e^{nt\widetilde{U}_{n,k}}\right]\right)^2$, we get
    \begin{align*}
        &\qquad \frac{\mathrm{Tr}\left[\rho^{\otimes n} \widetilde{U}_{n,k}^2 e^{nt\widetilde{U}_{n,k}}\right]}{\mathrm{Tr}\left[\rho^{\otimes n} e^{nt\widetilde{U}_{n,k}}\right]} \ge \left(\frac{\mathrm{Tr}\left[\rho^{\otimes n} \widetilde{U}_{n,k} e^{nt\widetilde{U}_{n,k}}\right]}{\mathrm{Tr}\left[\rho^{\otimes n} e^{nt\widetilde{U}_{n,k}}\right]}\right)^2\,,\\
        \text{or, }&\qquad \frac{\mathrm{Tr}\left[\rho^{\otimes n} \widetilde{U}_{n,k}^2 e^{nt\widetilde{U}_{n,k}}\right]}{\mathrm{Tr}\left[\rho^{\otimes n} e^{nt\widetilde{U}_{n,k}}\right]} - \left(\frac{\mathrm{Tr}\left[\rho^{\otimes n} \widetilde{U}_{n,k} e^{nt\widetilde{U}_{n,k}}\right]}{\mathrm{Tr}\left[\rho^{\otimes n} e^{nt\widetilde{U}_{n,k}}\right]}\right)^2\ge 0\,,\\
        \text{or, }&\qquad\frac{d^2}{dt^2}\Phi_n(t)\ge 0\,,
    \end{align*}
    
    where the last step follows from Lemma \ref{lemma:first_second_scaled_cumulants} (2) computing the explicit form of $d^2 \Phi_n(t)\big/dt^2$, and the fact that $n\in\mathbb{Z}_{>0}$ (as the number of identical copies of $\rho$ required by the estimation problem is strictly positive). This proves that the second-order derivative of the normalized CGF $\Phi_n(t)$ with respect to the parameter $t$ is non-negative, or that the normalized CGF is convex. We now prove that it is \textit{strictly} convex.
    
    Observe that Ineq.~\eqref{eq:psd_inner_product_cauchy_schwarz} becomes an exact equality if and only if $X$ and $Y$ are linearly dependent. Since $X = \widetilde{U}_{n,k}$ and $Y = \mathbb{I}$, the only way for the second-order derivative $d^2\Phi_n(t)\big/dt^2$ to be exactly equal to zero is if $\widetilde{U}_{n,k} = c \mathbb{I}\,,$
    for some scalar $c \in \mathbb{C}$. Unless the centred quantum U-statistic $\widetilde{U}_{n,k}$ is a degenerate, completely deterministic constant operator (which carries no statistical information and would never be used for estimation), it is not a scalar multiple of the identity. Therefore, $X$ and $Y$ are linearly independent, meaning that Ineq.~\eqref{eq:psd_inner_product_cauchy_schwarz} is a strict inequality.
\end{proof}

\subsection{Derivatives of the Normalized CGF $\Phi_n(t)$ and Statistics of the Shifted Distributions $\mathcal{Q}^t_{n,k}$}

We now connect the derivatives of the normalized CGF $\Phi_n(t)$ of a centred quantum U-statistic $\widetilde{U}_{n,k}$ to the mean and variance of the boosted measure $\mathcal{Q}^t_{n,k}$ as given by Definition~\eqref{eq:boosted_probability}.

\begin{lemma}\label{lemma:first_second_scaled_cumulants}
    For the normalized CGF $\Phi_n(t)$ of a centred quantum U-statistic $\widetilde{U}_{n,k}$ as given by Definition~\eqref{eq:normalized_CGF}, and the boosted (tilted) measure $\mathcal{Q}^t_{n,k}$ as given by Definition~\eqref{eq:boosted_probability}, we have 
    \begin{enumerate}
        \item $\dfrac{d}{dt}\Phi_n(t)=\dfrac{\mathbb{E}_{\rho^{\otimes n}}\left[\widetilde{U}_{n,k}e^{nt\widetilde{U}_{n,k}}\right]}{e^{n\Phi_n(t)}}=\mathbb{E}_{X\sim\mathcal{Q}^t_{n,k}}[X]\,,$ and
        \item $\dfrac{d^2}{dt^2}\Phi_n(t)=n \left( \dfrac{\mathbb{E}_{\rho^{\otimes n}}\left[\widetilde{U}_{n,k}^2 e^{nt\widetilde{U}_{n,k}}\right]}{\mathbb{E}_{\rho^{\otimes n}}\left[e^{nt\widetilde{U}_{n,k}}\right]} - \left(\dfrac{\mathbb{E}_{\rho^{\otimes n}}\left[\widetilde{U}_{n,k} e^{nt\widetilde{U}_{n,k}}\right]}{\mathbb{E}_{\rho^{\otimes n}}\left[e^{nt\widetilde{U}_{n,k}}\right]}\right)^2 \right)=n\mathrm{Var}_{X\sim\mathcal{Q}^t_{n,k}}(X)\,.$
    \end{enumerate}
\end{lemma}

\begin{proof}
    \begin{enumerate}
        \item Firstly, we have
            \begin{align*}
                \qquad\frac{d}{dt}\Phi_n(t)&=\frac{d}{dt}\left(\frac{1}{n}\log \mathbb{E}_{\rho^{\otimes n}}\left[e^{nt\widetilde{U}_{n,k}}\right]\right)\\
                &=\frac{1}{n}\cdot\frac{d}{dt}\left(\log \mathrm{Tr}\left[\rho^{\otimes n}e^{nt\widetilde{U}_{n,k}}\right]\right)\\
                &=\frac{1}{n}\cdot \frac{\frac{d}{dt} \mathrm{Tr}\left[\rho^{\otimes n}e^{nt\widetilde{U}_{n,k}}\right]}{\mathrm{Tr}\left[\rho^{\otimes n}e^{nt\widetilde{U}_{n,k}}\right]}\\
                &=\frac{1}{n} \cdot\frac{\mathrm{Tr}\left[\rho^{\otimes n} \left(n\widetilde{U}_{n,k} e^{nt\widetilde{U}_{n,k}}\right)\right]}{\mathrm{Tr}\left[\rho^{\otimes n}e^{nt\widetilde{U}_{n,k}}\right]}\\
                &=\frac{\mathrm{Tr}\left[\rho^{\otimes n}\widetilde{U}_{n,k} e^{nt\widetilde{U}_{n,k}}\right]}{\mathrm{Tr}\left[\rho^{\otimes n}e^{nt\widetilde{U}_{n,k}}\right]}\\
                &=\frac{\mathbb{E}_{\rho^{\otimes n}}\left[\widetilde{U}_{n,k} e^{nt\widetilde{U}_{n,k}}\right]}{\mathbb{E}_{\rho^{\otimes n}}\left[e^{nt\widetilde{U}_{n,k}}\right]}\\
                &=\frac{\mathbb{E}_{\rho^{\otimes n}}\left[\widetilde{U}_{n,k} e^{nt\widetilde{U}_{n,k}}\right]}{e^{n\Phi_n(t)}}\,.
            \end{align*}
            
            Now, the expected value of a random variable $X$ distributed according to the boosted probability measure $\mathcal{Q}^t_{n,k}$ is given by the integral over its support as $$\mathbb{E}_{X\sim\mathcal{Q}^t_{n,k}}[X] = \int\limits_{x\in \mathrm{Supp}\big(\mathcal{Q}^t_{n,k}\big)} x \, \mathcal{Q}^t_{n,k}(x) dx\,,\qquad \forall t\in\mathbb{R}_{>0}\,.$$
            Substituting the definition of the boosted measure, and recalling that $\mathrm{Supp}\big(\mathcal{Q}^t_{n,k}\big) = \mathrm{Supp}\left(\mathcal{P}_{n,k}\right)$, we have
            \begin{align}
                \mathbb{E}_{X\sim\mathcal{Q}^t_{n,k}}[X] &= \int\limits_{x\in \mathrm{Supp}\left(\mathcal{P}_{n,k}\right)} x \frac{\mathcal{P}_{n,k}(x)e^{ntx}}{e^{n\Phi_n(t)}} dx\,,\qquad \forall t\in\mathbb{R}_{>0}\,,\notag\\
                &= \frac{1}{e^{n\Phi_n(t)}} \int\limits_{x\in \mathrm{Supp}\left(\mathcal{P}_{n,k}\right)} x e^{ntx} \mathcal{P}_{n,k}(x) dx\,,\qquad \forall t\in\mathbb{R}_{>0}\,.\label{eq:boosted_prob_expectation}
            \end{align}
            Because $\mathcal{P}_{n,k}$ is the measure of the outcomes $x$ obtained by measuring the product state $\rho^{\otimes n}$ in the eigenbasis of $\widetilde{U}_{n,k}$, using the Law of the Unconscious Statistician, we have that the classical integral over $x$ of the function $x e^{ntx}$ with respect to $\mathcal{P}_{n,k}(x)$ is identically equal to the quantum expectation value of the corresponding operator $\widetilde{U}_{n,k} e^{nt\widetilde{U}_{n,k}}$. Therefore, we have $$\int\limits_{x\in \mathrm{Supp}\left(\mathcal{P}_{n,k}\right)} x e^{ntx} \mathcal{P}_{n,k}(x) dx = \mathbb{E}_{\rho^{\otimes n}}\left[\widetilde{U}_{n,k} e^{nt\widetilde{U}_{n,k}}\right]\,.$$
            Furthermore, exponentiating the definition of the normalized CGF, the denominator comes out to be exactly $e^{n\Phi_n(t)} = \mathbb{E}_{\rho^{\otimes n}}\big[e^{nt\widetilde{U}_{n,k}}\big]\,.$
            Substituting the expressions for both the numerator and the denominator back into the expression for expectation in Eq.~\eqref{eq:boosted_prob_expectation}, we get
            $$\mathbb{E}_{X\sim\mathcal{Q}^t_{n,k}}[X] = \frac{\mathbb{E}_{\rho^{\otimes n}}\left[\widetilde{U}_{n,k} e^{nt\widetilde{U}_{n,k}}\right]}{\mathbb{E}_{\rho^{\otimes n}}\left[e^{nt\widetilde{U}_{n,k}}\right]}\,.$$
            \item Again, observe that 
            \begin{align*}
                \frac{d^2}{dt^2}\Phi_n(t) &= \frac{d}{dt} \left( \frac{\mathrm{Tr}\left[\rho^{\otimes n}\widetilde{U}_{n,k} e^{nt\widetilde{U}_{n,k}}\right]}{\mathrm{Tr}\left[\rho^{\otimes n}e^{nt\widetilde{U}_{n,k}}\right]} \right) \\     
                &= \dfrac{1}{\left(\mathrm{Tr}\left[\rho^{\otimes n}e^{nt\widetilde{U}_{n,k}}\right]\right)^2}\bigg[\frac{d}{dt}\left(\mathrm{Tr}\left[\rho^{\otimes n}\widetilde{U}_{n,k} e^{nt\widetilde{U}_{n,k}}\right]\right) \cdot \mathrm{Tr}\left[\rho^{\otimes n}e^{nt\widetilde{U}_{n,k}}\right] - \\ &\hspace{3cm}\mathrm{Tr}\left[\rho^{\otimes n}\widetilde{U}_{n,k} e^{nt\widetilde{U}_{n,k}}\right] \cdot \frac{d}{dt}\left(\mathrm{Tr}\left[\rho^{\otimes n}e^{nt\widetilde{U}_{n,k}}\right]\right)\bigg]\\     
                &= \dfrac{1}{\left(\mathrm{Tr}\left[\rho^{\otimes n}e^{nt\widetilde{U}_{n,k}}\right]\right)^2}\bigg[\mathrm{Tr}\left[\rho^{\otimes n}\widetilde{U}_{n,k} \left(n\widetilde{U}_{n,k} e^{nt\widetilde{U}_{n,k}}\right)\right] \cdot \mathrm{Tr}\left[\rho^{\otimes n}e^{nt\widetilde{U}_{n,k}}\right] - \\ &\hspace{3cm}\mathrm{Tr}\left[\rho^{\otimes n}\widetilde{U}_{n,k} e^{nt\widetilde{U}_{n,k}}\right] \cdot \mathrm{Tr}\left[\rho^{\otimes n}\left(n\widetilde{U}_{n,k} e^{nt\widetilde{U}_{n,k}}\right)\right]\bigg] \\     
                &= n \left( \frac{\mathrm{Tr}\left[\rho^{\otimes n}\widetilde{U}_{n,k}^2 e^{nt\widetilde{U}_{n,k}}\right]}{\mathrm{Tr}\left[\rho^{\otimes n}e^{nt\widetilde{U}_{n,k}}\right]} - \left(\frac{\mathrm{Tr}\left[\rho^{\otimes n}\widetilde{U}_{n,k} e^{nt\widetilde{U}_{n,k}}\right]}{\mathrm{Tr}\left[\rho^{\otimes n}e^{nt\widetilde{U}_{n,k}}\right]}\right)^2 \right)\\
                &=n \left( \frac{\mathbb{E}_{\rho^{\otimes n}}\left[\widetilde{U}_{n,k}^2 e^{nt\widetilde{U}_{n,k}}\right]}{\mathbb{E}_{\rho^{\otimes n}}\left[e^{nt\widetilde{U}_{n,k}}\right]} - \left(\frac{\mathbb{E}_{\rho^{\otimes n}}\left[\widetilde{U}_{n,k} e^{nt\widetilde{U}_{n,k}}\right]}{\mathbb{E}_{\rho^{\otimes n}}\left[e^{nt\widetilde{U}_{n,k}}\right]}\right)^2 \right)\,.
            \end{align*}

             Now, let us compute the second moment of $X$ under the boosted distribution $\mathcal{Q}^{t}_{n,k}$. Factoring out the denominator as you did for the first moment, we get 
             \begin{align*}
                 \mathbb{E}_{X\sim\mathcal{Q}^{t}_{n,k}}\left[X^2\right] &= \int\limits_{x\in \mathrm{Supp}\left(\mathcal{P}_{n,k}\right)} x^2 \frac{\mathcal{P}_{n,k}(x)e^{ntx}}{e^{n\Phi_n(t)}} dx \\
                 &= \frac{1}{e^{n\Phi_n(t)}} \int\limits_{x\in \mathrm{Supp}\left(\mathcal{P}_{n,k}\right)} x^2 e^{ntx} \mathcal{P}_{n,k}(x) dx
             \end{align*}
             Again, using the Law of the Unconscious Statistician, and substituting $e^{n\Phi_n(t)} = \mathbb{E}_{\rho^{\otimes n}}\left[e^{nt\widetilde{U}_{n,k}}\right]$, we get
             \begin{align*}
                 \mathbb{E}_{X\sim\mathcal{Q}^{t}_{n,k}}\left[X^2\right] &= \frac{\mathbb{E}_{\rho^{\otimes n}}\left[\widetilde{U}_{n,k}^2 e^{nt\widetilde{U}_{n,k}}\right]}{\mathbb{E}_{\rho^{\otimes n}}\left[e^{nt\widetilde{U}_{n,k}}\right]}\\
                 \text{or,}\qquad \mathrm{Var}_{X\sim\mathcal{Q}^{t}_{n,k}}(X) &= \frac{\mathbb{E}_{\rho^{\otimes n}}\left[\widetilde{U}_{n,k}^2 e^{nt\widetilde{U}_{n,k}}\right]}{\mathbb{E}_{\rho^{\otimes n}}\left[e^{nt\widetilde{U}_{n,k}}\right]} - \left(\frac{\mathbb{E}_{\rho^{\otimes n}}\left[\widetilde{U}_{n,k} e^{nt\widetilde{U}_{n,k}}\right]}{\mathbb{E}_{\rho^{\otimes n}}\left[e^{nt\widetilde{U}_{n,k}}\right]}\right)^2\,.
             \end{align*}
    \end{enumerate}
\end{proof}

\end{document}